\documentclass[pra,10pt, twocolumn, superscriptaddress, nofootinbib]{revtex4-1}

\pdfoutput=1
\usepackage{amsmath}
\usepackage{wasysym}
\usepackage{amsfonts}
\usepackage{braket}
\usepackage{tensor}
\usepackage{amssymb}
\usepackage{bm,bbm}
\usepackage{graphicx}
\usepackage{tikz}
\usepackage{color,amsxtra,amsfonts,dsfont,enumerate}
\usepackage{amsthm}
\theoremstyle{definition}
\newtheorem{definition}{Definition}[section]

\usepackage[colorlinks=true,linkcolor=blue, citecolor=blue, urlcolor=blue, bookmarks]{hyperref}
\usetikzlibrary{shapes,snakes}

\def\be{\begin{equation}}
	\def\ee{\end{equation}}
\def\bea{\begin{eqnarray}}
	\def\eea{\end{eqnarray}}
\def\gtimes{\otimes _\mathfrak{g}}

\theoremstyle{plain}

\newtheorem{thm}{Theorem}[section]
\newtheorem{prop}[thm]{Proposition}
\newtheorem{cor}[thm]{Corollary}
\newtheorem{lem}[thm]{Lemma}
\newtheorem{defn}[thm]{Definition}

\newtheorem{rmk}[thm]{Remark}

\tikzset{
	ncbar angle/.initial=90,
	ncbar/.style={
		to path=(\tikztostart)
		-- ($(\tikztostart)!#1!\pgfkeysvalueof{/tikz/ncbar angle}:(\tikztotarget)$)
		-- ($(\tikztotarget)!($(\tikztostart)!#1!\pgfkeysvalueof{/tikz/ncbar angle}:(\tikztotarget)$)!\pgfkeysvalueof{/tikz/ncbar angle}:(\tikztostart)$)
		-- (\tikztotarget)
	},
	ncbar/.default=0.5cm,
}

\tikzset{square left brace/.style={ncbar=0.5cm}}
\tikzset{square right brace/.style={ncbar=-0.5cm}}

\tikzset{round left paren/.style={ncbar=0.5cm,out=120,in=-120}}
\tikzset{round right paren/.style={ncbar=0.5cm,out=60,in=-60}}

\begin{document}
	
\title{Fermionic quantum cellular automata and generalized matrix product unitaries}
	
\author{Lorenzo Piroli}
\affiliation{Max-Planck-Institut f\"ur Quantenoptik, Hans-Kopfermann-Str.~1, 85748 Garching, Germany}
\affiliation{Munich Center for Quantum Science and Technology, Schellingstra\ss e 4, 80799 M\"unchen, Germany}	
\author{Alex Turzillo}
\affiliation{Max-Planck-Institut f\"ur Quantenoptik, Hans-Kopfermann-Str.~1, 85748 Garching, Germany}
\affiliation{Munich Center for Quantum Science and Technology, Schellingstra\ss e 4, 80799 M\"unchen, Germany}	
\author{Sujeet K. Shukla$^\ast$}
\affiliation{Department of Physics, University of Washington, Seattle WA, 98105 USA}
\author{J. Ignacio Cirac}
\affiliation{Max-Planck-Institut f\"ur Quantenoptik, Hans-Kopfermann-Str.~1, 85748 Garching, Germany}
\affiliation{Munich Center for Quantum Science and Technology, Schellingstra\ss e 4, 80799 M\"unchen, Germany}	

\begin{abstract}
We study matrix product unitary operators (MPUs) for fermionic one-dimensional ($1D$) chains. In stark contrast with the case of $1D$ qudit systems, we show that $(i)$ fermionic MPUs do not necessarily feature a strict causal cone and $(ii)$ not all fermionic Quantum Cellular Automata (QCA) can be represented as fermionic MPUs. We then introduce a natural generalization of the latter, obtained by allowing for an additional operator acting on their auxiliary space. We characterize a family of such generalized MPUs that are locality-preserving, and show that, up to appending inert ancillary fermionic degrees of freedom, any representative of this family is a fermionic QCA and viceversa. Finally, we prove an index theorem for generalized MPUs, recovering the recently derived classification of fermionic QCA in one dimension. As a technical tool for our analysis, we also introduce a graded canonical form for fermionic matrix product states, proving its uniqueness up to similarity transformations. 
\end{abstract}
	
\maketitle

	\def\thefootnote{*}\footnotetext{Present address: Amazon Research, WA, USA}\def\thefootnote{\arabic{footnote}}


\section{Introduction}

Among the achievements of Tensor Network (TN) theory, several results stand out in the context of classification of topological phases of matter, now a pillar in modern quantum many-body physics~\cite{Chiu2016Classification}. For one-dimensional $(1D)$ systems, this problem can be naturally formulated in terms of Matrix Product States (MPSs)~\cite{fannes1992finitely,perez2006matrix}, and consists, loosely speaking, in determining all the possible equivalence classes under suitably defined smooth deformations. Arguably, the case where such a problem is most well-understood is that of  $1D$ bosonic symmetry-protected topological (SPT) phases~\cite{Haldane1983Nonlinear,Affleck1987Rigorous,Gu2009Tensor,Pollmann2010Entanglement}, where they have been completely classified by the second cohomology group~\cite{Schuch2011Classifying,Chen2011Classification,Chen2011Complete} using previously-derived results on the canonical forms of MPSs~\cite{perez2006matrix,perez2008string}.

Recently, increasing attention has been devoted to the classification of topological systems far from equilibrium. This was also motivated by the experimental advances in atomic, molecular and optical physics, which now make it possible to probe the quantum many-body dynamics in exquisite detail~\cite{barreiro2011open,schindler2013quantum,choi2017observation,bernien2017probing,zhang2017observation,zhang2017observationMany}. The classification of periodically driven Floquet systems, in particular, has attracted a lot of theoretical work in the past few years~\cite{Kitagawa2010Topological,bardyn2013topology,Else2016Classification,Keyserlingk2016Phase,Potter2016Classification,Roy2017Periodic,Yao2017Topological,Gong2018Topological,Gong2018TopologicalEntanglement,Higashikawa2019Floquet,McGinley2018Topology,McGinley2019Classification,McGinley2019Interacting,Coser2019Classification}. 

In this context, a relevant problem pertains the study of Matrix Product Unitary operators (MPUs)~\cite{Chen2011Complete,Po2016Chiral,cirac2017matrix,sahinoglu2018matrix,Gong2020Classification} in one dimension, namely of Matrix Product Operators (MPOs) that are also unitary. This is intimately connected to the classification of two-dimensional ($2D$) Floquet SPT phases~\cite{harper2019topology}: indeed, given a $2D$ Floquet system which exhibits many-body localization (MBL) in the bulk, its edge dynamics is well described by an MPU~\cite{Po2016Chiral,Potter2017Dynamically,Harper2017FloquetTopological,Roy2017Floquet}.

The theory of MPUs was first developed in Refs.~\cite{cirac2017matrix,sahinoglu2018matrix} for $1D$ qudit systems (see also Ref.~\cite{piroli2020Quantum}, for a recent generalization in higher spatial dimensions). As a nontrivial result, it was shown that any MPU is in fact a Quantum Cellular Automaton (QCA) and viceversa~\cite{schumacher2004reversible,farrelly2019review,farrelly2019review}, namely MPUs feature a causal cone, propagating information strictly by a finite distance only. This observation is particularly interesting, because it allows one to address the analysis of QCA by means of the powerful tools  developed within the theory of MPSs. For example, based on the latter, it was proven that, in the absence of symmetries, equivalence classes of MPUs under smooth deformations are labeled by an index which can be computed directly from their local tensors. This index quantifies the net quantum information flow through the MPU, and was shown to be equivalent to the GNVW index (for Gross, Nesme, Vogts, and Werner) introduced in Ref.~\cite{gross2012index}. Later, the same problem was addressed in the presence of a local symmetry in Ref.~\cite{Gong2020Classification}, where it was shown that a complete classification can be obtained by also taking into account the cohomology class of the symmetry group, thus proving a conjecture raised earlier by Hastings~\cite{hastings2013Classifying}.

While so far MPUs have been studied exclusively for qudit systems, $1D$ fermionic QCA (fQCA) have been recently analyzed in Refs.~\cite{fidkowski2017Radical,fidkowski2019interacting}. In particular, it was shown that while one can develop an index theory along the lines of Ref.~\cite{gross2012index}, the fermionic index need not be a rational number, as for qudits, but can also include a factor of  square root of two. Furthermore, such a classification is complete only if we allow for a more general notion of \emph{stable equivalence} of QCA, which involves enlarging the Hilbert space by appending inert ancillary fermionic degrees of freedom (d.o.f.)~\cite{fidkowski2019interacting}. We note that the emergence of a richer picture might have been expected based on previous studies on fermionic SPT phases in one dimension~\cite{kitaev2001unpaired,Fidkowski2010Effects,Fidkowski2011Topological}. In that case, by mapping fermions onto qudits via the Jordan-Wigner (JW) transformation~\cite{wigner1928paulische}, it was shown that the fermionic classification problem can be reduced to the bosonic one, in the presence of an additional $\mathbb{Z}_2$-symmetry, corresponding to the conserved fermionic parity~\cite{Fidkowski2011Topological,Chen2011Complete}.

At this point, a fundamental question is whether fQCA are also equivalent to fermionic MPUs (fMPUs), similarly to what happens for qudits. For this problem, it is not natural to reduce ourselves to the latter case via a JW transformation, since it typically generates nonlocal terms for periodic boundary conditions (PBC). In fact, it is more convenient to work using a genuinely fermionic formalism~\cite{kraus_fermionic_2010,wahl_symmetries_2014}, where the definition of MPSs and MPOs can be generalized in a straightforward way. This is the approach that we take in this work, where we provide a thorough study of fMPUs and their connections to fQCA. Our results show that the picture is much richer with respect  to the case of qudits, as we outline in the following. 

\subsection{Summary of our  results}

\paragraph{{\bf fMPUs are not equivalent to fQCA}.} Based on Refs.~\cite{cirac2017matrix, sahinoglu2018matrix}, one could expect that any fMPU is automatically a QCA, namely it displays a strict causal cone. Our first result is to show that this is not the case. In particular, in Sec.~\ref{sec:nonlocal_fmpu} we exhibit an example of an fMPU with periodic boundary conditions which is not locality-preserving. In fact, the inverse statement is also not true, and we find that the most natural generalization of MPUs to fermionic degrees of freedom is not enough to capture all fermionic QCA. This is discussed in Sec.~\ref{sec:majorana_shift_operator}, were we built an MPO implementing a translation of Majorana modes. This is a QCA~\cite{fidkowski2019interacting}, but we show that it can not be written in the expected form.\\[-0.2cm]

\paragraph{{\bf Characterization of locality-preserving fMPUs}.} As a second main result, we introduce a class of ``generalized'' fMPUs, and identify a condition on the corresponding local tensors such that any representative of this family is a fQCA and viceversa. This is discussed in Sec.~\ref{sec:generalized_fMPUs}, where generalized fMPUs are defined by allowing for an additional operator acting on the corresponding auxiliary  space, and further characterized in Sec.~\ref{sec:gfMPUs_two_kinds}. The condition guaranteeing a strictly causal cone is expressed in Eq.~\eqref{eq:simpleness_fermionic}, and generalizes the simpleness condition introduced for qudits~\cite{cirac2017matrix}. In addition, we find that any tensor generating an fMPU with antiperiodic boundary conditions (ABC) is necessarily simple after blocking, while this is not true in the periodic case. \\[-0.2cm]

\paragraph{{\bf Index theory from fMPUs}.} Third, we show that for the class of locality-preserving fMPUs,  one can define an index based on the local tensors, which coincides with that of Ref.~\cite{fidkowski2019interacting}. Although our construction is analogous to the one carried out for qudits~\cite{cirac2017matrix,sahinoglu2018matrix}, there are some practical differences, and the definition for the index contains additional signs, as discussed in Sec.~\ref{sec:index_theory}. Here, our main result is the Index Theorem~\ref{th:femionic_index}, which states that the fermionic index displays all the expected stability properties that are present in the case of qudits, and correctly classifies fMPUs with respect to smooth deformations preserving unitarity.\\[-0.2cm]

\paragraph{{\bf  Graded canonical form for fMPSs}.} Finally, as a byproduct of our work, we introduce and study a new graded canonical form (GCF) for fMPSs, which is based on the definition of irreducible fermionic tensors recently presented in Ref.~\cite{bultinck2017fermionic}. This is detailed in Appendix~\ref{sec:graded_canonical_form}, where we prove the existence and uniqueness of the GCF in the case of antiperiodic boundary conditions. This is a technical result which allows  us to generalize directly some derivations of Ref.~\cite{cirac2017matrix}, but which is also interesting per se and might have applications in other problems.

\subsection{Structure of the paper} 

The rest of this work is organized as follows. We begin in Sec.~\ref{sec:QCA}, where we briefly recall the definition of QCA, and review some key results obtained in the recent literature. We proceed with Sec.~\ref{sec:bosonic_tn}, where we introduce basic aspects of TNs in qudit systems,  and review some of the main results derived in Ref.~\cite{cirac2017matrix} on MPUs. We move on with Sec.~\ref{sec:fermionic_MPSs}, where we introduce the standard language of fermionic TNs using the so-called fiducial state formalism. In Sec.~\ref{sec:fMPUs} we discuss our first results on fMPUs, while generalized fMPUs are finally introduced and analyzed in Sec.~\ref{sec:generalized_fMPUs}. The corresponding index theory is then developed in Sec.~\ref{sec:index_theory}, which represents the most technical part of our work, and is carried out using the formalism of graded TNs. Finally, we report our conclusions in Sec.~\ref{sec:conclusions}.

\section{QCA and index theory}
\label{sec:QCA}

Before embarking on the study of fMPUs, we briefly recall some known facts about QCA. While they are usually defined in terms of the automorphisms of the $C^\ast$ algebra of local operators in infinite systems~\cite{gross2012index}, one can also define them as unitary operators acting on finite systems, which is the point of view taken in this work. Specifically,  considering first a $1D$ qudit system associated with the Hilbert space  $\mathcal{H}=\bigotimes_{j}\mathcal{H}_j$, with $\mathcal{H}_j\simeq \mathbb{C}_j^d$, a QCA of range $r$ is a unitary operator $U$ such that for any local operator $\mathcal{O}_j$ acting non-trivially only on $\mathcal{H}_j$, the transformed operator $U^\dagger\mathcal{O}_j U$ acts non-trivially only on the qudits $k=j-r,j-r+1,\ldots j+r$. A completely analogous definition can be given for fermionic degrees of freedom.

QCA have been completely characterized for $1D$ qudit systems in Ref.~\cite{gross2012index}. There, it was shown that any QCA is obtained by composing a finite number of shift operators, and ``finite depth'' quantum circuits, which are obtained by applying a finite number of ``layers'' of two-site local unitary gates acting on disjoint sets of qudits. Moreover, QCA can be classified according to a rational index (denoted by GNVW in this work), which measures the net quantum information flow through the system. Two QCA have the same index if and only if they can be continuously deformed into one another, or, equivalently, are related by the application of a finite-depth quantum circuit. For instance, a one-site shift of qubits on a ring has GNVW index $2$, and is thus in a different class from a finite-depth quantum circuit, which has GNVW index $1$. Finally, it was shown that this index is multiplicative with respect to composition.

These results were generalized for $1D$ fermionic systems in Ref.~\cite{fidkowski2019interacting}. In these works, it was shown that one can define a similar index, although technical complications arise due to the structure of the local fermionic algebra. Physically, the main result was the discovery that the classification of fQCA is richer than in the case of qudits: indeed, the fermionic index is not necessarily a rational number, but also includes factors of square root of two. This is due to the existence of a special fQCA: the Majorana-shift operator, which consists in the translation of Majorana modes (rather than physical fermions). A physical picture for this fQCA was already given in Ref.~\cite{fidkowski2017Radical}, where it was shown that this shift can be understood as the unitary operator that exchanges the topological phases of the $1D$ complex fermion chain~\cite{kitaev2001unpaired}. For PBC, the latter have different parity, implying that the Majorana-shift fQCA with PBC necessarily does not preserve it.

The above picture provides a strong motivation for the study of fMPUs. In particular, since ``plain'' fermionic MPOs (i.e. with no additional operator acting on the auxiliary space) are by construction parity-preserving, cf. Sec.~\ref{sec:fMPOs}, one can expect that the most natural generalization of MPUs defined for qudits~\cite{cirac2017matrix} will not be adequate to capture the Majorana-shift fQCA. Furthermore, a natural question is whether the fermionic index can be extracted from the local tensors of fermionic MPOs defining the associated unitary operators, in analogy to what happens for qudit MPUs~\cite{cirac2017matrix}. These issues will be addressed in the rest of this work.

\section{Basics of Matrix Product States}
\label{sec:bosonic_tn}

In this section we recall some basic definitions and results about MPSs, following the treatment in Ref.~\cite{cirac2017matrixRenormalization}.

We consider a Hilbert space $\mathcal{H}_d$ of dimension $d$, with an orthonormal basis denoted by $\{\ket{i}\}_{i=0}^{d-1}$. We introduce the tensor $\mathcal{A}$, which is defined by its elements $A^{n}_{\alpha,\beta}$, with $n=0, \ldots d-1$, and $\alpha,\beta=0,\ldots, D-1$. We call $n$ and $\alpha$, $\beta$ the physical and bond (auxiliary) indices respectively. The tensor $\mathcal{A}$ defines a family of translation-invariant states $\ket{V^{(N)}}\in \mathcal{H}_d^{\otimes N}$
\be
\ket{V^{(N)}}=\sum_{n_1,\ldots , n_N=0}^{d-1}c_{n_1,\ldots, n_N}\ket{n_1}\otimes \cdots \otimes\ket{n_N}\,,
\label{eq:MPS}
\ee
where
\be
c_{n_1,\ldots, n_N}={\rm tr}\left[A^{n_1}\cdots A^{n_N}\right]\,,
\label{eq:trace}
\ee
and where $A^n$ denotes the $D\times D$ matrix with elements $A^{n}_{\alpha,\beta}$. We call $\ket{V^{(N)}}$ the MPS generated by the tensor $\mathcal{A}$.

MPSs admit a convenient graphical representation~\cite{cirac2017matrixRenormalization}, as we now briefly recall. First, the individual tensors are denoted by
\be
\mathcal{A}=
\begin{tikzpicture}[baseline={([yshift=-2.5ex]current bounding box.center)}, scale=0.7]

\draw (0,0)  -- (0,0.7) node[above]{};
\draw (-0.7,0) node[left]{} -- (0.7,0) node[right]{};	
\filldraw[fill=white,draw=black] (0,0) circle [radius=0.25];
\end{tikzpicture}\,,
\ee
where the horizontal (vertical) lines represent the bond (physical) indices. The MPS $\ket{V^{(N)}}$ is then represented as
\be
\ket{V^{(N)}}=\quad 
\begin{tikzpicture}[baseline={([yshift=0.5ex]current bounding box.center)}, scale=0.7]

\draw (0,0)  -- (0,0.65) ;
\draw (-0.5,0) -- (1.25,0);	
\filldraw[fill=white,draw=black] (0,0) circle [radius=0.25];
\draw (0.75,0)  -- (0.75,0.65);
\filldraw[fill=white,draw=black] (0.75,0) circle [radius=0.25];
\draw(1.75,0) node{$\cdots$};
\draw (2.75,0)  -- (2.75,0.65);
\draw (2.25,0)  -- (3.15,0);	
\filldraw[fill=white,draw=black] (2.75,0) circle [radius=0.25];
\draw (-0.5,0) arc (90:270:0.15) node[right]{};
\draw (3.15,0) arc (90:-90:0.15) node[right]{};
\draw [decorate, decoration={brace,amplitude=4pt, raise=4pt},yshift=0pt]
(3.25,-0.5) -- (-0.6,-0.5) node [black,midway,yshift=-0.55cm] {$N$};
\end{tikzpicture}\,.
\ee
Here, the lines that join different tensors indicate that the corresponding indices are contracted, while the vertical lines represent the physical indices. The curvy lines at the end indicate that the last and first tensors are also contracted, mimicking the presence of the trace in Eq.~\eqref{eq:trace}.

An important object in the theory of MPSs is the so-called transfer matrix (TM), which can be defined as
\be
E=\sum_{n=0}^{d-1} A^{n} \otimes \bar{A}^{n}\,,
\label{eq:transfer_matrix_q}
\ee
and which admits the graphical representation
\bea
E=
\begin{tikzpicture}[baseline={([yshift=-0.5ex]current bounding box.center)}, scale=0.7]

\draw (0,-0.25) -- (0,1) node[left]{};
\draw (-0.5,0) node[left]{} -- (0.5,0) node[right]{};
\draw (-0.5,1) node[left]{} -- (0.5,1) node[right]{};

\filldraw[fill=white,draw=black] (0,0) circle [radius=0.25];
\fill[black] (0,1) circle [radius=0.25];

\end{tikzpicture}\,.
\label{eq:transfer_m_q_fermions}
\eea
In Eq.~\eqref{eq:transfer_matrix_q}, $\bar{A}^n$ is the complex-conjugated matrix of $A^n$, and corresponds to a black circle in Eq.~\eqref{eq:transfer_m_q_fermions}. In the following, we will denote by $\lambda_E$ the spectral radius of $E$, i.e. its eigenvalue with largest absolute value.

A fundamental result is that MPSs of the form~\eqref{eq:trace} can be brought into a canonical form (CF)~\cite{perez2006matrix,cirac2017matrixRenormalization}, which is particularly important when comparing the MPSs generated by different tensors. We recall here the precise definition, which we will refer to later on.
\begin{defn}\label{def:CF}
	We say that a tensor $\mathcal{A}$ generating an MPS is in CF if: $(i)$ the matrices are of the form $A^n = \oplus_{k=1}^r \mu_k A^n_k$, where $\mu_k\in \mathbb{C}$ and the spectral radius of the transfer matrix, $E_k$, associated with $A^n_k$ is equal to one; $(ii)$ for all $k$, there exists no projector, $P_k$, such that $A^n_k P_k = P_k A^n_k P_k$ for all $n$.
\end{defn}
Loosely speaking, this means that the matrices $A^n$ are written in a block-diagonal form, and the blocks can not be decomposed into smaller ones. It is also useful to recall the following
\begin{defn}
	We say that a tensor $\mathcal{A}$ generating an MPS is {irreducible} if there exists no projector, $P$, such that $A^n P = P A^n P$ for all $n$. Furthermore, we say that $\mathcal{A}$  is {normal} if it irreducible and its associated transfer matrix has a unique eigenvalue of magnitude (and value) equal to its spectral radius, which is equal to one.
\end{defn}
Given two tensors $\mathcal{A}$ and $\mathcal{B}$, it is straightforward to show that they generate the same MPS if they are related to one another by a \emph{gauge} transformation, namely $B^n = XA^n X^{-1}$, for some invertible $D\times D$ matrix $X$. One can also see that for a normal tensor it is always possible to find a gauge transformation defining a new normal tensor that is in Canonical Form II (CFII).
\begin{defn}
	Let $\mathcal{A}$ be a normal tensor, and $\Phi$ and $\rho$ the left and right eigenvectors of $E$ corresponding to the eigenvalue $1$. We say that $\mathcal{A}$ is in Canonical Form II (CFII) if
\begin{subequations}
	\begin{align}
		(\Phi |&=\sum_{n=0}^{D-1}(n, n |\,,\\
		| \rho) &=\sum_{n=0}^{D-1} \rho_{n} | n, n)\,,
	\end{align}
\label{eq:cfII}
\end{subequations}
where $\rho_n>0$ and $(\Phi|\rho)=1$.
\end{defn}
Note that in Eq.~\eqref{eq:cfII} we have used round brackets, to indicate that the bra and ket states correspond to the auxiliary space. Note also that $\Phi$ and $\rho$ can be considered as $D\times D$ matrices or as vectors in $\mathcal{H}_D\otimes \mathcal{H}_D$, and by definition we have the graphical equations
\begin{align}
\begin{tikzpicture}[baseline={([yshift=-0.5ex]current bounding box.center)}, scale=0.7]
\draw (-0.5,1) arc (90:270:0.5) node[right]{};
\draw (0,-0.25) -- (0,1) node[left]{};
\draw (-0.5,0) node[left]{} -- (0.5,0) node[right]{};
\draw (-0.5,1) node[left]{} -- (0.5,1) node[right]{};
\filldraw[fill=white,draw=black] (0,0) circle [radius=0.25];
\fill[black] (0,1) circle [radius=0.25];
\draw (1.25,0.5) node[left]{$=$};
\draw (1.8,1) arc (90:270:0.5) node[right]{};
\end{tikzpicture}\,,\\
\begin{tikzpicture}[baseline={([yshift=-0.5ex]current bounding box.center)}, scale=0.7]
\draw (6.,0) arc (-90:90:0.5) node[right]{};
\draw (6,0.5) node[left]{$=$};
\draw (4.5,0) arc (-90:90:0.5) node[right]{};
\draw (4,-0.25) -- (4,1) node[left]{};
\draw (3.5,0) node[left]{} -- (4.5,0) node[right]{};
\draw (3.5,1) node[left]{} -- (4.5,1) node[right]{};
\filldraw[fill=white,draw=black] (4,0) circle [radius=0.25];
\fill[black] (4,1) circle [radius=0.25];
\fill [black] (4.75,0.3) rectangle (5.15,0.7);
\fill [black] (6.25,0.3) rectangle (6.65,0.7);
\end{tikzpicture}\,,
\end{align}
where, we denoted by a rectangle the tensor corresponding to $\rho$ (while $\Phi$ corresponds to the identity operator, simply denoted by a continuous line).

When discussing notions of locality and renormalization procedures, a natural concept is that of \emph{blocking}. In essence, this consists in grouping $k$ neighboring sites to form a single one, with which we can associate a blocked tensor $\mathcal{A}_k$. This is formalized in the following.
\begin{defn}
Given the tensor $\mathcal{A}$ generating an MPS, we denote by $\mathcal{A}_k$ the blocked tensor, which is defined by the graphical representation
\be
\mathcal{A}_k=\quad 
\begin{tikzpicture}[baseline={([yshift=1.5ex]current bounding box.center)}, scale=0.7]

\draw (0,0)  -- (0,0.65) ;
\draw (-0.5,0) -- (1.25,0);	
\filldraw[fill=white,draw=black] (0,0) circle [radius=0.25];
\draw (0.75,0)  -- (0.75,0.65);
\filldraw[fill=white,draw=black] (0.75,0) circle [radius=0.25];
\draw(1.75,0) node{$\cdots$};
\draw (2.75,0)  -- (2.75,0.65);
\draw (2.25,0)  -- (3.15,0);	
\filldraw[fill=white,draw=black] (2.75,0) circle [radius=0.25];
\draw [decorate, decoration={brace,amplitude=4pt, raise=4pt},yshift=0pt]
(3.25,-0.5) -- (-0.6,-0.5) node [black,midway,yshift=-0.55cm] {$k$};
\end{tikzpicture}\,.
\ee
\end{defn}
Note that the physical and auxiliary dimensions for the blocked tensor $\mathcal{A}_k$ are $d_k = d^k$ and $D$, respectively. We recall that the blocking procedure is important because, after blocking, for any tensor $\mathcal{B}$ it is always possible to obtain another one, $\mathcal{A}$, in CF and generating the same MPS~\cite{cirac2017matrixRenormalization}.

While the above definitions have been given for MPSs, they can be straightforwardly extended to MPOs. To this end, we recall that an MPO $M^{(N)}$ admits the graphical representation
\be
M^{(N)}=\quad 
\begin{tikzpicture}[baseline={([yshift=0.8ex]current bounding box.center)}, scale=0.7]

\draw (0,-0.5)  -- (0,0.5) ;
\draw (-0.5,0) -- (1.25,0);	
\filldraw[fill=white,draw=black] (0,0) circle [radius=0.25];
\draw (0.75,-0.5)  -- (0.75,0.5);
\filldraw[fill=white,draw=black] (0.75,0) circle [radius=0.25];
\draw(1.75,0) node{$\cdots$};
\draw (2.75,-0.5)  -- (2.75,0.5);
\draw (2.25,0)  -- (3.15,0);	
\filldraw[fill=white,draw=black] (2.75,0) circle [radius=0.25];
\draw (-0.5,0) arc (90:270:0.15) node[right]{};
\draw (3.15,0) arc (90:-90:0.15) node[right]{};
\draw [decorate, decoration={brace,amplitude=4pt, raise=4pt},yshift=0pt]
(3.25,-0.5) -- (-0.6,-0.5) node [black,midway,yshift=-0.55cm] {$N$};
\end{tikzpicture}\,,
\label{eq:MPO_translation}
\ee
where the lower and upper vertical lines correspond to input and output qudits, respectively. Then, any MPO can be trivially mapped onto an MPS, by grouping both input and output lines to form a single physical index (corresponding to a local space with dimension $d^2$), with a graphical identification
\be
\begin{tikzpicture}[baseline={([yshift=-0.8ex]current bounding box.center)}, scale=0.7]
\draw (0,-0.6)  -- (0,0.6) ;
\draw (-0.6,0) node[left]{} -- (0.6,0) node[right]{};	
\filldraw[fill=white,draw=black] (0,0) circle [radius=0.25];
\end{tikzpicture}
\mapsto
\begin{tikzpicture}[baseline={([yshift=-2ex]current bounding box.center)}, scale=0.7]
\draw (0,0)  -- (0,0.8) ;
\draw[dashed] (0.15,0)  -- (0.15,0.8) ;
\draw (-0.6,0) node[left]{} -- (0.6,0) node[right]{};	
\filldraw[fill=white,draw=black] (0,0) circle [radius=0.25];
\end{tikzpicture}\,.
\label{eq:identification}
\ee
This way, the definitions of transfer matrix, CF, and blocking extend naturally also to MPOs.

\subsection{MPUs in qudit systems}
\label{sec:qudit_MPUs}

Let us consider a tensor $\mathcal{U}$ that generates a family of MPOs $U^{(N)}$ of the form in Eq.~\eqref{eq:MPO_translation}, such that $U^{(N)}$ is a unitary operator for all non-negative integers $N$.  The resulting MPOs $U^{(N)}$ are called Matrix Product Unitaries, and were investigated in Refs.~\cite{cirac2017matrix,sahinoglu2018matrix}. In preparation for the fermionic case, we now review some of their properties, and present the main results derived in Ref.~\cite{cirac2017matrix}. 

First, by viewing $\mathcal{U}$ as a tensor generating an MPS as in Eq.~\eqref{eq:identification}, we can define the normalized transfer matrix 
\be
E_{\mathcal{U}}=\frac{1}{d}
\begin{tikzpicture}[baseline={([yshift=-1ex]current bounding box.center)}, scale=0.7]
\draw (0,0)  -- (0,0.8) ;
\draw[dashed] (0.15,0)  -- (0.15,1) ;
\draw (-0.6,0) node[left]{} -- (0.6,0) node[right]{};	
\filldraw[fill=white,draw=black] (0,0) circle [radius=0.25];
\draw (-0.5,1) node[left]{} -- (0.5,1) node[right]{};
\fill[black] (0,1) circle [radius=0.25];
\end{tikzpicture}\,,
\label{eq:transfer_qudit}
\ee
which plays an important role in the analysis of MPUs. In particular, the starting point of Ref.~\cite{cirac2017matrix} is the observation that $E_{\mathcal{U}}$ has just one nonzero eigenvalue, which is equal to one, and that $\mathcal{U}/\sqrt{d}$ is a normal tensor. This follows simply from the unitarity condition $U^{(N)\dagger}U^{(N)}=\openone$, and 
\be
\frac{1}{d^N}{\rm tr}\left[U^{(N)\dagger}U^{(N)}\right]={\rm tr}\left[E_{\mathcal{U}}^N\right]\,.
\label{eq:trace_transfer}
\ee

Using this observation, it was possible to prove that a for any tensor $\mathcal{U}$ generating an MPU, there exists some $k\leq D^4$ such that the blocked tensor $\mathcal{U}_k$ is \emph{simple}. In general, we define a tensor $\mathcal{U}$ to be simple, if there exists two tensors $a$, $b$ such that
\begin{subequations}
\begin{align}
\begin{tikzpicture}[baseline={([yshift=-0.8ex]current bounding box.center)}, scale=0.7]
\draw (0.5,0) arc (-90:90:0.5) node[right]{};
\draw (-0.5,1) arc (90:270:0.5) node[right]{};
\draw (0,-0.5)  -- (0,1.5) ;
\draw (-0.5,0) node[left]{} -- (0.5,0) node[right]{};
\draw (-0.5,1) node[left]{} -- (0.5,1) node[right]{};
\filldraw[fill=white,draw=black] (0,0) circle [radius=0.25];
\filldraw[fill=black,draw=black] (0,1) circle [radius=0.25];
\draw [fill=white] (-1.15,0.28) rectangle (-0.75,0.72);
\draw [fill=white](0.75,0.28) rectangle (1.15,0.72);
\draw (2,0.5) node[left]{$=$};
\draw (2.5,-0.5)  -- (2.5,1.5) ;
\draw (-0.95,0.5) node{$a$};
\draw (0.95,0.5) node{$b$};
\end{tikzpicture}\quad \,,\\
\begin{tikzpicture}[baseline={([yshift=-0.8ex]current bounding box.center)}, scale=0.7]
\draw (0,-0.5)  -- (0,1.5) ;
\draw (1,-0.5)  -- (1,1.5) ;
\draw (-0.5,0) -- (1.5,0);
\draw (-0.5,1) -- (1.5,1);
\draw (2.5,0.5) node[left]{$=$};
\filldraw[fill=white,draw=black] (0,0) circle [radius=0.25];
\filldraw[fill=black,draw=black] (0,1) circle [radius=0.25];
\filldraw[fill=white,draw=black] (1,0) circle [radius=0.25];
\filldraw[fill=black,draw=black] (1,1) circle [radius=0.25];
\draw (2.5,0.5) node[left]{$=$};
\draw (3.2,-0.5)  -- (3.2,1.5) ;
\draw (4.8+0.2,-0.5)  -- (4.8+0.2,1.5) ;
\draw (3.2,0) arc (-90:90:0.5) node[right]{};
\draw (4.8+0.2,1) arc (90:270:0.5) node[right]{};
\draw (2.6,0) -- (3.2,0);
\draw (2.6,1) -- (3.2,1);
\draw (5+0.2,0) -- (5.3+0.2,0);
\draw (5+0.2,1) -- (5.3+0.2,1);
\filldraw[fill=white,draw=black] (3.2,0) circle [radius=0.25];
\filldraw[fill=black,draw=black] (3.2,1) circle [radius=0.25];
\filldraw[fill=white,draw=black] (4.8+0.2,0) circle [radius=0.25];
\filldraw[fill=black,draw=black] (4.8+0.2,1) circle [radius=0.25];
\draw [fill=white] (3.5,0.28) rectangle (3.9,0.72);
\draw [fill=white](4.05+0.2,0.28) rectangle (4.45+0.2,0.72);
\draw (3.7,0.5) node{$b$};
\draw (4.25+0.2,0.48) node{$a$};
\end{tikzpicture}\quad\,.
\end{align}\label{eq:simpleness_condition}
\end{subequations}
It is also clear using a graphical proof that any simple tensor generates an MPU, making the characterization of Ref.~\cite{cirac2017matrix} complete.

The simpleness condition makes it also possible to derive a standard form for the tensor $\mathcal{U}$, by means of which a Fundamental Theorem of MPUs could be proven in Ref.~\cite{cirac2017matrix}. The latter states that two tensors $\mathcal{U}$ and $\mathcal{V}$ generate the same MPU for all non-negative integers $N$ if and only if they have the same standard form (up to single-site gauge transformations).

Importantly, a simple corollary of these results is that for qudit systems any MPU (with finite bond dimension) is a $1D$ QCA, and viceversa. This means that any MPU $U^{(N)}$ maps any operator $\mathcal{O}$ supported on a finite region into another one, $U^{(N)\dagger}\mathcal{O}U^{(N)}$, which is also supported on a finite spatial region. Such an identification between MPUs and QCA is based on the classification of Ref.~\cite{gross2012index}, according to which any given QCA can be represented by a finite number of layers of finite-depth circuits and translations.

The aim of the rest of this work is to explore if and how such a picture generalizes to fermionic $1D$ systems. As we have already anticipated, significant differences emerge, as we lay out in the following. We begin our study in the next section, by introducing fermionic $1D$ tensor networks.

\section{Fermionic Tensor Networks}
\label{sec:fermionic_MPSs}

We consider a chain of $N$ sites, and in each site we have $n_F$ fermionic modes with (physical) annihilation operators $a_{x,j}$, $x = 1,\ldots,N$, $j=1,\ldots , n_F$. Creation and annihilation operators satisfy canonical anticommutation relations
\begin{subequations}
\begin{align}
\{a_{x,j}, a^{\dagger}_{y,k}\}&=\delta_{x,y}\delta_{j,k}\,,\\
\{a_{x,j}, a_{y,k}\}&=0 \,.
\end{align}
\end{subequations}
In the following, we denote by $\ket{\Omega}$ the physical vacuum, with $a_{x,j}\ket{\Omega}=0$, while we introduce the short-hand notation 
\be
\left(a^\dagger_x\right)^{n} = \left(a^\dagger_{x,1}\right)^{n^{(1)}} \cdots \left(a^\dagger_{x,n_F}\right)^{n^{(n_F)}}\,,
\label{eq:short_hand}
\ee
where $(n^{(1)},\ldots ,n^{(n_F)})$ is the binary decomposition of $n$. We can define a fermionic parity operator $\mathcal{P}$ satisfying $\mathcal{P}\ket{\Omega}=\ket{\Omega}$, and
\begin{align}
\mathcal{P}\left(a^\dagger_{x,1}\right)^{n^{(1)}} \cdots \left(a^\dagger_{x,n_F}\right)^{n^{(n_F)}}=(-1)^{|n|}\nonumber\\
\times \left(a^\dagger_{x,1}\right)^{n^{(1)}} \cdots \left(a^\dagger_{x,n_F}\right)^{n^{(n_F)}}\mathcal{P}
\end{align}
where
\be
|n|=\sum_{j=1}^{n_F}  n^{(i)} \quad (\bmod\ 2)\,.
\label{eq:physical_parity}
\ee
Using the above notations, any state in the system can be represented as
\be
|\Psi\rangle=\sum_{n_{1}, \dots, n_{N}=0}^{d-1} c_{n_{1}, \dots, n_{N}} a_{1}^{\dagger n_{1}} \cdots a_{N}^{\dagger n_{N}}|\Omega\rangle\,,
\label{eq:psi_explicit}
\ee
where $d=2^{n_F}$. In the following, we will always work with states having well-defined parity, namely
\be
\mathcal{P}\ket{\Psi}=(-1)^{|\Psi|}\ket{\Psi}\,.
\ee
This implies that $c_{n_{1}, \dots, n_{N}}$ is zero unless $\sum_{j=1}^N|n_j|\equiv |\Psi|$ $({\rm mod}\ 2)$.

Before discussing fMPSs, it is useful to recall that, in the case of fermions, there are two natural types of boundary conditions to be considered: periodic and antiperiodic. Accordingly, one can define two types of translation operators. If periodic boundary conditions are assumed, we define $T_P$ by its action 
\begin{subequations}\label{eq:periodic_tr}
\begin{align}
T_P\ket{\Omega}&=\ket{\Omega}\,,\\
T_P a^{\dagger}_{x,j} T_P^{-1}&=a^{\dagger}_{x+1,j} \,,\quad 1\leq x\leq N-1\,,\\
T_P a^{\dagger}_{N,j} T_P^{-1}&=a^{\dagger}_{1,j}\,,
\end{align}
\end{subequations}
while for antiperiodic boundary conditions, we define $T_{AP}$ by
\begin{subequations}\label{eq:antiperiodic_tr}
\begin{align}
T_{AP}\ket{\Omega}&=\ket{\Omega}\,,\\
T_{AP} a^{\dagger}_{x,j} T_{AP}^{-1}&=a^{\dagger}_{x+1,j} \,,\quad 1\leq x\leq N-1\,,\\
T_{AP} a^{\dagger}_{N,j} T_{AP}^{-1}&=-a^{\dagger}_{1,j}\,.
\end{align}
\end{subequations}
Using Eqs.~\eqref{eq:periodic_tr}, \eqref{eq:antiperiodic_tr}, one can easily find the condition on the coefficients $c_{n_{1}, \dots, n_{N}} $ for the state~\eqref{eq:psi_explicit} to be invariant under $T_P$ or $T_{AP}$. Specifically
\begin{align}
T_P\ket{\Psi}&=\ket{\Psi} \Leftrightarrow  c_{n_{1}, \dots, n_{N}} = (-1)^{|n_1|(|\Psi|+1)}c_{n_{2}, \dots, n_N, n_{1}}\,, \nonumber\\
T_{AP}\ket{\Psi}&=\ket{\Psi} \Leftrightarrow c_{n_{1}, \dots, n_{N}} = (-1)^{|n_1||\Psi|}c_{n_{2}, \dots, n_N,  n_{1}}\,.
\label{eq:invariant_tr}
\end{align} 

We stress that both types of boundary conditions appear frequently when working with fermionic systems. For instance, using Eq.~\eqref{eq:invariant_tr}, one can easily see that a state obtained by occupying each fermionic mode (with $n_F=1$) is invariant under translations with periodic or antiperiodic boundary conditions, depending on whether $N$ is odd or even. Accordingly, both types of boundary conditions will be studied in this work.

\subsection{Fermionic MPSs}
\label{sec:fMPOs}

Up to now, several equivalent formulations of fermionic TNs have been developed~\cite{corboz_fermionic_2009,gu_grassmann_2010,corboz_simulation_2010-1,corboz_simulation_2010,kraus_fermionic_2010,wahl_symmetries_2014,bultinck2017fermionic,kapustin_spin_2018,turzillo_fermionic_2019}. Here, we will focus on the fiducial-state formalism introduced in Refs.~\cite{kraus_fermionic_2010,wahl_symmetries_2014}, whose appeal lies mainly in its physically motivated construction. In the second, more technical part of the paper,  however, we will also make use of the formalism of graded tensor networks recently introduced in Refs.~\cite{bultinck2017fermionic,kapustin_spin_2018} (see also \cite{bultinck2017fermionic_PEPS,Shukla2020Tensor}), which is reviewed in Appendix~\ref{sec:graded_TN}.

For qudit systems, one can think of MPSs (or, more generally of PEPS) as obtained from a sequence of local projections onto maximally entangled pairs of auxiliary qudits~\cite{Verstraete2004Density}. The idea of Refs.~\cite{kraus_fermionic_2010,wahl_symmetries_2014} is that the same construction can be carried out for fermionic systems, provided that the auxiliary d.o.f. are taken to be fermionic particles themselves. As a technical point, it is convenient to choose such auxiliary particles as Majorana fermions. Furthermore, one needs to enforce a given fermionic parity on the local projectors, in order to ensure that the fMPS itself has well-defined parity. 

A detailed discussion of this construction is provided in Appendix~\ref{sec:fiducial_state}, while here we only report the final result for the coefficients in Eq.~\eqref{eq:psi_explicit}. The explicit form of the latter depend on whether periodic or antiperiodic boundary conditions are assumed. In particular, we have, respectively,
\begin{align}
c_{n_{1}, \ldots, n_{N}}&={\rm tr}\left(Z A^{n_{1}} \cdots A^{n_{N}}\right)\,,\quad ({\rm PBC})\label{eq:fMPS_pbc}\\
c_{n_{1}, \ldots, n_{N}}&={\rm tr}\left(A^{n_{1}} \cdots A^{n_{N}}\right)\quad ({\rm APB})
\label{eq:fMPS_apbc}
\end{align} 
Here we have introduced the parity operator $Z$
\be
Z=
\begin{pmatrix}
	\openone_e & 0 \\
	0 & -\openone_o
\end{pmatrix}\,,
\label{eq:parity_operator}
\ee
where $\openone_e$, $\openone_o$ are identity operators acting on the even and odd subspaces of dimensions $D_e$ and $D_o$, with $D_e=D_o=D/2$, and $D=2^{N_F}$, for some positive integer $N_F$. Furthermore, in order to ensure that the fMPS has well-defined parity, we require that the matrices $A^n$ satisfy\footnote{We note that Eq.~\eqref{eq:commutation} implies that $\mathcal{A}$ is an even tensor. One could also consider fMPSs built out of odd local tensors, which would lead to states with well-defined parity for each non-negative integer $N$. However, choosing $\mathcal{A}$ to be even is not a restriction, since blocking twice an odd tensor yields an even one. }
\be
Z A^{n}=(-1)^{|n|}  A^{n}Z\,,
\label{eq:commutation}
\ee
where $|n|$ is defined in Eq.~\eqref{eq:physical_parity}. In the following, we will thus define a periodic (antiperiodic) fMPS to be a state of the form~\eqref{eq:psi_explicit}, where the coefficients can be cast as in Eq.~\eqref{eq:fMPS_pbc} (Eq.~\eqref{eq:fMPS_apbc}), and where the local tensors satisfy~\eqref{eq:commutation}. Note that it is straightforward to verify that the coefficients~\eqref{eq:fMPS_pbc}, \eqref{eq:fMPS_apbc} satisfy Eqs.~\eqref{eq:invariant_tr}.

The parity operator $Z$ allows us to assign a $\mathbb{Z}_2$-\emph{grading} structure to the auxiliary space $\mathbb{C}^{D}$, which simply means that $\mathcal{H}_D$ can be divided into two complementary subspaces of even and odd states. Here we define a state to be even or odd if it is a superposition of  $Z$-eigenstates with eigenvalues $+1$ or $-1$, respectively. If $|\alpha)$ is an eigenstate of $Z$, we will denote the corresponding eigenvalue by $(-1)^{|\alpha|}$, with $|\alpha|=0$ ($|\alpha|=1$) for even (odd) states, namely
\begin{subequations}
	\begin{align}
	Z|\alpha)=&(-1)^{|\alpha|}|\alpha)\,.
	\end{align}
\end{subequations}
The observations above also allow us to define in a natural way the parity of tensors acting on both the auxiliary and physical spaces. In particular, let $A^{n}_{\alpha,\beta}\neq 0$ be an element of the tensor $\mathcal{A}$. Then, we can define the parity of $\mathcal{A}$ as
\begin{align}
|\mathcal{A}|&=|n|+|\alpha|+|\beta|\qquad (\bmod\ 2)\,.
\end{align}
Clearly, this definition only makes sense if it is independent from the choice of $n$, $\alpha$ and $\beta$.
This is the case if $A^n$ satisfy Eq.~\eqref{eq:commutation}, which in fact implies that $\mathcal{A}$ is an even tensor

Given the parity operator~\eqref{eq:parity_operator}, it is possible to write down the general form of the matrices $A^{n}$ satisfying Eq.~\eqref{eq:commutation}. In particular, it is easy to show that, in the basis where $Z$ is written as in Eq.~\eqref{eq:parity_operator}, $A^n$ must have the following block-structure 
\be
A^n=
\begin{pmatrix}
	B^n & 0 \\
	0 & C^n
\end{pmatrix}\,, \qquad |n|=0\,,
\label{eq:even}
\ee
\be
A^n=
\begin{pmatrix}
	0 & B^n \\
	C^n &0
\end{pmatrix}\,, \qquad |n|=1\,,
\label{eq:odd}
\ee
where $B^{n}$, $C^{n}$ are arbitrary matrices.  

So far, we have considered $D=2^{N_F}$, $d=2^{n_F}$. However, this can be naturally relaxed as the tensor $\mathcal{A}$ can have many zeros and thus we can compress its dimensions. This happens when for some values of $n$, $A^{n}_{\alpha,\beta} = 0$ for all $\alpha,\beta$ or, alternatively, for some $\alpha$, $A^n_{\alpha,\beta}=A^n_{\beta,\alpha}=0$ for all $n$ and $\beta$. In this case we can restrict ourselves to subspaces of $\mathbb{C}^d$ and $\mathbb{C}^D$ with dimension $d^\prime$ and $D^\prime$, respectively. We can also call $Z^\prime$ and $\mathcal{P}^\prime$ the projection of $Z$ and $\mathcal{P}$ onto these subspaces: importantly, they remain diagonal with diagonal elements $\pm 1$, meaning that the reduced spaces maintain the $\mathbb{Z}_2$-grading structure. In the following, we will consider that this is the case and drop the primes in the notation, so that $d$ and $D$ can take arbitrary values. 

It is straightforward to employ the fiducial-state formalism to also treat fMPOs. In general, any fermionic operator $U^{(N)}$ can be written in the form
\be
U^{(N)}=\sum_{n_{1}, \dots, n_{N}=0\atop m_{1}, \dots, m_{N}=0}^{d-1} c^{n_{1}, \dots, n_{N}}_{m_{1}, \dots, m_{N}} f_1^{n_1,m_1}\cdots f_N^{n_N,m_N}\,.
\label{eq:fermi_operator}
\ee
Here, we introduced the operators $f_j^{n,m}$, defined by $f_j^{n,m}\ket{\Omega}=\delta_{m,0}(a^{n}_j)^\dagger\ket{\Omega}$, and
\begin{subequations}\label{eq:small_f_operators}
\begin{align}
f_{j}^{n,m} \left(a_{k}^\dagger\right)^{p}=& (-1)^{(|n|+|m|)|p|} \left(a_{k}^\dagger\right)^{p} f_{j}^{n,m}\,, \quad {j\neq k}\,,\label{eq:commutation_fermi_local}\\
f_{j}^{n,m} \left(a^\dagger_j\right)^{p}&=\delta_{m,p} \left(a_j^\dagger\right)^{n}\,,
\end{align}
\end{subequations}
where $(a_j^\dagger)^{n}$ was defined  in Eq.~\eqref{eq:short_hand}. Note that from these equations it also follows
\be
\left[f_{j}^{n,m}\right]^\dagger= f_{j}^{m,n}\,.
\label{eq:adjoint_f}
\ee
For example, in the case $n_F=1$, we have $f_j^{0,0}=a_j a_j^\dagger$, $f_j^{0,1}= a_j$, $f_j^{1,0}= a_j^\dagger$ and $f_j^{1,1}=a_j^\dagger a_j$.

Analogously to fMPSs, fMPOs are defined by a specific form for the coefficients $c^{n_{1}, \dots, n_{N}}_{m_{1}, \dots, m_{N}}$, which is obtained by following the construction outlined in Appendix.~\ref{sec:fiducial_state}. Once again, we have to distinguish between periodic and antiperiodic boundary conditions, for which we obtain, respectively,
\begin{align}
c^{n_{1}, \dots, n_{N}}_{m_{1}, \dots, m_{N}}&={\rm tr}\left(Z U^{n_1,m_1} \cdots U^{n_N,m_N}\right)\,,\quad ({\rm PBC})\label{eq:fMPO_pbc}\\
c^{n_{1}, \dots, n_{N}}_{m_{1}, \dots, m_{N}}&={\rm tr}\left(U^{n_1,m_1} \cdots U^{n_N,m_N}\right)\quad ({\rm ABC})
\label{eq:fMPO_apbc}
\end{align} 
Here $Z$ is defined as in Eq.~\eqref{eq:parity_operator}, while $U^{n,m}$ are $D\times D$ matrices which must satisfy 
\be
Z U^{n,m}=(-1)^{|n|+|m|} U^{n,m} Z\,.
\label{eq:op_commutation}
\ee

It is important to note that elementary operations with fMPSs and fMPOs, such as blocking or composition, in general result in additional signs for the corresponding tensors. Since these signs are at the root of some qualitative differences arising in the case of fermionic MPUs, we discuss them in some detail in Appendix~\ref{sec:elementary_operators}.

\section{Fermionic Matrix Product Unitaries}
\label{sec:fMPUs}

We now finally introduce the fMPUs. There are two main results in this section: the construction of a translation-invariant fMPO which is unitary but not a QCA (cf. Sec.~\ref{sec:nonlocal_fmpu}), and the derivation of the TN form of the Majorana-shift operator for open and periodic boundary conditions, cf. Eqs.~\eqref{eq:maj_anti}, and \eqref{eq:maj_per}.

Based on the formalism of fermionic TNs, we can formulate a very natural generalization of the definition of MPUs (presented in Sec.~\ref{sec:qudit_MPUs}) to fermionic systems. Namely, we call $U^{(N)}$ an fMPU if is it an fMPO and $\left[U^{(N)}\right]^\dagger U^{(N)}=\openone$ for all $N\geq 1$. Furthermore, differently from the case of qudits, it makes sense to define fMPUs both for periodic and antiperiodic boundary conditions, corresponding to the coefficients~\eqref{eq:fMPO_pbc} and \eqref{eq:fMPO_apbc}, respectively.

Arguably, one of the most important questions regarding fMPUs is whether they are all locality-preserving. We find that this is not the case for periodic fMPUs, in stark contrast with the case of qudits. We show this in the next subsection, by providing an explicit example of such a non-locality-preserving fMPU. Technically, this qualitative difference between fermions and qudits arises because of the presence of the operator $Z$ in the trace in Eq.~\eqref{eq:fMPO_pbc}, which makes it possible for the transfer matrix  of periodic fMPUs to display a nontrivial spectrum, cf. also Eq.~\eqref{eq:trace_transfer_periodic}.

\subsection{Non-locality-preserving fMPUs}
\label{sec:nonlocal_fmpu}

We now provide an explicit example of a periodic fMPU which is not locality-preserving. Once again, we use the fiducial-state formalism, but our  example can be straightforwardly translated in the language of graded TNs, cf. Appendix~\ref{sec:graded_TN}.

Let us consider a chain of $N$ sites, with two fermionic modes per site, $n_F=2$, corresponding to annihilation operators $a_{x,1}$, $a_{x,2}$, $x=1,\ldots , N$. Let $U^{(N)}$ be a periodic fMPO of the form~\eqref{eq:fermi_operator} with coefficients given in Eq.~\eqref{eq:fMPO_pbc}. In order to construct our example, we choose $D=3$, so that
\be
Z=
\begin{pmatrix}
	+ 1 & 0 & 0 \\
	0 & +1 & 0 \\
	0 & 0 & -1 \\
\end{pmatrix}\,.
\ee
We stress, once again, that the bond-dimension can be taken to be a power of two, by simply appending arbitrarily many lines and columns filled with zeros. We now define the local tensor $\mathcal{U}$, by specifying its non-zero matrix elements $U^{m,n}_{\alpha,\beta}$, which are
\begin{align}
U^{0,0}_{\alpha,\beta}&=\delta_{\alpha,\beta}\,, \quad \alpha,\beta=0,1,2\,,\label{eq:u_00}\\
U^{n,m}_{\alpha,\beta}&=\delta_{\alpha,n-1}\delta_{\beta,m-1}\,, \ \alpha,\beta=0,1,2\,,\ n,m=1,2,3\,.\label{eq:u_nm}
\end{align}
For completeness, we also tabulate the explicit matrices $U^{n,m}$ in Appendix~\ref{sec:tabulated_local_tensors}, from which it is immediate to verify that they satisfy Eq.~\eqref{eq:op_commutation}, and thus generate a legitimate periodic fMPO. In the following, we prove unitarity of $U^{(N)}$, and analyze its properties.

\paragraph{Spectral properties of the TM.} First of all, it is interesting to note that the spectrum of the transfer matrix [defined in Eq.~\eqref{eq:fermionic_transfer_matrix} for fMPOs], denoted by $\mathcal{S}_E$, is nontrivial. In particular, a direct calculation gives $\mathcal{S}_E=\{\lambda_j\}_{j=0}^8$, with $\lambda_0=1$ and $\lambda_j=1/4$, for $j=1,\ldots 8$.  This is already a point of departure with respect to the case of qudits, where the transfer matrix of an MPU has a single nonzero eigenvalue, which is equal to one. By further inspection, one can see that the eigenvector of $E$ associated with $\lambda_0=1$ is also an eigenstate of $Z\otimes Z$ with eigenvalue $+1$. Analogously, one can see that the rest of the eigenvectors associated with $\lambda_j$, $j>0$, are also eigenstates of $Z\otimes Z$: four of them with  eigenvalue $+1$, and the other four with eigenvalue $-1$. Thus, we obtain from Eq.~\eqref{eq:trace_transfer_periodic}, $(1/d^N){\rm tr}U^{(N)\dagger}U^{(N)}=1+(4/4^N)-(4/4^N)=1$, as it should for a unitary operator.

\paragraph{Proof of unitarity} Next, we show that $U^{(N)}$ is unitary. One could do this by checking that the fMPO generated by the tensor $\mathbb{U}^{(N)}$ defined in Eq.~\eqref{eq:unitary_stacked} is the identity. Here, we follow a different strategy, which is based on the action of $U^{(N)}$ on basis states. In particular, from the explicit form of the local tensors, we show in Appendix~\ref{sec:tabulated_local_tensors} that
\be
U^{(N)} \ket{\Omega}=\ket{\Omega}\,,
\label{eq:action_1}
\ee
and, $\forall P=1,\ldots, N$, $\{j_\ell\}_{\ell =1}^P\in \{1,2,3\}^{ P}$,
\be
U^{(N)} {a}^{j_1\dagger}_{i_1} {a}^{ j_2 \dagger}_{i_2}\cdots {a}^{j_P \dagger}_{i_P}\ket{\Omega} = (-1)^\gamma  {a}^{j_P\dagger}_{i_1} {a}^{j_1\dagger}_{i_2}\cdots {a}^{ j_{P-1}\dagger}_{i_P}\ket{\Omega}\,.
\label{eq:action_2}
\ee
Here $\{i_{k}\}_{k=1}^P$ is a strictly increasing sequence of integers, with $1\leq i_k \leq N$, which label the site at which the creation operators act on. Furthermore, $(-1)^\gamma$ is a sign which depends on the specific state. Note that in Eq.~\eqref{eq:action_2}, {${j_\ell}$} takes value in the set ${\{1,2,3\}}$, namely ${j_\ell\neq 0}$, and that ${U}^{(N)}$ does not act as a translation operator, but only shifts the creation operators at fixed positions. From Eqs.~\eqref{eq:action_1}, \eqref{eq:action_2}, it is clear that $U^{(N)}$ maps two orthonormal bases one onto another, and it is thus unitary.

\paragraph{Proof of non-locality}
Finally, we prove explicitly that the operator $U^{(N)}$ is not locality-preserving. Namely there exists a local operator $\mathcal{O}_j$, acting only on site $j$, such that the support of $U^{\dagger}\mathcal{O}_j U$ is not contained in any finite region. In order to see this, consider the state $\ket{\Psi_{j,k}}= a^{\dagger}_{j,1}  a^{\dagger}_{k,1} \ket{\Omega}$ and choose the operator
$\mathcal{O}_j=a^\dagger_{j,2} a_{j,1}$. Then, using the explicit action in Eq.~\eqref{eq:action_2}, we have 
\begin{align}
U^{\dagger} \mathcal{O}_j U \ket{\Psi_{j,k}}=(-1)^{\gamma}& U^{\dagger} \mathcal{O}_j\ket{\Psi_{j,k}}\nonumber\\
=(-1)^{\gamma} U^{\dagger}a^\dagger_{j,2} a^\dagger_{k,1}\ket{\Omega}&= (-1)^{\gamma^\prime} a^\dagger_{j,1} a^\dagger_{k,2}\ket{\Omega}\,,
\end{align}
where we chose $j\neq k$ and where $(-1)^{\gamma}$, $(-1)^{\gamma^\prime}$ are signs that are irrelevant for our discussion. We see that the operator $U^{\dagger} \mathcal{O}_j U $ induces a modification on the mode at site $k$, which can be taken arbitrarily far away from $j$, and thus it is not localized in a neighborhood of site $j$.

In conclusion, we exhibited an example of an fMPU with periodic boundary conditions which is not locality-preserving. The choice of periodic boundary conditions was important: indeed as we discuss in the next section, fMPUs with antiperiodic boundary conditions are necessarily QCA.

\subsection{The Majorana-shift operator}
\label{sec:majorana_shift_operator}

The results of the previous subsection show that, in general, fMPUs are not QCA.  On the other hand, it is also true that fMPUs with coefficients in the form of either Eq.~\eqref{eq:fMPO_pbc} or Eq.~\eqref{eq:fMPO_apbc} do not exhaust all the possible QCA (not even after blocking), once again in stark contrast with the case of qudits. This was already mentioned in Sec.~\ref{sec:QCA}, where we anticipated the special role played by the translation of Majorana modes~\cite{fidkowski2017Radical,fidkowski2019interacting}. We review it explicitly in this subsection, where we also derive the corresponding TN representation.

Let us consider a chain of $N$ sites, with one fermionic mode per site, $n_F=1$, associated with fermionic annihilation operators $a_n$. We can introduce the $2N$ Majorana modes $\gamma_n$ by
\begin{align}
a_{n}&=\frac{1}{2}\left(\gamma_{2 n-1}+i \gamma_{2 n}\right)\,, \\ a_{n}^{\dagger}&=\frac{1}{2}\left(\gamma_{2 n-1}-i \gamma_{2 n}\right)\,.
\label{eq:real_vs_majorana}
\end{align}
We consider now two automorphisms of the operator algebra $\alpha_{P}$, $\alpha_{AP}$, defined by the following action
\begin{align}
\alpha_{P}(\gamma_{j})&=\gamma_{j+1}\,,\quad j=1,\ldots 2N-1\,,\\
\alpha_{P}(\gamma_{2N})&=\gamma_{1}\,,
\end{align}
and
\begin{align}
\alpha_{AP}(\gamma_{j})&=\gamma_{j+1}\,,\quad j=1,\ldots 2N-1\,,\\
\alpha_{AP}(\gamma_{2N})&=-\gamma_{1}\,.
\end{align}
Clearly, $\alpha_{P}$ and $\alpha_{AP}$ implement a fermionic QCA, which we will refer to as Majorana shift (or translation).

Importantly, both  $\alpha_{P}$ and $\alpha_{AP}$ can be represented by a unitary operator. Namely there exist $M^{(N)}_{P}$, $M^{(N)}_{AP}$, such that $M_{P}^{(N)}M^{(N)\dagger}_{P}=M_{AP}^{(N)}M^{(N)\dagger}_{AP}=\openone$, and $\alpha_{P}(\mathcal{O})=M_P^{(N)}\mathcal{O}M^{(N)\dagger}_P$, $\alpha_{AP}(\mathcal{O})=M_{AP}^{(N)}\mathcal{O}M^{(N)\dagger}_{AP}$, for all operators $\mathcal{O}$. Let us focus, for instance, on the case of  antiperiodic boundary conditions. After a bit of guesswork, it is not difficult to arrive at the following explicit expression for the operator $M_{AP}$~\footnote{Interestingly, in Ref.~\cite{fidkowski2019interacting} it was shown that $M^{(N)}_{AP}$ can also be written as a Gaussian operator, namely as the exponential of an expression which is quadratic in the Majorana modes. This provides an alternative representation to the one given here in terms of fermionic TNs. We also note that an equation similar to~\eqref{eq:maj_product} have appeared recently in Ref.~\cite{Sannomiya2019Supersymmetry} for periodic boundary conditions.}
\be
M_{AP}^{(N)}=\frac{(1-\gamma_1 \gamma_2)}{\sqrt{2}} \frac{(1-\gamma_2 \gamma_3)}{\sqrt{2}}
\cdots \frac{(1-\gamma_{2N-1}\gamma_{2N})}{\sqrt{2}}\,.
\label{eq:maj_product}
\ee 
Note that the order of the factors here is important, since they do not all commute with one another. Starting from Eq.~\eqref{eq:maj_product}, one can rewrite $M_{AP}^{(N)}$ as in Eq.~\eqref{eq:fermi_operator}, where the coefficients read
\be
c^{n_{1}, \dots, n_{N}}_{m_{1}, \dots, m_{N}}=\frac{1}{\sqrt{2}}{\rm tr}\left(M^{n_1,m_1} \cdots M^{n_N,m_N}\right)\quad ({\rm ABC})\,.
\label{eq:maj_anti}
\ee
Here the trace is over a $2$ dimensional space, and we introduced
\begin{align}
M^{n,m}=\frac{1}{\sqrt{2}}i^m\sigma_x^{n+m}\,, \quad n,m=0,1\,,
\label{eq:m_matrices}
\end{align}
in terms of the Pauli matrix $\sigma^x$. A proof of Eq.~\eqref{eq:maj_anti} is reported in Appendix~\ref{sec:majorana_shifts}. Note that the  matrices $M^{n,m}$ satisfy~\eqref{eq:commutation}, with $Z={\rm diag}(1,-1)$, so that~\eqref{eq:maj_anti} admits a fiducial-state representation. Note also that the local tensor of the inverse shift can be obtained using Eq.~\eqref{eq:conjugate_tensor}.

In order to obtain $M_{P}$, one could be tempted to simply insert the operator $Z={\rm diag}(1,-1)$ into the trace in Eq.~\eqref{eq:maj_anti}. However, due to the specific form of the tensors  $M^{n,m}$, the resulting coefficients are all vanishing. In fact, as we show in Appendix~\ref{sec:majorana_shifts}, it turns out that a unitary operator with periodic boundary conditions is obtained if the matrix $X=\sigma^x$ is inserted instead. In particular, the unitary operator $M^{(N)}_P$ can be represented as in Eq.~\eqref{eq:fermi_operator}, where the coefficients are given by
\be
c^{n_{1}, \dots, n_{N}}_{m_{1}, \dots, m_{N}}=\frac{1}{\sqrt{2}}{\rm tr}\left(X \tilde{M}^{n_1,m_1} \cdots \tilde{M}^{n_N,m_N}\right)\quad ({\rm PBC})\,,
\label{eq:maj_per}
\ee 
with
\begin{align}
\tilde{M}^{n,m}&=\frac{1}{\sqrt{2}}(-i)^m\sigma_x^{n+m}\,, \quad n,m=0,1\,.
\label{eq:tilde_m_matrices}
\end{align}
Crucially, viewing the auxiliary space as a graded space with parity $Z={\rm diag}(1,-1)$, the operator $X$ is odd, namely $|X|=1$. Furthermore, $[X,M^{n,m}]=0$ for all $n$, $m$. Thus, using Eq.~\eqref{eq:invariant_tr}, we see immediately that $M_{P}$ is indeed invariant under translation with periodic boundary conditions. 

Although they appear different, a legitimate question is whether the coefficients~\eqref{eq:maj_anti}, \eqref{eq:maj_per} can be cast in the form~\eqref{eq:fMPO_pbc}, ~\eqref{eq:fMPO_apbc}. In fact, this is not the case, and for periodic boundary conditions this can be seen very easily. Indeed, the presence of $X$ in the trace~\eqref{eq:maj_per} implies that $M^{(N)}_{P}$ is an odd operator (namely, it maps even states into odd ones, and viceversa), while any fMPU defined by~\eqref{eq:fMPO_pbc} is necessarily even. The same conclusion holds for antiperiodic boundary conditions, as it can be seen by inspection of the corresponding canonical form, discussed in the next section. 

Finally, it is instructive to apply a JW transformation to $M^{(N)}_{AP}$. In fact, starting from Eq.~\eqref{eq:maj_product} and using standard techniques, it is possible to rewrite it as a qubit ($d=2$) MPO with open boundaries. A crucial point, however, is that the bulk tensors of this MPO do not generate an MPU when periodic boundary conditions are imposed. This example shows that fermionic QCA can not be simply understood in terms of periodic MPUs on qudit chains.

\section{Generalized fermionic MPUs}
\label{sec:generalized_fMPUs}
\noindent In this section, we present our second main result, namely we introduce a class of ``generalized'' fMPUs (Sec.~\ref{sec:gfMPUs_two_kinds}), and identify a condition on the corresponding local tensors such that any representative of this family is a fQCA and viceversa (Secs.~\ref{sec:antiperiodic_fMPUs}, and \ref{sec:periodic_fMPUs}).

The main motivation for generalized fMPUs is to have a class of fMPOs which also include the Majorana-shift operators~\eqref{eq:maj_anti}, \eqref{eq:maj_per}. A natural way to do this is to allow for an additional operator acting on the auxiliary space, such that it has well-defined parity and that it implements the correct boundary conditions. More precisely, we say that $U^{(N)}$ is a generalized fMPU if it is unitary for $N\geq 1$, and can be written as in~\eqref{eq:fermi_operator} with coefficients of the form
\be
c^{n_{1}, \dots, n_{N}}_{m_{1}, \dots, m_{N}}={\rm tr}\left(S U^{n_1,m_1} \cdots U^{n_N,m_N}\right)\,.
\label{eq:generalized_fMPU}
\ee
As usual, we require that all the local tensors have well-defined parity. Namely, the trace in Eq.~\eqref{eq:generalized_fMPU} is over a graded space with parity operator $Z$ defined in Eq.~\eqref{eq:parity_operator}, and the matrices $U^{n,m}$ satisfy Eq.~\eqref{eq:commutation}. Furthermore, the operator $S$ must also have well-defined parity, and satisfy particular commutation relations with $U^{n,m}$. Since these depend on the boundary conditions chosen, it is convenient to  separate the periodic and antiperiodic cases. For periodic boundary conditions, we allow for $S$ to be even or odd, but we require that 
\be
S U^{n,m}=(-1)^{(|S|+1)(|n|+|m|)} U^{n,m}S  \quad ({\rm PBC})
\ee
On the other hand, it can be shown, using Eq.~\eqref{eq:invariant_tr}, that there are no odd states or operators that are invariant under translations with antiperiodic boundary conditions. Therefore, in the latter case we require that $S$ is even, and satisfies
\be
[S,U^{m,n}]=0\,,\quad |S|= 0 \quad ({\rm a. b. c.})
\ee
The conditions imposed above constitute the minimal requirement in order to have an fMPO with well-defined parity, that is invariant under translations with periodic and antiperiodic boundary conditions, respectively.

It is immediate to see that the Majorana-shift operators are indeed generalized fMPUs. In particular, for periodic and antiperiodic boundary conditions,  we see from Eqs.~\eqref{eq:maj_anti}, \eqref{eq:maj_per} that $S=X/\sqrt{2}$ and $S=\openone_2/\sqrt{2}$, respectively. On the other hand, we know from Sec.~\ref{sec:nonlocal_fmpu} that not all generalized fMPUs are QCA, so that a natural question is what conditions on $S$ and $U^{n,m}$ guarantee that $U^{(N)}$ is locality-preserving.  

In order to address this problem, we need to introduce some technical definitions that generalize those presented in Sec.~\ref{sec:bosonic_tn} for irreducible and normal tensors. These are inspired by the work~\cite{bultinck2017fermionic}, where it was pointed out that the notion of irreducible tensor should be modified in the fermionic case, due to the requirement that all the tensors are even. We begin with the following
\begin{defn}\label{def:git}
	Let $V$ be a graded Hilbert space, with parity operator $Z$, and $W\subset V$ with associated orthogonal projector $P_{W}$. We say that $W$ is a graded subspace if $[P_W,Z]=0$. Accordingly, an even tensor $\mathcal{A}$ is said to be graded irreducible (GI) if there is no proper graded subspace which is left invariant by all the matrices $A^n$.
\end{defn}

Clearly, if $\mathcal{A}$ is GI, there are two possibilities: either there is no invariant subspace, or there is an invariant subspace that is non-graded. In the latter case, we prove in Appendix~\ref{sec:graded_canonical_form} (see also Ref.~\cite{bultinck2017fermionic}) that there exists a quite precise characterization of the matrices $A^{n}$, which, up to an even gauge transformation, take the form
\be
A^{n}=\left(\sigma^{x}\right)^{|n|} \otimes B^{n}\,,
\label{eq:GItensors}
\ee
where $B^{n}$ are $D/2\times D/2$ matrices, and where the parity operator is $Z=\sigma^z\otimes \openone$. Note that in this case the auxiliary space splits into even and odd subspaces with the same dimension. 

Next, we generalize the notion of normal tensors and canonical form, which will be of great importance for the classification of fMPUs.
\begin{defn}\label{def:graded_normal_tensor}
	We say that an even tensor $\mathcal{A}$ is a graded normal tensor (GNT) if $(i)$: there is no non-trivial graded invariant subspace; $(ii)$ the corresponding transfer matrix has either exactly one or exactly two eigenvalues of magnitude and value equal to its spectral radius which is equal to $1$. In the first case we say that the GNT is non-degenerate, while in the second case we say that it is degenerate.
\end{defn}

\begin{defn}\label{def:graded_canonical_form}
	We say that an even tensor $\mathcal{A}$ generating an fMPS is in Graded Canonical Form (GCF) if: $(i)$: the matrices are of the form $A^n=\oplus_{k=1}^r \mu_k A_k^n $, where $\mu_k\in \mathbb{C}$ and the spectral radius of the transfer matrix $E_k$ associated with $A_k^n$ is equal to one; $(ii)$: the parity operator $Z$ has the same block structure as $A^n$ and, for all $k$, $\mathcal{A}_k$ is a GNT.
\end{defn}
Here it is useful to mention that for antiperiodic boundary conditions, we can always change the sign of the parity operator $Z$ in the auxiliary space without modifying the state. This is because $Z$ does not enter into the trace defining the fMPS coefficients~\eqref{eq:fMPS_apbc} and, if $ZA^{i}=(-1)^{|i|}A^iZ$, we also have $\tilde{Z}A^{i}=(-1)^{|i|}A^i\tilde{Z}$ with $\tilde{Z}=-Z$. Furthermore, if $\mathcal{A}$ is in GCF, we have the freedom to changing the sign of the parity for each graded normal block independently, without modifying the state.

The GCF is discussed in detail in Appendix~\ref{sec:graded_canonical_form}, where it is shown that, in the case of antiperiodic boundary conditions, it is possible to derive a series of results that generalize in a nontrivial way those for the canonical form of MPSs. In particular, we prove the following fundamental theorems.
 \begin{thm}
	After blocking, for any even tensor, $\mathcal{A}$, it is always possible to obtain another even tensor, $\mathcal{A}_{GCF}$, in GCF and generating the same fMPS with antiperiodic boundary conditions. 
\end{thm}

\begin{thm}\label{thm2_fundamental}
	Consider two (even) tensors $\mathcal{A}$ and $\mathcal{B}$ in GCF, with diagonal parity operators in the auxiliary space $Z_a$, $Z_b$. If they generate the same fMPS with antiperiodic boundary conditions for all $N$ then: (i) the dimensions of the matrices $A^i$ and $B^i$ coincide; (ii) there exists an invertible matrix $X$, and permutation matrix $\Pi$ such that $A^{i}=X\Pi B^{i} \Pi^{-1}X^{-1}$,  and $[X,Z_a]=0$.
\end{thm}

\subsection{Type I and type II generalized fMPUs}
\label{sec:gfMPUs_two_kinds}

The definition of graded normal tensors gives us a hint of why fMPUs of the form~\eqref{eq:fMPO_pbc}, ~\eqref{eq:fMPO_apbc} are not enough to capture all fermionic QCA. In the case of periodic boundary conditions, for instance, one can see from Eq.~\eqref{eq:GItensors} that a degenerate normal tensor yields vanishing coefficients when plugged into Eq.~\eqref{eq:fMPO_pbc}. On the other hand, the additional operator $S$ in the definition~\eqref{eq:generalized_fMPU} allows us to ``close the trace'' in such a way that degenerate normal tensors give non-vanishing contributions. In fact, this is exactly what happens for the Majorana-shift operator~\eqref{eq:maj_per}.

Motivated by this discussion, we define two special classes of generalized fMPUs. These provide a ``minimal'' subset of the operators~\eqref{eq:generalized_fMPU} which allow for both degenerate and non-degenerate normal tensors. In particular, for periodic boundary conditions, we say that $U^{(N)}$ is a generalized fMPU of the first kind (or type I) if it can be cast in the form~\eqref{eq:generalized_fMPU} with  $S=e^{i\alpha}Z$, while it is of the second kind (or type II) if 
\begin{align}
S&=\frac{e^{i\alpha}}{\sqrt{2}}\sigma^x\otimes \openone \ {\rm and}\  U^{m,n}=\left(\sigma^{x}\right)^{|n|+|m|}\otimes N^{n,m}\,,
\end{align}
for $\alpha\in \mathbb{R}$ and $N^{n,m}$ arbitrary, and with parity operator $Z=\sigma^z\otimes \openone$. Analogously, for antiperiodic boundary conditions, we say that $U^{(N)}$ is a generalized fMPU of the first kind if it can be cast in the form~\eqref{eq:generalized_fMPU} with $S=e^{i\alpha}\openone$, while it is of the second kind if 
\begin{align}
S&=\frac{e^{i\alpha}}{\sqrt{2}} \openone \ {\rm and}\  U^{m,n}=\left(\sigma^{x}\right)^{|n|+|m|}\otimes N^{n,m}\,,
\label{eq:second_kind_apb}
\end{align}
for $\alpha\in \mathbb{R}$ and $N^{n,m}$ arbitrary, and with parity operator $Z=\sigma^z\otimes \openone$. We note that, for fMPUs  of the second kind, the boundary tensors bear a normalization constant $1/\sqrt{2}$, which takes into account that degenerate normal tensors have two eigenvalues equal to one. 

We are now in a position to address the relation between fMPUs and QCA. At this point, it is convenient to distinguish between periodic and antiperiodic boundary conditions, since the emerging picture is quite different.

\subsection{FMPUs and QCA: antiperiodic boundary conditions}
\label{sec:antiperiodic_fMPUs}

We begin our discussion with the case of antiperiodic boundary conditions, where a quite strong statement holds: namely, all type I and type II fMPUs introduced in Sec.~\ref{sec:gfMPUs_two_kinds} are $1D$ fermionic QCA and viceversa. This section is devoted to establishing this equivalence. For clarity, we will report here only the main statements, while we refer the reader to Appendix~\ref{sec:technical_results_antiperiodic} for all the technical proofs. As a first step, it is useful to characterize the GCF for tensors generating type I and type II fMPUs.
\begin{prop}\label{prop:normality_fmpus}
Let $\mathcal{U}$ be in GCF, and suppose $\mathcal{U}$ generates a type I (type II) fMPU $U^{(N)}$ with antiperiodic boundary conditions. Then, $\mathcal{U}/\sqrt{d}$ is graded normal non-degenerate (degenerate).
\end{prop}
We report the proof in Appendix~\ref{sec:technical_results_antiperiodic}. This proposition is a direct generalization of the one given in Ref.~\cite{cirac2017matrix} for MPUs in qudit systems. In fact, one could push the analogy further, and show that, after blocking, any tensor $\mathcal{U}$ generating a type I or type II fMPU is \emph{simple}, where one also needs to generalize the notion of simpleness as follow.
\begin{defn}\label{def:simpleness_fermionic}
	We say that an even tensor  $\mathcal{U}$ with transfer matrix $E_{\mathcal{U}}$ is simple if 
	\begin{subequations}\label{eq:simpleness_fermionic}
	\bea
	\begin{tikzpicture}[baseline={([yshift=-0.5ex]current bounding box.center)}, scale=0.7]
	
	\draw (-0.5,1) arc (-270:-90:0.5);
	\draw (-1.5,1) arc (90:-90:0.5);
	\draw (1.5,1) arc (-270:-90:0.5);
	\draw (0.5,1) arc (90:-90:0.5);
	\draw (0,-0.5) node[left]{} -- (0,1.5) node[left]{};
	\draw (-0.5,0) node[left]{} -- (0.5,0) node[right]{};
	\draw (-0.5,1) node[left]{} -- (0.5,1) node[right]{};
	
	\filldraw[fill=white,draw=black] (0,0) circle [radius=0.25];
	\fill[black] (0,1) circle [radius=0.25];
	\draw [fill=white] (-1.25,0.25) rectangle (-0.75,0.75);
	\draw [fill=white] (0.75,0.25) rectangle (1.25,0.75);
	\draw(-1,-0.2)node{$E_{\mathcal{U}}$};
	\draw(1,-0.2)node{$E_{\mathcal{U}}$};
	\end{tikzpicture}
	=
	\begin{tikzpicture}[baseline={([yshift=-0.5ex]current bounding box.center)}, scale=0.7]
	\draw (-0.5,1) arc (-270:-90:0.5);
	\draw (-1.5,1) arc (90:-90:0.5);
	\draw (0,-0.5) node[left]{} -- (0,1.5) node[left]{};
	\draw[dashed] (-0.5,0) node[left]{} -- (0.5,0) node[right]{};
	\draw[dashed](-0.5,1) node[left]{} -- (0.5,1) node[right]{};
	\draw [fill=white] (-1.25,0.25) rectangle (-0.75,0.75);
	\draw(-1,-0.2)node{$E_{\mathcal{U}}$};
	\end{tikzpicture}\\
	\begin{tikzpicture}[baseline={([yshift=-0.5ex]current bounding box.center)}, scale=0.7]	
	\draw (0,-0.5) node[left]{} -- (0,1.5) node[left]{};
	\draw (-0.5,0) node[left]{} -- (1.75,0) node[right]{};
	\draw (-0.5,1) node[left]{} -- (1.75,1) node[right]{};
	\draw (1.25,-0.5) node[left]{} -- (1.25,1.5) node[left]{};
	\filldraw[fill=white,draw=black] (0,0) circle [radius=0.25];
	\fill[black] (0,1) circle [radius=0.25];
	\filldraw[fill=white,draw=black] (1.25,0) circle [radius=0.25];
	\fill[black] (1.25,1) circle [radius=0.25];
	\end{tikzpicture}
	=
	\begin{tikzpicture}[baseline={([yshift=-0.5ex]current bounding box.center)}, scale=0.7]	
	\draw (1.25,1) arc (-270:-90:0.5);
	\draw (0.05,1) arc (90:-90:0.5);
	\draw (0,-0.5) node[left]{} -- (0,1.5) node[left]{};
	\draw (-0.5,0) node[left]{} -- (0,0) node[right]{};
	\draw (1.25,0) node[left]{} -- (1.75,0) node[right]{};
	\draw (-0.5,1) node[left]{} -- (0,1) node[right]{};
	\draw (1.25,1) node[left]{} -- (1.75,1) node[right]{};
	\draw (1.25,-0.5) node[left]{} -- (1.25,1.5) node[left]{};
	\filldraw[fill=white,draw=black] (0,0) circle [radius=0.25];
	\fill[black] (0,1) circle [radius=0.25];
	\filldraw[fill=white,draw=black] (1.25,0) circle [radius=0.25];
	\fill[black] (1.25,1) circle [radius=0.25];
	\draw [fill=white] (0.4,0.25) rectangle (0.9,0.75);
	\draw(0.7,-0.2)node{$E_{\mathcal{U}}$};
	\end{tikzpicture}
	\eea
	\end{subequations}
	and $E_{\mathcal{U}}^2=E_{\mathcal{U}}$, where we denoted by a square the transfer matrix $E_{\mathcal{U}}$, and where we have used the graphical notation explained in Appendix.~\ref{sec:elementary_operators} for fermionic tensor networks. 
\end{defn}
Note that for type I fMPUs, this condition coincides with Eq.~\eqref{eq:simpleness_condition}, because, after blocking the transfer matrix reads $E=\ket{r}\bra{l}$, where $\ket{l}$, $\ket{r}$ are the left and right eigenvectors associated with the eigenvalue $1$. However, for type II fMPUs the simpleness condition is different, because the transfer matrix has two eigenvectors associated with $1$. In Appendix~\ref{sec:technical_results_antiperiodic}, we prove the following.

\begin{prop}
Suppose that the tensor $\mathcal{U}$ generates a type I (type II) fMPU $U^{(N)}$ with antiperiodic boundary conditions. Then, there exists $k\leq D^4$ such that $\mathcal{U}_k$ is simple (according to the definition~\ref{def:simpleness_fermionic}).
\end{prop}

Given a tensor $\mathcal{U}$, the simpleness condition~\ref{def:simpleness_fermionic} clearly implies that $U^{(N)}$ is locality-preserving. In order to establish the equivalence between fMPUs and fQCA, it remains to show that the latter can always be represented as a type I or type II fMPU. This is proven in Appendix~\ref{sec:technical_results_antiperiodic}, and we arrive at the main result of this section.

\begin{prop}\label{prop:fQCA_vs_fMPUs}
Up to appending inert ancillary fermionic degrees of freedom, any type I or type II fMPU with antiperiodic boundary conditions is a $1D$ fermionic QCA and viceversa.
\end{prop}

\subsection{FMPUs and QCA: periodic boundary conditions}
\label{sec:periodic_fMPUs}

As we have seen in Sec.~\ref{sec:nonlocal_fmpu}, fMPUs with periodic boundary conditions are not necessarily locality-preserving and in this case we need to impose further conditions on the tensors. An important point is that any tensor $\mathcal{U}$ (in GCF) generating a type I or type II fMPU with antiperiodic boundary conditions, also generates one in the periodic case. Indeed, let $\mathcal{U}$ be a non-degenerate normal tensor generating a type I fMPU $U^{(N)}_A$ with antiperiodic boundary conditions. Then, by blocking a finite number of times, we can assume that it is simple. Let $U^{(N)}_P$ be the fMPU with periodic boundary conditions, obtained by inserting the operator $Z$ into the trace. In order to show that $U^{(N)}_P$ is unitary, we can apply the simpleness condition to $U^{(N)\dagger}_PU^{(N)}_P$, and obtain
\be
U^{(N)\dagger}_PU^{(N)}_P=
\begin{tikzpicture}[baseline={([yshift=-0.5ex]current bounding box.center)}, scale=0.7]
\draw [fill=black] (-2.15,-0.15) rectangle (-1.85,0.15);
\draw [fill=black] (-2.15,1-0.15) rectangle (-1.85,1+0.15);
\draw (-0.5,1) arc (-270:-90:0.5);
\draw (-1.5,1) arc (90:-90:0.5);
\draw (-1.5,-0.5) node[left]{} -- (-1.5,1.5) node[left]{};\\
\draw (-2.4,0) node[left]{} -- (-1.5,0) node[right]{};
\draw (-2.4,1) node[left]{} -- (-1.5,1) node[right]{};
\draw (0,-0.5) node[left]{} -- (0,1.5) node[left]{};
\draw[dashed] (-0.5,0) node[left]{} -- (2.,0) node[right]{};
\draw[dashed](-0.5,1) node[left]{} -- (2,1) node[right]{};
\draw [fill=white] (-1.25,0.25) rectangle (-0.75,0.75);
\draw(-0.95,-0.2)node{$E_{\mathcal{U}}$};
\filldraw[fill=white,draw=black] (-1.5,0) circle [radius=0.25];
\fill[black] (-1.5,1) circle [radius=0.25];
\draw (1.4,-0.5) node[left]{} -- (1.4,1.5) node[left]{};\\
\draw (2,-0.5) node[left]{} -- (2,1.5) node[left]{};\\
\filldraw[fill=white,draw=black] (-1.5,0) circle [radius=0.25];
\fill[black] (-1.5,1) circle [radius=0.25];
\draw (2,0) node[left]{} -- (2.5,0) node[right]{};
\draw (2,1) node[left]{} -- (2.5,1) node[right]{};
\filldraw[fill=white,draw=black] (2.,0) circle [radius=0.25];
\fill[black] (2.,1) circle [radius=0.25];
\end{tikzpicture}\,,
\ee
where we denoted by a small black box the parity operator $Z$. Now, using that each tensor is even, and $Z^2=\openone$, we have
\be
\begin{tikzpicture}[baseline={([yshift=-0.5ex]current bounding box.center)}, scale=0.7]
\draw [fill=black] (-2.15,-0.15) rectangle (-1.85,0.15);
\draw [fill=black] (-2.15,1-0.15) rectangle (-1.85,1+0.15);
\draw (-1.5,-0.5) node[left]{} -- (-1.5,1.5) node[left]{};\\
\draw (-2.4,0) node[left]{} -- (-1.,0) node[right]{};
\draw (-2.4,1) node[left]{} -- (-1.,1) node[right]{};
\filldraw[fill=white,draw=black] (-1.5,0) circle [radius=0.25];
\fill[black] (-1.5,1) circle [radius=0.25];
\end{tikzpicture}
=
\begin{tikzpicture}[baseline={([yshift=-0.5ex]current bounding box.center)}, scale=0.7]
\draw [fill=black] (1+-2.15,-0.15) rectangle (1.+-1.85,0.15);
\draw [fill=black] (1.+-2.15,1-0.15) rectangle (1.+-1.85,1+0.15);
\draw [fill=black] (-1.5-0.15,1.5-0.15) rectangle (-1.5+0.15,1.5+0.15);
\draw [fill=black] (-1.5-0.15,-0.5-0.15) rectangle (-1.5+0.15,-0.5+0.15);
\draw (-1.5,-0.8) node[left]{} -- (-1.5,1.8) node[left]{};\\
\draw (-2.,0) node[left]{} -- (-0.6,0) node[right]{};
\draw (-2.,1) node[left]{} -- (-0.6,1) node[right]{};
\filldraw[fill=white,draw=black] (-1.5,0) circle [radius=0.25];
\fill[black] (-1.5,1) circle [radius=0.25];
\end{tikzpicture}\,.
\ee
Finally, as it is shown in Appendix~\ref{sec:technical_results_antiperiodic}, for a normal non-degenerate tensor, the right eigenvector associated with the eigenvalue $1$ of the transfer matrix $E$ is even, which implies $(Z\otimes Z )E_{\mathcal{U}}=E_{\mathcal{U}}$. Thus,  using that $Z^2=\openone$, we have $U^{(N)\dagger}_PU_P^{(N)}=\openone$. In a similar way, one can show that if $\mathcal{U}$ is a normal degenerate tensor generating a type II fMPU with antiperiodic boundary conditions, then it also generates a type II fMPU in the periodic case. To see this, one has to use that $X=\sigma^x\otimes \openone$ commutes with $U^{i,j}=(\sigma^{x})^{|i|+|j|}\otimes N^{i,j}$ for all $i,j$, and that $(X\otimes X )E_{\mathcal{U}}= E_{\mathcal{U}}$, which follows from the properties of the transfer matrix $E_{\mathcal{U}}$ associated with a graded normal degenerate tensor, cf. Appendix~\ref{sec:technical_results_antiperiodic}.

The above discussion, together with the results of the previous subsection, tells us that any fQCA can be represented as an fMPU with periodic boundary conditions. It also allows us to identify which properties we need to require in order for periodic fMPUs to be locality-preserving. In particular, if $U^{(N)}$ is a type I fMPU, then we require that the corresponding tensor $\mathcal{U}$ is simple, and that $(Z\otimes Z )E_{\mathcal{U}}=E_{\mathcal{U}}$, while if  $U^{(N)}$  is a type II fMPU we require that $(X\otimes X)E_{\mathcal{U}}= E_{\mathcal{U}}$ (in addition to the simpleness condition). With these additional constraints, we have finally an identification between fMPUs and fermionic QCA in the case of periodic boundary conditions. We can summarize the results of this section in the following theorem.

\begin{thm}
Up to appending inert ancillary fermionic degrees of freedom, any fermionic QCA can be represented as an fMPU of type I or type II with either periodic or antiperiodic boundary conditions. Furthermore\\[-0.5cm]
\begin{enumerate}
	\item any  type I or type II fMPU with ABC is necessarily an fQCA;\\[-0.5cm]
	\item any type I fMPU with PBC is an fQCA if it is generated by a simple tensor  $\mathcal{U}$,  and the transfer matrix $E_{\mathcal{U}}$  satisfies $(Z\otimes Z )E_{\mathcal{U}}=E_{\mathcal{U}}$, ${\rm tr}(E_{\mathcal{U}})=1$;\\[-0.5cm]
	\item any type II fMPU with PBC is an fQCA if it is generated by a simple tensor  $\mathcal{U}$,  and the transfer matrix $E_{\mathcal{U}}$  satisfies $(X\otimes X)E_{\mathcal{U}}=E_{\mathcal{U}}$, ${\rm tr}(E_{\mathcal{U}})=2$.
\end{enumerate}
\end{thm}

Note that the fMPU discussed in Sec.~\ref{sec:nonlocal_fmpu} is not simple, and, accordingly, is not locality-preserving.

\section{Index theory for generalized fMPUs}
\label{sec:index_theory}

In this section, we finally discuss how to extract an index to classify the locality-preserving fMPUs defined in Sec.~\ref{sec:gfMPUs_two_kinds}, thus recovering the results  recently derived in Ref.~\cite{fidkowski2019interacting}, where the GNVW index~\cite{gross2012index} was generalized to the case of fermionic systems. Our logic is very similar to the one employed in Ref.~\cite{cirac2017matrix}, but there are non-trivial practical differences, due once again to the fermionic nature of the elementary degrees of freedom.

First of all, we recall the definition of the index for qudits. Let $\mathcal{U}$ be a tensor generating an MPU, and denote by $U^{n,m}_{\alpha,\beta}$ the corresponding elements. One can define two maps $\mathcal{M}_{1,2}:\mathbb{C}^{d}\otimes \mathbb{C}^{D}\to \mathbb{C}^{d}\otimes \mathbb{C}^{D}$, by
\begin{subequations}\label{eq:bosonic_maps}
\begin{align}
\mathcal{M}_{1}&: \ket{m}\otimes \ket{\alpha}\mapsto U^{n,m}_{\alpha,\beta}\ket{n}\otimes \ket{\beta}\,,\\
\mathcal{M}_{2}&: \ket{m}\otimes \ket{\beta}\mapsto U^{n,m}_{\alpha,\beta}\ket{n}\otimes \ket{\alpha}\,.
\end{align}
\end{subequations}
Denoting the rank of these maps by $r$, and $\ell $, respectively, the MPU index was defined in Ref.~\cite{cirac2017matrix} as ${\rm ind}=(1/2)[\log_2(r)-\log_2(\ell)]$, where it was also shown to coincide with the one first introduced in Ref.~\cite{gross2012index}.

Let us  now consider the fermionic case. First of all, we need to restrict to the class of locality-preserving fMPUs, namely we consider tensors $\mathcal{U}$ that become simple after blocking sufficiently many times. Now, if $\mathcal{U}$ is an even tensor, with elements denoted again by $U^{n,m}_{\alpha,\beta}$, generating a type I or type II fMPU, one could be tempted to define the index as for qudits, in terms of the maps $\mathcal{M}_{1,2}$ introduced above. However, this turns out \emph{not} to be a valid definition for the fermionic index. In order to see it, consider a depth-two quantum circuit, with elementary two-site gate defined by
\be
U_{j,j+1}=\frac{1}{\sqrt{2}}\left(\openone-Y_jX_{j+1}\right)\,,
\ee
with $Y_j=-i(a_j-a_j^{\dagger})$, $X_j=(a_j+a_j^{\dagger})$. It is straightforward to construct the corresponding fMPU, which is characterized by local and physical dimension $d=D=2$. The associated tensor $\mathcal{U}$ is normal non-degenerate and simple. However, one can see that a definition based on Eq.~\eqref{eq:bosonic_maps} would yield a non-vanishing index, which is the wrong result for a quantum circuit (cf. Ref.~\cite{fidkowski2019interacting}). Furthermore, we note that blocking does not remedy to this problem~\footnote{Importantly, we recall that in order to obtain the blocked tensors of fMPOs, one needs to introduce additional signs, as specified by Eq.~\eqref{eq:signs_blocking}. These signs are crucial, since in general they modify the ranks of the operators defined when discussing the fermionic index.}.

 In fact, it turns out that the construction for qudits must be modified by introducing additional signs in the definition of the maps $\mathcal{M}_{1,2}$. In particular, in the fermionic case, we define
\begin{subequations}\label{eq:f_m_maps}
	\begin{align}
		\mathcal{M}^f_{1}&: \ket{m}\otimes \ket{\alpha}\mapsto (-1)^{|n||m|} U^{n,m}_{\alpha,\beta}\ket{n}\otimes \ket{\beta}\,,\\
		\mathcal{M}^f_{2}&: \ket{m}\otimes \ket{\beta}\mapsto U^{n,m}_{\alpha,\beta}\ket{n}\otimes \ket{\alpha}\,.
	\end{align}
\end{subequations}
where the label $f$, specifying that we are dealing with the fermionic case, will be omitted when it does not generate confusion. We are now in a position to define the fermionic index.
\begin{definition}
	Let $\mathcal{U}$ be a tensor in GCF generating a type I or type II locality-preserving fMPU. Take $k$ such that the blocked tensor $\mathcal{U}_k$ is simple, and define $r$, $\ell$ as the ranks of the maps $\mathcal{M}^f_{1}$, $\mathcal{M}^f_{2}$ in Eq.~\eqref{eq:f_m_maps}, respectively. The (fermionic) index of $\mathcal{U}$ is defined as ${\rm ind}_f=(1/2)(\log_2(r)-\log_2(\ell))$, while we denote the exponentiated index by $\mathcal{I}_f=\sqrt{r/\ell}$.
\end{definition}
For clarity, we first state our main result on the fermionic index, which will be proven in the rest of this technical section.

\begin{thm}[Fermionic Index Theorem]\label{th:femionic_index} Let $\mathcal{U}$ be an even tensor in GCF generating a locality-preserving fMPUs. Then\\[-0.5cm]
	\begin{enumerate}
		\item the exponentiated index $\mathcal{I}_f$ is a rational number for type I fMPUs, while $\mathcal{I}_f=\sqrt{2}(p/q)$ for type II fMPUs, with $p,q\in \mathbb{N}$ and coprime;\\[-0.5cm]
		\item the index does not change by blocking;\\[-0.5cm]
		\item the index is additive by tensoring and composition; \\[-0.5cm]
		\item the index is robust, i.e. by changing continuously $\mathcal{U}$ (and remaining in the class of locality-preserving fMPUs) one cannot change it;\\[-0.5cm]
		\item two tensors have the same index iff they are equivalent.
	\end{enumerate}
\end{thm}
Here we have introduced the notion of equivalence, which will be defined in a precise way later on. Loosely speaking, two tensors are equivalent if they can be transformed continuously into one another, by blocking and attaching ancillas, respectively. The proof of Theorem~\ref{th:femionic_index} is carried out in the rest of this section, which represents the most technical part of our work. While one could carry out the discussion using the fiducial-state formalism employed so far, the notation becomes significantly simpler exploiting the language of graded tensor networks, recently developed in Refs.~\cite{bultinck2017fermionic, kapustin_spin_2018}. The latter is reviewed in Appendix~\ref{sec:graded_TN}, where it is shown to be in $1$-to-$1$ correspondence with the fiducial-state formalism. 

Finally, from now on we will focus on the case of antiperiodic boundary conditions, where any generalized fMPU is guaranteed to be locality-preserving. On the technical level, this allows us to exploit the uniqueness of the graded canonical form (cf. Appendix~\ref{sec:graded_canonical_form}), which is needed in order to carry out some proofs, and to exclude a priori the possibility of non-locality-preserving fMPUs. We stress, however, that this is not a restriction, because we have shown in the previous section that tensors generating locality-preserving fMPUs in the periodic case, are also fMPUs with antiperiodic boundary conditions. 

\subsection{Standard form for type I fMPUs}
\label{sec:standard_form_type_I}

We begin by showing that any type I fMPU admits a standard form, analogous to the one of qudit MPUs defined in Ref.~\cite{cirac2017matrix}. Unless specified otherwise, in the rest of this section we will assume that $\mathcal{U}/{\sqrt{d}}$ is graded normal non-degenerate, where the left and right eigenvectors of the transfer matrix $E_{\mathcal{U}}$ corresponding to the eigenvalue $1$, $(\Phi|$ and $|\rho)$, are of the form
\begin{subequations}
\begin{align}
\langle \Phi|&=\sum_{n=1}^D(n,n|\,,\\
|\rho)&=\sum_{n=1}^D\rho_n|n,n)\,.\label{eq:leading_right_eigen}
\end{align}
\end{subequations}
This can be done without loss of generality, as it follows from the results in Appendices~\ref{sec:graded_canonical_form} and \ref{sec:technical_results_antiperiodic}.

 We recall that, in the graded tensor network formalism (cf. Appendix~\ref{sec:graded_TN}), a given local tensor $\mathcal{U}$ is represented as
\be
\mathcal{U}=\sum_{n,m, \alpha, \beta} U_{\alpha, \beta}^{n,m} | \alpha ) |n\rangle \langle m|(\beta | =
\begin{tikzpicture}[baseline={([yshift=-1.5ex]current bounding box.center)}, scale=0.7]

\draw (0,-0.7) node[left]{$m$} -- (0,0.7) node[left]{$n$};
\draw (-0.7,0) node[left]{$\alpha$} -- (0.7,0) node[right]{$\beta$};	
\filldraw[fill=white,draw=black] (0,0) circle [radius=0.25];

\end{tikzpicture}\,,
\label{eq:graded_u_tensor}
\ee
where $| \alpha ) |n\rangle \langle m|(\beta |$ is a short-hand notation for $| \alpha )\gtimes |n\rangle\gtimes \langle m|\gtimes(\beta |$, and $\gtimes$ is the graded tensor product. Analogously, the tensor $\bar{\mathcal{U}}$ generating the conjugate transposed operator, reads 
\begin{align}
\bar{\mathcal{U}}=\sum_{n,m, \alpha, \beta}(-1)^{|\beta|+|\alpha||\beta|}(\bar{U}^{n,m}_{\alpha,\beta})|\alpha)|n\rangle\langle m|(\beta|\,.
\end{align}
By simply using the contraction rules for graded tensor networks explained in  Appendix~\ref{sec:graded_TN}, one could immediately write down the transfer operator
\bea
\mathbb{E}_{\mathcal{U}}=\frac{1}{d}\sum_{j,k} \bar{U}^{ j,k}_{\gamma,\delta}U^{j,k}_{\alpha,\beta} |\alpha)|\gamma)(\delta|(\beta|\,,
\label{eq:fTO}
\eea
where repeated indices are summed over. Note that the order of the bra and ket vectors in the rhs is important. 

Following Ref.~\cite{cirac2017matrix}, we now consider two different singular value decompositions of $\mathcal{U}$. Since this requires us to rearrange its indices, this procedure introduces additional non-trivial signs wrt the qudit case, as we now explain. Our goal is to rewrite
\be\label{eq:svd_graded}
\begin{tikzpicture}[baseline={([yshift=-1.5ex]current bounding box.center)}, scale=0.7]

\draw (0,-0.7) node[left]{$m$} -- (0,0.7) node[left]{$n$};
\draw (-0.7,0) node[left]{$\alpha$} -- (0.7,0) node[right]{$\beta$};	
\filldraw[fill=white,draw=black] (0,0) circle [radius=0.25];

\end{tikzpicture}
=
\begin{tikzpicture}[baseline={([yshift=-1.ex]current bounding box.center)}, scale=0.6]

\draw (0,-1) node[left]{$m$} -- (0,0.) ;
\draw (-0.-1,0) node[left]{$\alpha$} -- (0.,0) ;
\draw (-0.5,0)  -- (0,-0.5) ;
\draw[dotted, line width=0.5mm] (0,0)  -- (0,0.75) ;
\draw (0,1.25)  -- (1,1.25) node[right]{$\beta$};
\draw (0.0,0.75)  -- (0.,1.75) node[left]{$n$} ;
\draw (0,0.75)  -- (0.5,1.25) ;
\draw (0,0.35)node[left]{$r$};
\end{tikzpicture}
=
\begin{tikzpicture}[baseline={([yshift=-1.ex]current bounding box.center)}, scale=0.6]

\draw (0,-1)node[left]{$m$}  -- (0,0.);
\draw (0,0)  -- (1.,0) node[right]{$\beta$};
\draw[fill=gray] (0.5,0) -- (0,-0.5) -- (0,0) -- cycle;
\draw[line width=0.5mm] (0,0)  -- (0,0.75) ;
\draw (0,1.25)  -- (-1,1.25) node[left]{$\alpha$};
\draw (0.0,0.75)  -- (0.,1.75) node[right]{$n$};
\draw[fill=gray] (0.,0.75) -- (-0.5,1.25) -- (0,1.25) -- cycle;
\draw (0,0.35)node[right]{$\ell$};
\end{tikzpicture}\,.
\ee
Importantly, every diagram here corresponds to an \emph{even} graded tensor. Explicitly, we define
\begin{subequations}
\be\label{eq:x1}
\mathcal{X}_1=
\begin{tikzpicture}[baseline={([yshift=-1.ex]current bounding box.center)}, scale=0.6]
\draw (0,1.25)  -- (1,1.25) node[right]{$\beta$};
\draw (0.0,0.75)  -- (0.,1.75) node[left]{$n$} ;
\draw (0,0.75)  -- (0.5,1.25) ;
\draw[dotted, line width=0.5mm] (0,0)node[left]{$i$}   -- (0,0.75) ;
 \end{tikzpicture}=X_1^{n,\beta,i}|n\rangle (\beta|\langle i|\,,
\ee
\be\label{eq:y1}
\mathcal{Y}_1=
\begin{tikzpicture}[baseline={([yshift=-0.6ex]current bounding box.center)}, scale=0.6]
\draw[dotted, line width=0.5mm] (0,0)  -- (0,0.75) node[left]{$i$} ;
\draw (0,-1) node[left]{$m$} -- (0,0.) ;
\draw (-0.-1,0) node[left]{$\alpha$} -- (0.,0) ;
\draw (-0.5,0)  -- (0,-0.5) ;
\end{tikzpicture}=Y^{i,\alpha,m}_{1}|i\rangle|\alpha)\langle m|\,,
\ee
\be
\mathcal{X}_2=
\begin{tikzpicture}[baseline={([yshift=-1.ex]current bounding box.center)}, scale=0.6]
\draw (0,1.25)  -- (-1,1.25) node[left]{$\alpha$};
\draw (0.0,0.75)node[right]{$i$}  -- (0.,1.75) node[right]{$n$};
\draw[line width=0.5mm]  (0.0,0.25)  -- (0.,0.75);
\draw[fill=gray] (0.,0.75) -- (-0.5,1.25) -- (0,1.25) -- cycle;
\end{tikzpicture}=X^{\alpha,n,i}_{2}|\alpha)|n\rangle\langle i|\,,
\ee
\be
\mathcal{Y}_2=
\begin{tikzpicture}[baseline={([yshift=-1.ex]current bounding box.center)}, scale=0.6]
\draw (0,-1)node[left]{$m$}  -- (0,0.);
\draw (0,0)  -- (1.,0) node[right]{$\beta$};
\draw[fill=gray] (0.5,0) -- (0,-0.5) -- (0,0) -- cycle;
\draw[line width=0.5mm]  (0.0,0.)  -- (0.,0.75)node[left]{$i$} ;
\end{tikzpicture}=Y^{i,m,\beta}_{2}|i)\langle m|(\beta|\,.
\ee
\end{subequations}
We show that we can always decompose $\mathcal{U}$ as in Eq.~\eqref{eq:svd_graded}, with $\mathcal{X}_{i}$, $\mathcal{Y}_i$ even tensors. First, we note that $\mathcal{U}$ defines the linear function
\be
|\beta)\gtimes \ket{m}\mapsto \mathcal{C}(\mathcal{U}|\beta)\gtimes \ket{m})=U^{n,m}_{\alpha,\beta}|\alpha)\gtimes\ket{n}\,,
\ee
where $\mathcal{C}$ is the contraction map introduced in Appendix~\ref{sec:graded_TN}.  This map can be represented as a matrix, $\mathcal{M}^f_2$, with elements $U^{n,m}_{\alpha,\beta}$. Using a singular value decomposition, we can write
\be
\mathcal{M}^f_2=V^\dagger_{2} D_2 W_2\,,
\ee
where $V_2$, $W_2$ are isometries, fulfilling $V_2V^\dagger_2=W_2W^\dagger_2=\openone$, while $D_2$ is a diagonal positive matrix of dimension $\ell$. Crucially, we can choose these matrices to be even, when seen as operators acting on graded spaces. Indeed, let us consider the matrix $\mathcal{M}^f_2$ in the basis associated with the elements $U^{n,m}_{\alpha,\beta}$. We reorder the basis vectors, in such a way that even vectors come first. Since $\mathcal{U}$ is even, $\mathcal{M}^f_2$ is now block diagonal, with two blocks corresponding to the even and odd subspaces, respectively. We can then apply a singular value decomposition to each block individually, yielding isometries $V_2^{e,o}$, $W_2^{e,o}$ and diagonal matrices $D^{e,o}_2$. The desired decomposition in terms of even matrices is simply obtained choosing $V_2=V_2^{e}\oplus V_2^{o}$, $W_2=W_2^{e}\oplus W_2^{o}$ and $D_2=D_2^e\oplus D_2^{o}$.  Next, defining the matrices $X_2=V^\dagger_2$ and $Y_2=D_2W_2$, we have 
\be
\mathcal{U}=\mathcal{X}_2 \mathcal{Y}_2\,,
\ee
where appropriate contractions are implied. Since the matrices $X_2$, $Y_2$ are even, $\mathcal{X}_2$, $\mathcal{Y}_2$ are even graded tensors, and the second decomposition in Eq.~\eqref{eq:svd_graded} is established. 

The first decomposition in Eq.~\eqref{eq:svd_graded} can be derived in a similar way, but we need to be careful with the signs arising from reordering the graded tensors. Rewriting
\be\label{eq:map_f2}
\mathcal{U}=(-1)^{|n||m|+|\alpha|+|m|}U^{n,m}_{\alpha,\beta}\ket{n}( \beta|\gtimes |\alpha)\bra{m}\,,
\ee
we see that, in order to arrive at a decomposition of the form $\mathcal{U}=\mathcal{X}_1\mathcal{Y}_1$, with $\mathcal{X}_1$, $\mathcal{Y}_1$ as in Eqs.~\eqref{eq:x1}, \eqref{eq:y1}, we need to apply a singular value decomposition to the matrix $\mathcal{M}^f_1$ with elements $(-1)^{|n||m|+|\alpha|+|m|}U^{n,m}_{\alpha,\beta}$. Once again, this can be done by choosing even matrices, so that $\mathcal{M}^f_1=V^{\dagger}_1D_1W_1$, where $V_1$, $W_1$ are isometries, and $D_1$ is a diagonal positive matrix of dimension $r$. Now we chose $X_1=V_1^\dagger\left[V_1\left(\openone\otimes \rho\right)V_1^\dagger\right]^{-1/2}$, so that
\be
X_1^\dagger\left(\openone\otimes \rho \right)X_1=\openone\,,
\label{eq:x_1_condition}
\ee
where $\rho$ is the positive matrix corresponding to the right eigenstate \eqref{eq:leading_right_eigen}. Since all the matrices involved are even, $X_1$ is also even, and similarly for $Y_1=\left[V_1\left(\openone\otimes \rho\right)V_1^\dagger\right]^{1/2}D_1W_1$. Thus, we have established the first decomposition in Eq.~\eqref{eq:svd_graded}.

Importantly, the dimensions of the diagonal matrices $D_{1,2}$ introduced above, $r$ and $\ell$, coincide, by construction, with the ranks of the maps $\mathcal{M}^{f}_{1,2}$ in Eq.~\eqref{eq:f_m_maps}. For the map in Eq.~\eqref{eq:map_f2}, this follows from the fact that the factor $(-1)^{|\alpha|+|m|}$ clearly does not change the rank $r$. Next, in terms of the previous decomposition, we introduce
\begin{subequations}\label{eq:u_v}
\begin{align}
u=
\begin{tikzpicture}[baseline={([yshift=-1.ex]current bounding box.center)}, scale=0.6]
\draw (0,-1)  -- (0,0.);
\draw (0,0)  -- (1.,0) ;
\draw[fill=gray] (0.5,0) -- (0,-0.5) -- (0,0) -- cycle;
\draw[line width=0.5mm]  (0.0,0.)  -- (0.,0.75);
\draw[dotted, line width=0.5mm] (+1.5+0,0)  -- (+1.5+0,0.75);
\draw (+1.5+0,-1)  -- (+1.5+0,0.) ;
\draw (+1.5+-0.-1,0)  -- (+1.5,0) ;
\draw (+1.5+-0.5,0)  -- (+1.5+0,-0.5) ;
\end{tikzpicture}\,,\\
v=
\begin{tikzpicture}[baseline={([yshift=-1.ex]current bounding box.center)}, scale=0.6]
\draw (0,1.25)  -- (1,1.25);
\draw (0.0,0.75)  -- (0.,1.75);
\draw (0,0.75)  -- (0.5,1.25) ;
\draw[dotted, line width=0.5mm] (0,0.2)  -- (0,0.75) ;
\draw (+1.5+0,1.25)  -- (+1.5+-1,1.25);
\draw (+1.5+0.0,0.75)  -- (+1.5+0.,1.75);
\draw[line width=0.5mm]  (+1.5+0.0,0.2)  -- (+1.5+0.,0.75);
\draw[fill=gray] (+1.5+0.,0.75) -- (+1.5+-0.5,1.25) -- (+1.5+0,1.25) -- cycle;
\end{tikzpicture}\,,
\end{align}
\end{subequations}
where as usual joined legs denote contraction of the corresponding graded spaces. Following Refs.~\cite{cirac2017matrix}, we can now prove two important statements

\begin{lem}
	\label{lemuisometry}
	For any tensor $\mathcal{U}$, $u^\dagger u=\openone$ and hence $r\ell \ge d^2$
\end{lem}
\begin{proof}
The proof follows the same steps as in Lemma III.7 in Ref.~\cite{cirac2017matrix}. We present it for completeness here, to show how to deal with the additional fermionic signs.  First, we have
\be
\begin{tikzpicture}[baseline={([yshift=-0.5ex]current bounding box.center)}, scale=0.7]	
\draw (0,-0.5) node[left]{} -- (0,1.5) node[left]{};
\draw (-0.5,0) node[left]{} -- (1.75,0) node[right]{};
\draw (-0.5,1) node[left]{} -- (1.75,1) node[right]{};
\draw (1.25,-0.5) node[left]{} -- (1.25,1.5) node[left]{};
\filldraw[fill=white,draw=black] (0,0) circle [radius=0.25];
\fill[black] (0,1) circle [radius=0.25];
\filldraw[fill=white,draw=black] (1.25,0) circle [radius=0.25];
\fill[black] (1.25,1) circle [radius=0.25];
\draw (3.2-1.5,0) arc (-90:90:0.5) node[right]{};
\draw (4.8+0.2-5.5,1) arc (90:270:0.5) node[right]{};
\draw [fill=white] (3.5-1.5,0.28) rectangle (3.9-1.5,0.72);
\draw (3.7-1.5,0.5) node{$\rho$};
\end{tikzpicture}
=
\begin{tikzpicture}[baseline={([yshift=-0.5ex]current bounding box.center)}, scale=0.7]	
\draw (0,-0.5) node[left]{} -- (0,1.5) node[left]{};
\draw (1.,-0.5) node[left]{} -- (1.,1.5) node[left]{};
\end{tikzpicture}\,,
\ee
which follows from the fact that $U^{(N)}$ is unitary for all $N$. Given the order of bra and ket states in Eq.~\eqref{eq:fTO}, the left eigenstate is
\be
\begin{tikzpicture}[baseline={([yshift=-0.5ex]current bounding box.center)}, scale=0.7]	
\draw (4.8+0.2-5.5,1) arc (90:270:0.5) node[right]{};
\draw (0,1) node{$\alpha$};
\draw (0,0) node{$\beta$};
\end{tikzpicture}
=\sum_{\alpha,\beta}\delta_{\alpha,\beta}  \bra{\alpha}\gtimes\bra{\beta}\,,
\ee
while the right eigenstate is
\be
\begin{tikzpicture}[baseline={([yshift=-0.5ex]current bounding box.center)}, scale=0.7]	
\draw (3.2-1.5,0) arc (-90:90:0.5) node[right]{};
\draw [fill=white] (3.5-1.5,0.28) rectangle (3.9-1.5,0.72);
\draw (3.7-1.5,0.5) node{$\rho$};
\draw (1.25,1) node{$\alpha$};
\draw (1.25,0) node{$\beta$};
\end{tikzpicture}
=\sum_{\alpha,\beta}\delta_{\alpha,\beta}  \rho_{\alpha}\ket{\beta}\gtimes\ket{\alpha}\,,
\ee
where, once again, the order of the bra and ket states is important. With these definitions, we can perform explicitly the contractions in the lhs, keeping track of the signs as we move ket and bra states close together. Note that this procedure is well-defined because  all individual tensors have well-defined parity. Using $X_2^\dagger X_2=\openone$ and Eq.~\eqref{eq:x_1_condition}, a direct calculation shows that the lhs is simply $u^\dagger u$, where $u$ is the graded operator defined in Eq.~\eqref{eq:u_v}. Thus, $u^\dagger u=\openone$,  and the statement is proved. 
\end{proof}

Note that this lemma implies that $u$ is an isometry. Next, we can also prove the following proposition.
\begin{prop}
	\label{ThmFund1}
	The following are equivalent for an even tensor ${\cal U}$ generating an fMPU.
	\begin{enumerate}
		\item ${\cal U}$ is simple.\\[-0.5cm]
		\item $r \ell =d^2$.\\[-0.5cm]
		\item $u$ is unitary.\\[-0.5cm]
		\item $v$ is unitary.\\[-0.5cm]
	\end{enumerate}
\end{prop}
We omit the proof of this statement, because it follows without modifications the one presented in Ref.~\cite{cirac2017matrix}. Note that this result allows us introduce a standard form for type I fMPUs. In particular, blocking a simple tensor twice, we have
\be\label{eq:standard_form}
\begin{tikzpicture}[baseline={([yshift=-1.5ex]current bounding box.center)}, scale=0.7]

\draw (0,-0.7)  -- (0,0.7);
\draw (-0.7,0)  -- (0.7,0);	
\filldraw[fill=white,draw=black] (0,0) circle [radius=0.25];

\end{tikzpicture}
=
\begin{tikzpicture}[baseline={([yshift=-1.ex]current bounding box.center)}, scale=0.6]
\draw (+1+0,1.25)  -- (+1+1,1.25);
\draw (+1+0.0,0.75)  -- (+1+0.,1.75);
\draw (+1+0,0.75)  -- (+1+0.5,1.25) ;
\draw[dotted, line width=0.5mm] (+1+0,0.2)  -- (+1+0,0.75) ;
\draw (0,1.25)  -- (-1,1.25);
\draw (+0.0,0.75)  -- (+0.,1.75);
\draw[line width=0.5mm]  (+0.0,0.2)  -- (+0.,0.75);
\draw[fill=gray] (+0.,0.75) -- (-0.5,1.25) -- (+0,1.25) -- cycle;
\draw (0,0)  -- (0,-1);
\draw (1,0)  -- (1,-1);
\draw[fill=white] (-0.25,-0.6) rectangle (1.25,0.4);
\draw (0.5,-0.1) node{$u$};
\end{tikzpicture}\,,
\ee
where $u$ is unitary, and we used the notation introduced in Eq.~\eqref{eq:u_v} for the unitary $v$. As in the case of qudits~\cite{cirac2017matrix}, the standard form is essentially unique, up to single-site unitary invertible matrices acting on the legs defining $u$ or $v$.

\begin{prop}\label{prop:uniqueness_type_II}
Given two simple tensors, ${\cal U}$ and ${\cal W}$ with standard forms
\be
\mathcal{U}=\begin{tikzpicture}[baseline={([yshift=-1.ex]current bounding box.center)}, scale=0.6]
\draw (+1+0,1.25)  -- (+1+1,1.25);
\draw (+1+0.0,0.75)  -- (+1+0.,1.75);
\draw (+1+0,0.75)  -- (+1+0.5,1.25) ;
\draw[dotted, line width=0.5mm] (+1+0,0.2)  -- (+1+0,0.75) ;
\draw (0,1.25)  -- (-1,1.25);
\draw (+0.0,0.75)  -- (+0.,1.75);
\draw[line width=0.5mm]  (+0.0,0.2)  -- (+0.,0.75);
\draw[fill=gray] (+0.,0.75) -- (-0.5,1.25) -- (+0,1.25) -- cycle;
\draw (0,0)  -- (0,-1);
\draw (1,0)  -- (1,-1);
\draw[fill=white] (-0.25,-0.6) rectangle (1.25,0.4);
\draw (0.5,-0.1) node{$u$};
\end{tikzpicture}\,,
\qquad 
\mathcal{W}=
\begin{tikzpicture}[baseline={([yshift=-1.ex]current bounding box.center)}, scale=0.6]
\draw (+1+0,1.25)  -- (+1+1,1.25);
\draw (+1+0.0,0.75)  -- (+1+0.,1.75);
\draw (+1,0.75)  arc (-90:0:0.5);
\draw[dotted, line width=0.5mm] (+1+0,0.2)  -- (+1+0,0.75) ;
\draw (0,1.25)  -- (-1,1.25);
\draw (+0.0,0.75)  -- (+0.,1.75);
\draw[line width=0.5mm]  (+0.0,0.2)  -- (+0.,0.75);
\draw[fill=gray] (+0.,0.75)  arc (-90:-180:0.5) -- (+0,1.25) -- cycle;
\draw (0,0)  -- (0,-1);
\draw (1,0)  -- (1,-1);
\draw[fill=white] (-0.25,-0.6) rectangle (1.25,0.4);
\draw (0.5,-0.1) node{$w$};
\end{tikzpicture}\,,
\ee
they generate the same type I fMPU, i.e.\
$U^{(N)}=V^{(N)}$ for all $N$, iff there exist even
unitaries $x,y$ and $z$, such that 
\begin{subequations}
\be
\begin{tikzpicture}[baseline={([yshift=-1.ex]current bounding box.center)}, scale=0.6]
\draw[line width=0.5mm]  (+0.0,0.2)  -- (+0.,1);
\draw[dotted, line width=0.5mm] (+1+0,0.2)  -- (+1+0,1) ;
\draw (0,0)  -- (0,-1);
\draw (1,0)  -- (1,-1);
\draw[fill=white] (-0.25,-0.6) rectangle (1.25,0.4);
\draw (0.5,-0.1) node{$w$};
\end{tikzpicture}
=
\begin{tikzpicture}[baseline={([yshift=-2.ex]current bounding box.center)}, scale=0.6]
\draw[line width=0.5mm]  (+0.0,0.2)  -- (+0.,1.5);
\draw[dotted, line width=0.5mm] (+1+0,0.2)  -- (+1+0,1.5) ;
\draw (0,0)  -- (0,-1);
\draw (1,0)  -- (1,-1);
\draw[fill=white] (-0.25,-0.6) rectangle (1.25,0.4);
\draw (0.5,-0.1) node{$u$};
\draw[fill=white] (-0.25,0.75) rectangle (0.25,1.25);
\draw[fill=white] (0.75,0.75) rectangle (1.25,1.25);
\draw (0.,1.) node{$x$};
\draw (1.,1.) node{$y$};
\end{tikzpicture}\,,
\ee
\be
\begin{tikzpicture}[baseline={([yshift=-1.ex]current bounding box.center)}, scale=0.6]
\draw (0,1.25)  -- (-1,1.25);
\draw (+0.0,0.75)  -- (+0.,1.75);
\draw[line width=0.5mm]  (+0.0,0.2)  -- (+0.,0.75);
\draw[fill=gray] (+0.,0.75)  arc (-90:-180:0.5) -- (+0,1.25) -- cycle;
\end{tikzpicture}
=
\begin{tikzpicture}[baseline={([yshift=-0.ex]current bounding box.center)}, scale=0.6]
\draw (0,1.25)  -- (-1.5,1.25);
\draw (+0.0,0.75)  -- (+0.,1.75);
\draw[line width=0.5mm]  (+0.0,-0.5)  -- (+0.,0.75);
\draw[fill=gray] (+0.,0.75) -- (-0.5,1.25) -- (+0,1.25) -- cycle;
\draw[fill=white] (-1.2,+1) rectangle (-0.7,+1.5);
\draw[fill=white] (-0.35,+0.2-0.35) rectangle (0.35,+0.2+0.35);
\draw (0.,0.2) node{$x^\dagger$};
\draw (-0.95,1.25) node{$z$};
\end{tikzpicture}\,,
\qquad \qquad 
\begin{tikzpicture}[baseline={([yshift=-1.ex]current bounding box.center)}, scale=0.6]
\draw (+1+0,1.25)  -- (+1+1,1.25);
\draw (+1+0.0,0.75)  -- (+1+0.,1.75);
\draw (+1,0.75)  arc (-90:0:0.5);
\draw[dotted, line width=0.5mm] (+1+0,0.2)  -- (+1+0,0.75) ;
\end{tikzpicture}
=
\begin{tikzpicture}[baseline={([yshift=-0.ex]current bounding box.center)}, scale=0.6]
\draw (+1+0,1.25)  -- (+2.5,1.25);
\draw (+1+0.0,0.75)  -- (+1+0.,1.75);
\draw (+1+0,0.75)  -- (+1+0.5,1.25) ;
\draw[dotted, line width=0.5mm] (+1+0,-0.5)  -- (+1+0,0.75) ;
\draw[fill=white] (1.7,0.9) rectangle (2.3,1.6);
\draw[fill=white] (0.65,-0.25) rectangle (1.35,0.45);
\draw (2,1.25) node{$z^\dagger$};
\draw (1.1,0.1) node{$y^\dagger$};
\end{tikzpicture}\,.
\ee
\end{subequations}
\end{prop}

\subsection{Standard form for type II fMPUs}
\label{sec:standard_form_type_II}

Let us now consider a type II fMPU $U^{(N)}$, generated by a graded degenerate normal tensor $\mathcal{U}$. It follows from Corollary~\ref{cor:local_dimension_degenerate} in Appendix \ref{sec:technical_results_antiperiodic} that the local graded space $\mathcal{H}_j$ has even dimension $d$, with even and odd subspaces of the same dimension, $d_e=d_o=d/2$. Hence, we can assume without loss of generality that  $\mathcal{H}_j\simeq \mathcal{H}^{1}_j\gtimes \mathcal{H}^{2}_{j}$, with ${\rm dim}(\mathcal{H}^{2}_{j})=d/2$ while $\mathcal{H}^1_j\simeq \mathbb{C}^{1|1}$, namely $\mathcal{H}^1_j$ is isomorphic to the two-dimensional complex coordinate space with parity operator $\mathcal{P}={\rm diag}(1,-1)$.  We can then introduce a type II fMPU $M^{(N)}_A$ implementing a Majorana translation [defined in Eq.~\eqref{eq:maj_anti}] on the local spaces $\mathcal{H}^{1}_{j}$, and acting as the identity on $\mathcal{H}^{2}_{j}$. Finally, we define 
\be
\tilde{U}^{(N)}=U^{(N)}M^{(N)\dagger}_A\,,
\ee
which is clearly a unitary operator for all $N$. Now, $\tilde{U}^{(N)}$ is the product of two type II fMPUs, so that it follows from Prop.~\ref{prop:z2_structure} that it can be represented as a type I fMPU, generated by a graded normal non-degenerate tensor $\tilde{\mathcal{U}}$ . Inverting the above relation, we get the following form for type II fMPUs
\be
U^{(N)}=\tilde{U}^{(N)} M^{(N)}_A\,.
\label{eq:standard_form_type_II}
\ee
Now by blocking sufficiently many times, the tensor $\tilde{\mathcal{U}}$ becomes simple and we can cast it in standard form. This gives us a standard form for the tensor $\mathcal{U}$ itself and, due to Prop.~\ref{prop:uniqueness_type_II}, this is also unique up to single-site (even) invertible matrices. Note, however, that  the tensor $\mathcal{U}$ obtained by composing $\tilde{\mathcal{U}}$ and $\mathcal{M}$, corresponding to the Majorana shift, is not necessarily in GCF.

\subsection{Index}

Based on the constructions carried out so far, we are now in a position to prove Theorem~\ref{th:femionic_index}, which is the main result of this section. Before this, however, we need to state precisely the definition of equivalent graded tensors. Let us consider two even tensors, ${\cal U}$ and ${\cal V}$ generating fMPUs, of physical dimensions $d_{a,b}$ (with even/odd subspaces $d^{e/o}_{a,b}$), and let us denote by $p_{a,b}$ two coprimes such that $p_a d_a=p_b d_b$. We also denote by $\openone_x$ the identity operator acting on a (graded Hilbert) space of dimension $x$, and by ${\cal U}^{(x)}={\cal U}\gtimes \openone_x$, i.e. the tensor generating the fMPU $U^{(N)}\gtimes \openone_x^{\otimes N}$.

\begin{defn}
\label{def:strictly-equivalent-tensors}
Two even tensors ${\cal U}$ and ${\cal V}$, with auxiliary parity operators $Z_U$, $Z_V$ and in GCF are strictly equivalent if $d^{e/o}_a=d^{e/o}_b$ and there exist a continuous path ${\cal W}(p)$ of even tensors with respect to a parity $Z_W(p)$, not
necessarily in GCF, with $p\in[0,1]$ such that ${\cal W}(0)={\cal U}$,
${\cal W}(1)={\cal V}$, and $Z_W(0)=Z_{U}$, $Z_W(1)=Z_{V}$.
\end{defn}

\begin{defn}
\label{def:equivalent-tensors}
Two even tensors $\mathcal{U}$ and $\mathcal{ V}$ are {equivalent} if there exists
some $k\in \mathbb{N}$ and $p_a$, $p_b$ such that $\mathcal{ U}^{(p_a)}_{k}$
and $\mathcal{V}^{(p_b)}_{k}$ are strictly equivalent.
\end{defn}
Here, it is important to stress that in the above definitions we allow for the ancilla to be a graded space where the dimensions of the even and odd subspaces are arbitrary.

As we have seen in Sec.~\ref{sec:standard_form_type_I}, type I fMPUs admit the same kind of standard form as qudit MPUs. For this reason, their index can be analyzed in the exact same way as done in Ref.~\cite{cirac2017matrix}. In particular, following the same steps therein, the properties of stability of the index for type I fMPUs can be established without additional difficulties. Then, one arrives at following result, for which we omit the proof (since, once again, it is completely analogous to the one for qudit MPUs).
\begin{prop}\label{prop:type_I_Index}
	For type I fMPUs  the exponentiated index $\mathcal{I}_f$ is a rational number, and
	\begin{enumerate}
		\item the index does not change by blocking; furthermore if $k$ is such that $\mathcal{U}_k$  is simple and $q>k$, then $r_{q}=r_k d^{q-k}$, $\ell_{q}=\ell_k d^{q-k}$, where $r_{k,q}$, $\ell_{k,q}$ are the right and left ranks of  $\mathcal{U}_{k,q}$, while $d$ is the local physical dimension;\label{point:1_typeI}\\[-0.5cm]
		\item the index is additive by tensoring and composition; \\[-0.5cm]
		\item the index is robust, i.e. by changing continuously $\mathcal{U}$ one cannot change it;\\[-0.5cm]
		\item two tensors have the same index iff they are equivalent.\label{point_3}\\[-0.5cm]
	\end{enumerate}
\end{prop}

Despite the proof of the above proposition is the same for fermions and qudits, it is useful to stress one difference between the two cases, which was already pointed out in Ref.~\cite{fidkowski2019interacting}. Let us consider two even simple tensors $\mathcal{U}_A$, $\mathcal{U}_B$ generating type I fMPUs, and suppose that they both have index $0$, with the same local dimension $d$. Both fMPUs admit a standard form in terms of the unitaries $u_{A,B}$, $v_{A,B}$ introduced in Eq.~\eqref{eq:u_v}, and since both tensors have vanishing indices all input and output legs are associated with the same dimension $d$. Now, in the case of qudits, this would imply that there exists a continuous path of unitaries connecting $u_{A}$, $v_{A}$ with $u_{B}$, $v_{B}$. However, this is not  always true for fermions, if we also require that all unitaries have well-defined parity, because graded spaces of the same dimension are not necessarily isomorphic. Namely, denoting by $\mathcal{H}^{A}_j$, $\mathcal{H}^{B}_j$ the local spaces associated with the output of $u_A$, $u_B$, in the above example we can only conclude that $\mathcal{H}^{A}_j\simeq \mathbb{C}^{p_A|q_A}$, $\mathcal{H}^{B}_j\simeq \mathbb{C}^{p_B|q_B}$ with $p_{A}+q_{A}=p_{B}+q_{B}=d$. At this point, however, a continuous path of even unitaries can always be constructed by appending an ancillary space $\tilde{\mathcal{H}}_j\simeq \mathbb{C}^{1|1}$, since $\mathcal{H}_j^{A,B}\gtimes \tilde{\mathcal{H}}_j\simeq \mathbb{C}^{\tilde{p}_{A,B}|\tilde{q}_{A,B}}$ and $\tilde{p}_{A}=\tilde{q}_{A}=\tilde{p}_{B}=\tilde{q}_{B}=d$. 

Next, in order to complete the proof of Theorem~\ref{th:femionic_index}, we analyze the index for type II fMPUs. In fact, this can be done quite easily based on their standard form introduced in Sec.~\ref{sec:standard_form_type_II}. However, at this point we need two additional lemmas, that are proved in the following. 

\begin{lem}\label{lem:majorana_rank}
Let $\mathcal{M}$ be the tensor associated with the Majorana shift~\eqref{eq:m_matrices}, and consider $\mathcal{M}^{(x)}=\mathcal{M}\gtimes \openone_x$. Denoting by $d=2x$ the dimension of the associated physical local space, the right and left ranks for the tensor $\mathcal{M}^{(x)}_{k}$ obtained by blocking $k$ times are $r_k=2d^k$, $\ell_k=d^k$. 
\end{lem} 
\begin{proof}

Let us first consider the right rank $r$. Clearly, $\mathcal{M}^{(x)}$ and $\mathcal{M}$ have the same bond dimension, and (normalized) transfer matrix $E_{\mathcal{M}}$. Now, recalling that for any operator $A$, ${\rm rank}(A)={\rm rank}(A^{\dagger}A)$, we have
\begin{equation}\label{eq:rank_r_majorana}
\begin{tikzpicture}[baseline={([yshift=1.5ex]current bounding box.center)}, scale=0.7]	
\draw(-1,0.) node{rank};
\draw [black] (0,-0.5) to [round left paren] (0,0.5);
\draw [black] (3.2,-0.5) to [round right paren] (3.2,0.5);
\draw (0.25,0) -- (1.25,0);	
\draw (0.75,-0.5)  -- (0.75,0.5);
\filldraw[fill=white,draw=black] (0.75,0) circle [radius=0.25];
\draw (2.7,-0.5)  -- (2.7,0.5);
\draw (2.2,0)  -- (3.1,0);	
\filldraw[fill=white,draw=black] (2.7,0) circle [radius=0.25];
\draw [decorate, decoration={brace,amplitude=4pt, raise=4pt},yshift=0pt]
(3.,-0.5) -- (0.25,-0.5) node [black,midway,yshift=-0.55cm] {$k$};
\draw[dashed, thick, color=gray](0.15,0.15) -- (3.2,-0.15);	
\draw(1.75,0) node{$\cdots$};
\end{tikzpicture}
=
\begin{tikzpicture}[baseline={([yshift=1.ex]current bounding box.center)}, scale=0.7]	
\draw(-1,0.25) node{rank};
\draw [black] (0,-0.5) to [round left paren] (0,1);
\draw [black] (3.5,-0.5) to [round right paren] (3.5,1);
\draw (0.25,0) -- (1.25,0);	
\draw (0.75,-0.5)  -- (0.75,1.25);
\filldraw[fill=white,draw=black] (0.75,0) circle [radius=0.25];
\draw (2.7,-0.5)  -- (2.7,1.25);
\draw (2.2,0)  -- (3.1,0);	
\filldraw[fill=white,draw=black] (2.7,0) circle [radius=0.25];
\draw [decorate, decoration={brace,amplitude=4pt, raise=4pt},yshift=0pt]
(3.0,-0.5) -- (0.25,-0.5) node [black,midway,yshift=-0.55cm] {$k$};
\draw(1.75,0) node{$\cdots$};
\draw (0.25,0.75) -- (1.25,0.75);	
\draw (2.2,0.75)  -- (3.1,0.75);	
\draw(1.75,0.75) node{$\cdots$};
\fill[black] (0.75,0.75) circle [radius=0.25];
\fill[black] (2.7,0.75) circle [radius=0.25];
\draw[dashed, thick, color=gray](0.15,0.35) -- (3.2,0.35);	
 \draw (3.1,0) arc (-90:90:0.375);
\end{tikzpicture}\,,
\end{equation}
where we separated input and output with a gray dotted line. Note that pictures on both sides of this equation define linear maps on graded spaces, and the rank refers to these maps. Next, for the Majorana shift tensor $\mathcal{M}^{(x)}$, we have
\be\label{eq:zip}
\begin{tikzpicture}[baseline={([yshift=-0.5ex]current bounding box.center)}, scale=0.7]	
\draw (1.25,1) arc (-270:-90:0.5);
\draw (0.05,1) arc (90:-90:0.5);
\draw (0,-0.5) node[left]{} -- (0,1.5) node[left]{};
\draw (-0.5,0) node[left]{} -- (0,0) node[right]{};
\draw (1.25,0) node[left]{} -- (1.5,0) node[right]{};
\draw (-0.5,1) node[left]{} -- (0,1) node[right]{};
\draw (1.25,1) node[left]{} -- (1.5,1) node[right]{};
\draw (1.25,-0.5) node[left]{} -- (1.25,1.5) node[left]{};
\filldraw[fill=white,draw=black] (0,0) circle [radius=0.25];
\fill[black] (0,1) circle [radius=0.25];
\draw (1.45,1) arc (90:-90:0.5);
\filldraw[fill=white,draw=black] (1.25,0) circle [radius=0.25];
\fill[black] (1.25,1) circle [radius=0.25];
\draw [fill=white] (0.4,0.25) rectangle (0.9,0.75);
\draw(0.7,-0.2)node{$E$};
\end{tikzpicture}
=
\begin{tikzpicture}[baseline={([yshift=-0.5ex]current bounding box.center)}, scale=0.7]	
\draw (0.05,1) arc (90:-90:0.5);
\draw (0,-0.5) node[left]{} -- (0,1.5) node[left]{};
\draw (-0.5,0) node[left]{} -- (0,0) node[right]{};
\draw (-0.5,1) node[left]{} -- (0,1) node[right]{};
\draw (1.05,-0.5) node[left]{} -- (1.05,1.5) node[left]{};
\filldraw[fill=white,draw=black] (0,0) circle [radius=0.25];
\fill[black] (0,1) circle [radius=0.25];
\end{tikzpicture}\,.
\ee 
Here we used that the identity operator corresponds to the even eigenstate of the transfer matrix with eigenvalue $1$, while $\sigma^x$ corresponds to the odd eigenstate with the same eigenvalue, and
\be
\begin{tikzpicture}[baseline={([yshift=-0.5ex]current bounding box.center)}, scale=0.7]	
\draw (0.05,1) arc (90:-90:0.5);
\draw (-0.05,0) arc (270:90:0.5);
\draw (0,-0.5) node[left]{} -- (0,1.5) node[left]{};
\draw [fill=white] (-1,0.25) rectangle (-0.4,0.75);
\draw (-0.7,0.5) node{$\sigma^x$};
\filldraw[fill=white,draw=black] (0,0) circle [radius=0.25];
\fill[black] (0,1) circle [radius=0.25];
\end{tikzpicture}
=0\,.
\ee
Using the simpleness condition, we can simplify the rhs of Eq.~\eqref{eq:rank_r_majorana} by iterating Eq.~\eqref{eq:zip} $k-1$ times. In the  resulting diagram, we have $k-1$ identities, whereas in the leftmost site we have the same tensor appearing on the left of the rhs of Eq.~\eqref{eq:zip}. For this one, the rank can be computed directly, obtaining $2d$, which yields $r_k=2d^k$. Applying a similar procedure for the left rank, we can also show $\ell=d^k$, thus completing the proof. Note that it is crucial to keep track of the signs arising from rearranging the tensors in the different diagrams as they are contracted, since these are at the root of the difference between the right and left ranks.
\end{proof} 

\begin{lem}\label{lemma:index_type_II}
Let $U^{(N)}$ be a type II fMPU in the standard form~\eqref{eq:standard_form_type_II}, and denote by $\tilde{\mathcal{U}}$, $\mathcal{M}$ the tensors in GCF associated with $\tilde{U}^{(N)}$ and $M^{(N)}_A$, respectively. Define by $\mathcal{U}$ the tensor obtained by composing $\tilde{\mathcal{U}}$ and $\mathcal{M}$, and let $k$ be such that $\tilde{\mathcal{U}}_k$ is simple. Then, the exponentiated index for $\mathcal{U}_{q}$ with $q\geq 2k$ is $\mathcal{I}_f=\sqrt{2}\tilde{\mathcal{I}}_f$, where $\tilde{\mathcal{I}}_f$ is the index of $\tilde{\mathcal{U}}_k$.
\end{lem}
\begin{proof}
We denote by $d$ be the local physical dimension, and use the following notations for the tensors blocked $k$ and $k^\prime=q-k\geq k$ times (where $k$ is such that $\mathcal{U}_k$ is simple)
\begin{subequations}
\be
\tilde{\mathcal{U}}_k=
\begin{tikzpicture}[baseline={([yshift=-0.5ex]current bounding box.center)}, scale=0.7]	
\draw (0,-0.5) node[left]{} -- (0,0.5) node[left]{};
\draw (-0.5,0) node[left]{} -- (0.5,0) node[right]{};
\draw [fill=white] (-0.25,-0.25) rectangle (0.25,0.25);
\end{tikzpicture}\,,
\qquad 
\mathcal{M}_k=
\begin{tikzpicture}[baseline={([yshift=-0.5ex]current bounding box.center)}, scale=0.7]
\draw (0,-0.5)  -- (0,0.5);
\draw (-0.5,0)  -- (0.5,0);	
\filldraw[fill=white,draw=black] (0,0) circle [radius=0.25];
\end{tikzpicture}\,,
\ee 
\be
\tilde{\mathcal{U}}_{k^\prime}=
\begin{tikzpicture}[baseline={([yshift=-0.5ex]current bounding box.center)}, scale=0.7]	
\draw (0,-0.5) node[left]{} -- (0,0.5) node[left]{};
\draw (-0.5,0) node[left]{} -- (0.5,0) node[right]{};
\draw [fill=gray] (-0.25,-0.25) rectangle (0.25,0.25);
\end{tikzpicture}\,,
\qquad 
\mathcal{M}_{k^\prime}=
\begin{tikzpicture}[baseline={([yshift=-0.5ex]current bounding box.center)}, scale=0.7]
\draw (0,-0.5)  -- (0,0.5);
\draw (-0.5,0)  -- (0.5,0);	
\filldraw[fill=gray,draw=black] (0,0) circle [radius=0.25];
\end{tikzpicture}\,.
\ee 
\end{subequations}
Here $\mathcal{M}_{k}$ denotes the tensor obtained by blocking $k$ times the one associated with the Majorana shift appearing in Eq.~\eqref{eq:standard_form_type_II}. Finally, we denote by $d_k$ the physical dimension associated with the blocked tensors, $d_k=d^k$, and  by $r_2$, $\tilde{r}_2$, $R_2$ the right rank corresponding to $\mathcal{U}_{q}$, $\tilde{ \mathcal{U}}_{q}$ and $\mathcal{M}_{q}$, respectively (and analogously for $\ell_2$, $\tilde{\ell}_2$, $L_2$). Due to point~\ref{point:1_typeI} in Prop.~\ref{prop:type_I_Index}, we have $\tilde{r}_2=\tilde{r}d_{k^\prime}$, where $\tilde{r}$ is the right rank of $\tilde{ \mathcal{U}}_k$. On the other hand, from Lemma~\ref{lem:majorana_rank}, we have $R_2=2d_k d_{k^\prime}$. Since this is the maximum possible rank, it is immediate to see that $r_2=2\tilde{r}_2=2\tilde{r}d_{k^\prime}$. Let us now consider the left rank $\ell_2$. As a first step, we define the maps
\be
\mathcal{F}_\alpha: |\beta)\gtimes\ket{i}\gtimes\ket{j}\mapsto \left(\mathcal{M}_{q}\right)_{\alpha,\beta}^{(x,y),(i,j)}\ket{x}\gtimes\ket{y}\,,
\ee
for $\alpha=0,1$, which correspond to the operators on graded spaces with graphical representation
\be
\mathcal{F}_\alpha=
\begin{tikzpicture}[baseline={([yshift=-0.5ex]current bounding box.center)}, scale=0.7]	
\draw (0,-0.5) node[left]{} -- (0,0.5) node[left]{};
\draw (-0.5,0) node[left]{} -- (1.5,0) node[right]{};
\draw (1.,-0.5) node[left]{} -- (1.,0.5) node[left]{};
\filldraw[fill=white,draw=black] (0,0) circle [radius=0.25];
\filldraw[fill=gray,draw=black] (1.,0) circle [radius=0.25];
\draw[dashed, thick, color=gray](-0.5,-0.15) -- (1.5,0.15);	
\draw(-1,0)node{$(\alpha|$};
\end{tikzpicture}
=
\left(M_{q}\right)^{^{(x,y),(i,j)}}_{\alpha,\beta} \ket{x} \ket{y}\bra{j}\bra{i}(\beta|\,,
\ee
where we separated with a dashed gray line input and output local spaces. Note that the output associated with the auxiliary space is fixed to $|\alpha)$, while $(i,j)$, $(x,y)$ label the input and output physical spaces, respectively. We also introduce the operators 
\be\label{eq:graphical_g_k}
\mathcal{G}_\alpha=
\begin{tikzpicture}[baseline={([yshift=-0.5ex]current bounding box.center)}, scale=0.7]	
\draw (0,-0.5) node[left]{} -- (0,1.25) node[left]{};
\draw (-0.5,0) node[left]{} -- (1.5,0) node[right]{};
\draw (-0.5,0.75) node[left]{} -- (1.5,0.75) node[right]{};
\draw (1.,-0.5) node[left]{} -- (1.,1.25) node[left]{};
\filldraw[fill=white,draw=black] (0,0) circle [radius=0.25];
\filldraw[fill=gray,draw=black] (1.,0) circle [radius=0.25];
\draw(-1,0)node{$(\alpha|$};
\draw [fill=white] (-0.25,0.5) rectangle (0.25,1.);
\draw [fill=gray] (1+-0.25,0.5) rectangle (1+0.25,1.);
\draw[dashed, thick, color=gray](-0.5,-0.5) -- (1.5,1.25);	
\end{tikzpicture}\,,
\qquad 
\mathcal{K}=
\begin{tikzpicture}[baseline={([yshift=-0.5ex]current bounding box.center)}, scale=0.7]	
\draw (0,-0.5) node[left]{} -- (0,1.25) node[left]{};
\draw (-0.5,0) node[left]{} -- (1.5,0) node[right]{};
\draw (-0.5,0.75) node[left]{} -- (1.5,0.75) node[right]{};
\draw (1.,-0.5) node[left]{} -- (1.,1.25) node[left]{};
\filldraw[fill=white,draw=black] (0,0) circle [radius=0.25];
\filldraw[fill=gray,draw=black] (1.,0) circle [radius=0.25];
\draw [fill=white] (-0.25,0.5) rectangle (0.25,1.);
\draw [fill=gray] (1+-0.25,0.5) rectangle (1+0.25,1.);
\draw[dashed, thick, color=gray](-0.5,-0.5) -- (1.5,1.25);	
\end{tikzpicture}\,,
\ee
with explicit action defined by
\begin{align}
\mathcal{G}_\alpha: &|\beta)|\delta)\ket{i}\ket{j}\mapsto \left(\mathcal{U}_{q}\right)_{(\alpha,\gamma),(\beta,\delta)}^{(x,y),(i,j)}|\gamma)\ket{x}\ket{y}\,,\\
\mathcal{K}:&|\beta)|\delta)\ket{i}\ket{j}\mapsto \left(\mathcal{U}_{q}\right)_{(\alpha,\gamma),(\beta,\delta)}^{(x,y),(i,j)}|\alpha)|\gamma)\ket{x}\ket{y}\,.
\end{align}
where repeated indices are summed over. Repeating the steps in Lemma~\ref{lem:majorana_rank}, we obtain ${\rm rank}(\mathcal{F}_0)={\rm rank}(\mathcal{F}_1)=d_kd_{k^\prime}$. It follows that ${\rm rank}(\mathcal{G}_0)={\rm rank}(\mathcal{G}_1)=\tilde{\ell}_2=\tilde{\ell}d_{k^\prime}$,  where $\tilde{\ell}$ is the right rank of $\tilde{ \mathcal{U}}_k$. In turn, this implies $\ell_2\geq \tilde{\ell}d_{k^\prime}$. Indeed, let $\ket{v_1},\ldots \ket{v_{\tilde{\ell}_2}}$ be a basis for the image of $\mathcal{G}_0$, and take $\{\ket{w_j}\}_{j=1}^{\tilde{\ell}_2}$ such that $\mathcal{G}_0 \ket{w_j}=\ket{v_j}$. Then, 
\be
\mathcal{K}\ket{w_j}=|0)\gtimes \ket{v_j}+ |1)\gtimes \ket{z_j}\,,
\ee
for some $\ket{z_j}$, where we used that $\mathcal{M}_{k}$ has bond dimension $2$. Thus, $\{\mathcal{K}\ket{w_j}\}_{j=1}^{\tilde{\ell}_2}$ are linearly independent, namely ${\rm rank}(\mathcal{K})=\ell_2\geq \tilde{\ell}d_{k^\prime}$. On the other hand, from the graphical representation of $\mathcal{K}$ in Eq.~\eqref{eq:graphical_g_k}, it is clear that ${\rm rank}(\mathcal{K})=\ell_2\leq (\tilde{\ell}d^{k^{\prime}-k}d_{k})=\tilde{\ell}d_{k^{\prime}}$, where we used that the left rank of $\tilde{\mathcal{U}}_{k^\prime}$ is $\tilde{\ell}d^{k^{\prime}-k}$ (Prop.~\ref{prop:type_I_Index}) and the left rank of $\mathcal{M}_k$ is $d_k$ (Lemma~\ref{lem:majorana_rank}). Thus, $\ell_2=\tilde{\ell}d_{k^\prime} $. In summary, we have $r_2=2\tilde{r}d_{k^\prime}$ and  $\ell_2=\tilde{\ell}d_{k^\prime}$, which proves the claim. 
\end{proof}

At this point it is important to note that the tensor $\mathcal{U}_q$, obtained by composing $\tilde{\mathcal{U}}_q$ and $\mathcal{M}_q$, is not necessarily in GCF. On the other hand, the index was defined for tensors in GCF, so that one needs to make sure that the index of $\mathcal{U}_q$ coincides with that computed in the corresponding GCF (which is unique, due to Theorem~\ref{thm2_fundamental}). This is true, up to blocking $\tilde{q}\geq 4k$ times, and we report the proof of this statement in Appendix~\ref{sec:index_GCF}. Combining now Prop.~\ref{prop:type_I_Index} with Lemma~\ref{lemma:index_type_II} it is straightforward to prove the following proposition.
\begin{prop}\label{prop:type_II_Index}
	For type II fMPUs  the exponentiated index is in the form $\mathcal{I}_f=\sqrt{2}(p/q)$, with $p, q\in\mathbb{N}$ and coprime. Furthermore
	\begin{enumerate}
		\item the index does not change by blocking; \\[-0.5cm]
		\item the composition (or tensor product) of two type II fMPUs is a type I fMPU, and the index of the latter is obtained by summing the indices of the former two; \\[-0.5cm] 
		\item the index is robust, i.e. by changing continuously $\mathcal{U}$ one cannot change it.\\[-0.5cm]
		\item two tensors have the same index iff they are equivalent.\label{point_3_iff}
	\end{enumerate}
\end{prop}
Note that in order to prove point~\ref{point_3_iff} one needs to use that if $\mathcal{U}$ is a tensor generating an fMPU, then it can be deformed continuously to a tensor in GCF generating the same fMPU, which is easily proven based on the fact that $\mathcal{U}$ has a single block in the GCF.  Now, putting together Prop.~\ref{prop:type_I_Index} and Prop.~\ref{prop:type_II_Index}, we finally arrive at the statement of Theorem~\ref{th:femionic_index}, anticipated at the beginning of this section. 

 \section{Conclusions}
 \label{sec:conclusions}
 
In this work we have studied matrix product unitaries for fermionic $1D$ chains, and highlighted several qualitative differences with respect to the case of qudits. In particular, we have shown that $(i)$ fMPUs are not necessarily locality-preserving and $(ii)$ not all fQCA can be represented as standard fMPUs, with either periodic or antiperiodic boundary conditions. Next, we have defined the class of generalized fMPUs, and identified a subset of the latter that are equivalent to the family of fQCA. Finally, we have shown how the index for  fQCA~\cite{fidkowski2019interacting} can be extracted directly from the tensors generating fMPUs. As a techical byproduct of our work, we have also introduced a graded canonical form for fermionic MPSs, which might be useful for more general problems. Overall, our work shows that fMPUs display significantly richer features when compared to the case of  MPUs. 

There are several interesting questions that remain open. For example, we have shown that fMPUs of the second kind are always generated by a tensor that is obtained by composing a Majorana shift and another tensor $\tilde{\mathcal{U}}$, generating an fMPU of the first kind, cf. Eq.~\eqref{eq:standard_form_type_II}. However, one could wonder whether a more explicit standard form exists for graded normal tensors generating type II fMPUs, in analogy with Eq.~\eqref{eq:standard_form} for type I fMPUs.  

Clearly, another natural question pertains the classification of fMPUs in the presence of additional symmetries, which was recently addressed in the case of qudits ~\cite{Gong2020Classification}. While the tools introduced in this work allow us to tackle this problem, based on the case of qudits we expect that additional difficulties will arise, and we leave this question for future investigations.

\section*{Acknowledgments} 
AT and SS would like to thank Xie Chen and Burak \c{S}ahino\u{g}lu for relevant conversations. LP acknowledges support from the Alexander von Humboldt foundation. JIC acknowledges support by the EU Horizon 2020 program through the ERC Advanced Grant QENOCOBA No. 742102, and from the DFG (German Research Foundation) under Germany’s Excellence Strategy - EXC-2111 - 390814868.

\appendix

\section{Fiducial-state formalism}
\label{sec:fiducial_state}

\subsection{General definitions}

In this appendix, we develop the fiducial-state formalism for $1D$ fermionic systems, following the construction introduced in Ref.~\cite{wahl_symmetries_2014} for two spatial dimensions. The latter is a natural generalization to that carried out for qudits, and consists in performing a sequence of projections onto maximally entangled auxiliary fermionic modes. While in Ref.~\cite{kraus_fermionic_2010} this was done by associating  with each lattice site two auxiliary real fermions, a slightly simpler formulation can be obtained by considering maximally entangled Majorana modes instead~\cite{wahl_symmetries_2014}. Here we follow the latter approach, which is detailed in the following.
 
 As in the main text, we consider a chain of $N$ sites labeled by $x = 1,\ldots,N$. We associate with each site $n_F$ fermionic modes with (physical) annihilation operators $a_{x,j}$. Furthermore, we introduce two sets of auxiliary Majorana operators $\{\ell_{x,\mu}\}$  and $\{r_{x,\mu}\}$ with $x=1,\ldots N$, and $\mu=1,\ldots N_F$,  where $N_F$ is some non-negative integer, while $\ell$ and $r$ stand for left and right, respectively.  The Majorana operators are self-adjoint  and satisfy
 \begin{align}
\{r_{x,\mu},\ell_{y,\nu}\}&=0\,,\\
\{r_{x,\mu},r_{y,\nu}\}&=\{\ell_{x,\mu},\ell_{y,\nu}\}=2\delta_{x,y}\delta_{\mu,\nu}\,.
 \end{align}

 As a starting point, we introduce the local fermionic tensors acting on one physical site, and two adjacent auxiliary Majorana modes: 
\begin{align}
F_{x}&=\sum_{\alpha, \beta=0}^{D-1} \sum_{n=0}^{d-1} A_{\alpha , \beta}^{n} \ell_{x}^{\alpha} \left(a^\dagger_{x}\right)^{n} r_{x}^{\beta}\,\qquad x=1\,,\ldots N\,, 
\label{eq:F_local_operator}
\end{align}
where $A^{n}_{\alpha,\beta}$ are complex numbers, and $D = 2^{N_F}$, $d = 2^{n_F}$. Here, we have used the short-hand notation introduced in Eq.~\eqref{eq:short_hand} for $\left(a^\dagger_{x}\right)^{n} $, and also
\begin{align}
\ell_{x}^{\alpha}= \ell_{x,1}^{\alpha^{(1)}} \cdots \ell_{x,N_F}^{\alpha^{(N_F)}} \,, \label{eq:short_hand_ell}\\
 r_{x}^{\beta}=r_{x,N_F}^{\beta^{(N_F)}}\cdots r_{x,1}^{\beta^{(1)}} \label{eq:short_hand_r}\,,
\end{align}
where $(n^{(1)},\ldots ,n^{(n_F)})$ is the binary representation for $n$, and analogously for $\alpha$ and $\beta$. Importantly, the order of the factors in the product of Eq.~\eqref{eq:short_hand_r} is reversed with respect to \eqref{eq:short_hand_ell}. While this is just a convention, it simplifies some of the subsequent algebraic manipulations.

For our construction, it is crucial that $F_{x}$ has well-defined fermionic parity  (with respect to all physical and auxiliary modes). As in the main text, we can define
\be
|n|= \sum_{i=1}^{n_F} n^{(i)}\,, \ |\alpha|=\sum_{j=1}^{N_F}\alpha^{(j)}\,,\ 
|\beta|=\sum_{j=1}^{N_F}\beta^{(j)}\,.
\ee
Then, saying that $F_{x}$ has well-defined fermionic parity means that $A^{n}_{\alpha,\beta}=0$ if 
\be
|n|+|\alpha|+|\beta| \quad (\bmod\ 2)
\ee
is $1$ or $0$, depending on whether $F_j$ is even or odd. As we have already mentioned in the main text, in the following we will always assume $F_x$ to be even, which is not a restriction if we are allowed to  blocking. Next, we introduce  a  projection onto neighboring Majorana modes 
\be
\eta_{x, y}=\prod_{\mu=0}^{D-1} \eta_{x, y, \mu}\,,
\ee
where
\be
\eta_{x, y, \mu}=\frac{1}{2}\left(1+i r_{x, \mu} \ell_{y, \mu}\right)\,.
\ee
We note that, after acting with $\eta_{x, y,\mu}$, the Majorana fermions $r_{x, \mu}$, $\ell_{y, \mu}$ become maximally entangled, forming a pure fermionic state. 

Now, an fMPS is obtained in  two steps. First, we act on the fermionic vacuum (with respect to all modes) with the operators  $F_x$, $x=1,\ldots, N$. Second, we ``concatenate'' them by projecting neighboring Majorana modes onto maximally entangled  pairs with $\eta_{x, y}$~\cite{wahl_symmetries_2014}. Importantly,  using this procedure we can generate fermionic states that are invariant under translation with either periodic or antiperiodic boundary conditions. Let us consider the former case. Then, applying the prescription above, we obtain
\be
|\Psi\rangle=\left\langle\eta_{N, 1} \eta_{1,2} \ldots \eta_{N-1, N} F_{1} \ldots F_{N} \right\rangle_{a}|\Omega\rangle\,.
\label{eq:fMPS}
\ee
Here the expectation value is with respect to  the auxiliary vacuum  $\ket{\Omega_a}$, defined by
\be
(r_{x-1,\mu}-i\ell_{x,\mu})\ket{\Omega_a}=0\,,
\ee
 while we  denoted by $\ket{\Omega}$ the physical one, $a^{n}_{x}\ket{\Omega}=0$ for $n\neq (0,\ldots 0)$. Since we always have an even number of auxiliary operators, we can associate a tensor product structure between the spaces of real and auxiliary fermions, and the  expression above is well defined. 

It is not difficult to see that $\ket{\Psi}$ is invariant with respect to translations with periodic boundary conditions, due to the presence of the projector $\eta_{N, 1}$. If one is interested in the case of antiperiodic boundary conditions instead, it is enough to replace it with $\eta_{1,N}$. Since for the sake of our present discussion both types of boundary  conditions can be treated analogously,  in the following we only focus on the periodic case.

By performing the expectation value with respect  to the auxiliary vacuum, one can cast the state~\eqref{eq:fMPS} into the form~\eqref{eq:psi_explicit}, where the coefficients are written as in Eq.~\eqref{eq:fMPS_pbc}. This  can be seen as follows. We first rewrite Eq.~\eqref{eq:fMPS} as
\begin{align}
|\Psi\rangle=A^{n_1}_{\alpha_1,\beta_1}\cdots A^{n_N}_{\alpha_N,\beta_N}\langle\Omega_{a}|\eta_{N, 1} \cdots \eta_{N-1, N}\nonumber\\
\ell^{\alpha_1}_1\left(a_1^{\dagger}\right)^{n_1}r^{\beta_1}_1\cdots \ell^{\alpha_N}_N\left(a_N^{\dagger}\right)^{n_N}r^{\beta_N}_N|\Omega_{a}\rangle \otimes|\Omega\rangle\,,
\end{align}
where repeated indices are summed over. Moving $r_N^{\beta_N}$ to the left and using the anticommutation relations, we obtain
\begin{align}
|\Psi\rangle&=A^{n_1}_{\alpha_1,\beta_1}\cdots A^{n_N}_{\alpha_N,\beta_N}\langle\Omega_{a}|\eta_{N, 1} \cdots \eta_{N-1, N}(-1)^{p|\beta_N|}\nonumber\\
&r^{\beta_N}_N\ell^{\alpha_1}_1\left(a_1^{\dagger}\right)^{n_1}r^{\beta_1}_1\cdots \ell^{\alpha_N}_N\left(a_N^{\dagger}\right)^{n_N}|\Omega_{a}\rangle \otimes|\Omega\rangle\,,
\end{align}
where
\be
p=\sum_{j=1}^N|n_j|+\sum_{j=1}^N|\alpha_j|+\sum_{j=1}^{N-1}|\beta_j|\quad (\bmod\ 2)\,.
\ee
We can now move each $\eta_{j,j+1}$ to the right until it acts on $\ket{\Omega}$. This procedure yields zero unless $\beta_{j}=\alpha_{j+1}$ (with the identification $\alpha_{N+1}=\alpha_1$). Furthermore, for the non-vanishing terms, we have
\be
(-1)^{p|\beta_N|}=(-1)^{|\alpha_1||\beta_N|}=(-1)^{|\alpha_1|}\,.
\ee
Next, we can move to the right each remaining product $r^{\alpha_j}_{j}\ell^{\alpha_j}_{j+1}$, which gives $(-i)^{N_F}$ when acting on $\ket{\Omega_a}$. Finally, we can rearrange the product of elements $A^{n}_{\alpha,\beta}$ into a trace over an auxiliary graded space, which is generated by the basis  vectors $\ket{\alpha}$, with parity $|\alpha|$. Putting all together, we get
\begin{align}
|\Psi\rangle=(-i)^{NN_F}{\rm tr}\left[\tilde{Z} A^{n_1}\cdots A^{n_N}\right] \left(a_{1}^{\dagger}\right)^{n_1}\cdots  \left(a_{N}^{\dagger}\right)^{n_N}\ket{\Omega}\,.
\end{align}
Here $\tilde{Z}$ is a  diagonal matrix with entries $+1$ and $-1$ which acts as the parity operator on  the auxiliary space. As a last step, we perform a gauge transformation corresponding to a reordering of the basis vectors, and absorbing the factor $(-i)^{NN_F}$ into the matrices $A^{n}$ we arrive at the form in Eq.~\eqref{eq:fMPO_pbc}. Note that the commutation relations~\eqref{eq:commutation} follow from the parity of the tensor $F_x$ in Eq.~\eqref{eq:F_local_operator}.

It is straightforward to extend this formalism to define fMPOs. In this case, we consider the same setting as for fMPSs, but the local tensors~\eqref{eq:F_local_operator} must be replaced by
\begin{align}
G_{x}&=\sum_{\alpha, \beta=0}^{D-1} \sum_{n,m=0}^{d-1} M_{\alpha , \beta}^{n,m} \ell_{x}^{\alpha} f_{x}^{n,m} r_{x}^{\beta}\,\quad x=1\,,\ldots N\,,
\label{eq:F_local_operators}
\end{align}
where $f^{n,m}$ are the fermionic operators introduced in Eq.~\eqref{eq:small_f_operators}. Then, we can defined fMPOs by
\be
M=\left\langle\eta_{N, 1} \eta_{1,2} \ldots \eta_{N-1, N} G_{1} \ldots G_{N} \right\rangle_{a}\,,
\label{eq:fMPO}
\ee
where the expectation value is with respect to the auxiliary vacuum, as in Eq.~\eqref{eq:fMPS}. Repeating the derivation above, it is straightforward to cast $M$ as in Eq.~\eqref{eq:fermi_operator}, with the coefficients in the form~\eqref{eq:fMPO_pbc}.

\subsection{Elementary operations with fMPOs}
\label{sec:elementary_operators}

As for the case of qudits, one can see that the family of fMPOs is closed with respect to elementary operations such as sum, conjugate transposition or composition. However, some differences arise when computing the corresponding local tensors, as we now discuss. Let us first consider the case of the conjugate transposition of a given fMPO $U^{(N)}$. From Eq.~\eqref{eq:fermi_operator}, we have
\be
\left[U^{(N)}\right]^{\dagger}= \sum_{n_{1}, \dots, n_{N}=0\atop m_{1}, \dots, m_{N}=0}^{d-1} \bar{c}^{n_{1}, \dots, n_{N}}_{m_{1}, \dots, m_{N}} \left[f_N^{n_N,m_N}\right]^{\dagger}\cdots \left[f_1^{n_1,m_1}\right]^\dagger\,.
\ee
We would like to rewrite this  expression as in Eq.~\eqref{eq:fermi_operator}, with the coefficients $c^{n_{1}, \dots, n_{N}}_{m_{1}, \dots, m_{N}}$ in the form~\eqref{eq:fMPO_pbc} or \eqref{eq:fMPO_apbc}. In fact, this can be easily done, by making use of the identities
\begin{subequations}
	\begin{align}
	\left[f_{j}^{n,m}\right]^\dagger&= f_{j}^{m,n}\,,\\
	f_{x}^{m,n} f_{y}^{p,q}&=(-1)^{(|p|+|q|)(|m|+|n|)}f_{y}^{p,q}f_{x}^{m,n}\,,\\
	f_{x}^{m,n} f_{x}^{p,q}&=\delta_{n,p}f_x^{m,q}\,,
	\end{align}
\end{subequations}
which can be established by computing the matrix elements of both sides of the equations on basis (Fock) states. As a final result, we find that $\left[U^{(N)}\right]^{\dagger}$ is  an fMPO with local tensors defined by\footnote{More generally, if an additional operator $X$ is inserted into the trace, for instance as in Eq.~\eqref{eq:maj_per}, one also needs to replace $X\to \tilde{X}$, with $\tilde{X}_{\alpha,\beta}=(-1)^{|\beta|(|\alpha|+|\beta|)}\bar{X}_{\alpha,\beta}$}
\be
\tilde{U}_{\alpha, \beta}^{i,j}=(-1)^{|\beta|+|\alpha||\beta|}\bar{U}^{ji}_{\alpha,\beta}\,,
\label{eq:conjugate_tensor}
\ee
where as usual we denoted by $\bar{x}$ the complex conjugate of $x\in\mathbb{C}$.  

Next, let us consider two fMPOs $U^{(N)}$, $V^{(N)}$, and define $W^{(N)}=U^{(N)}V^{(N)}$. By exploiting the fiducial-state formalism, $W^{(N)}$ can be written again as an fMPO, where the local tensors are now defined by
\be
W^{k,i}_{(\alpha,\gamma),(\beta,\delta)}=\sum_j(-1)^{|\gamma|(|k|+|j|)} U^{k,j}_{\alpha,\beta}V^{j,i}_{\gamma,\delta} \,.
\label{eq:fermionic_stack_tensor}
\ee
This formula is similar to the corresponding one for qudits, but additional signs appear. 
In order to prove Eq.~\eqref{eq:fermionic_stack_tensor}, it is convenient to exploit the fiducial-state representation~\eqref{eq:fMPO}. Define
\begin{align}
U^{(N)}=\left\langle\eta^{U}_{N, 1} \eta^{U}_{1,2} \ldots \eta^{U}_{N-1, N} G^{U}_{1} \ldots G^{U}_{N} \right\rangle_{a}\,,\\
V^{(N)}=\left\langle\eta^{V}_{N, 1} \eta^{V}_{1,2} \ldots \eta^{V}_{N-1, N} G^{V}_{1} \ldots G^{V}_{N} \right\rangle_{a}\,,
\end{align}
where
\begin{align}
G^{U}_{x}&=\sum_{\alpha, \beta=0}^{D-1} \sum_{n,m=0}^{d-1} U_{\alpha , \beta}^{n,m} \left(\ell_{x}^{U}\right)^{\alpha} f_{x}^{n,m} \left(r^U_{x}\right)^{\beta}\,,\\
\eta^{U}_{x, y}&=\prod_{\mu=0}^{D-1} \frac{1}{2}\left(1+i r^{U}_{x, \mu} \ell^{U}_{y, \mu}\right)\,,
\end{align}
and analogously for $G^{V}_x$, $\eta_{x,y}^{V}$. Note that here we considered two different sets of Majorana operators, labeled by $U$, $V$. Defining $W^{(N)}=U^{(N)}V^{(N)}$, and exploiting the parity of the local tensors, we have
\be
W^{(N)}=\left\langle\eta^{UV}_{N, 1} \eta^{UV}_{1,2} \ldots \eta^{U}_{N-1, N} G^{UV}_{1} \ldots G^{UV}_{N} \right\rangle_{\tilde{a}}\,.
\label{eq:product}
\ee
The expectation value is now taken with respect to the vacuum $\ket{\Omega_{\tilde{a}}}$ of all auxiliary Majorana fermions, satisfying
\be
(r^{U}_{x-1,\mu}-i\ell^{U}_{x,\mu})\ket{\Omega_{\tilde{a}}}=(r^{V}_{x-1,\mu}-i\ell^{V}_{x,\mu})\ket{\Omega_{\tilde{a}}}=0\,.
\ee
Furthermore, we introduced $\eta_{x,y}^{UV}=\eta_{x,y}^{U}\eta_{x,y}^{V}$, and 
\begin{align}
G^{UV}_{x}&=G^{U}_{x}G^{V}_{x}=U^{k,j}_{\alpha,\beta} V^{j^\prime,i}_{\gamma,\delta}\left(\ell_{x}^U\right)^\alpha f_x^{k,j}\left(r_{x}^{U}\right)^\beta\nonumber\\
&\times \left(\ell_{x}^V\right)^\gamma f_x^{j^\prime,i}\left(r_{x}^{V}\right)^\delta=(-1)^{|\gamma|(|k|+|j|)}U^{k,j}_{\alpha,\beta}\nonumber\\
&\times  V^{j,i}_{\gamma,\delta}\left(\ell_{x}^U\right)^\alpha \left(\ell_{x}^V\right)^\gamma  f_x^{k,i} \left(r^{V}_{x}\right)^\delta\left(r^{U}_{x}\right)^\beta\,,
\label{eq:gUV}
\end{align}
where repeated indices are summed over. Note that in the last line we have moved $\left(r^{U}_x\right)^{\beta}$ to the left, for consistency with the conventions~\eqref{eq:short_hand_ell}, \eqref{eq:short_hand_r}. 

From Eqs.~\eqref{eq:product}, \eqref{eq:gUV}, we see that $W^{(N)}$ is already written in the form~\eqref{eq:fMPO}. Thus, repeating the derivation in the last subsection, we obtain that the local tensors appearing in~\eqref{eq:fMPO_pbc} can be read off directly from Eq.~\eqref{eq:gUV}, thus proving Eq.~\eqref{eq:fermionic_stack_tensor}.

Using the results above, it is also straightforward to derive the local tensor associated with $\mathbb{U}^{(N)}=\left[U^{(N)}\right]^\dagger U^{(N)}$. In particular, up to an even gauge transformation of the form 
\be
G_{(\alpha,\gamma),(\beta,\delta)}=(-1)^{|\alpha||\gamma|}\delta_{\alpha,\beta}\delta_{\gamma,\delta}\,,
\ee
a direct computation yields
\be
\mathbb{U}^{k,i}_{(\alpha,\gamma),(\beta,\delta)}=\sum_j(-1)^{|\beta|(|k|+|i|)} \bar{U}^{j,k}_{\alpha,\beta}U^{j,i}_{\gamma,\delta} \,.
\label{eq:unitary_stacked}
\ee
Accordingly, by taking $k=i$ and summing over $i$, we can also define the fermionic transfer matrix $E$, whose elements simply read
\be\label{eq:fermionic_transfer_matrix}
E_{(\alpha,\gamma),(\beta,\delta)}=\frac{1}{d}\sum_{i,j} \bar{U}^{j,i}_{\alpha,\beta}U^{j,i}_{\gamma,\delta}\,,
\ee 
completely analogously to the case of qudits. This definition allows us to write down a relation similar to Eq.~\eqref{eq:trace_transfer} for fermionic MPOs, which is important when discussing fMPUs. In particular, in the case of periodic boundary conditions, we immediately derive
\be
\frac{1}{d^N}{\rm tr}\left[U^{(N)\dagger}U^{(N)}\right]={\rm tr}\left[\left(Z\otimes Z\right)E_{\mathcal{U}}^N\right]\,,
\label{eq:trace_transfer_periodic}
\ee
while for antiperiodic boundary conditions we simply have
\be
\frac{1}{d^N}{\rm tr}\left[U^{(N)\dagger}U^{(N)}\right]={\rm tr}\left[E_{\mathcal{U}}^N\right]\,.
\label{eq:trace_transfer_antiperiodic}
\ee

Finally, one can define the blocking procedure also for fermionic tensor networks. Consider the fMPO~\eqref{eq:fermi_operator}, and suppose that we are interested in blocking pairs of neighboring sites. The annihilation operators associated with the ``doubled site'' at position $j$  are
\be
\tilde{a}_{j}^{(n,m)}=a^{n}_{2j} a^{m}_{2j+1}\,.
\ee
Now, it is easy to verify that
\be
\tilde{f}_{j}^{(n,m), (q,p)}=(-1)^{|m|(|p|+|q|)}f^{m,n}_{2j}f_{2j+1}^{q,p}\,,
\ee 
satisfy Eqs.~\eqref{eq:small_f_operators}, with the replacement $a_j^n\to \tilde{a}_{j}^{(n,m)}$. Accordingly, blocking leads to an fMPO, where the local tensor $\mathcal{U}_B$ is defined by the 
\be
\left(U^{(n,m),(p,q)}_B\right)_{\alpha,\beta} =U^{n,m} _{\alpha,\gamma} U^{p,q}_{\gamma,\beta}(-1)^{|m|(|p|+|q|)}\,.
\label{eq:signs_blocking}
\ee 

Based on these formulas, one can extend the graphical notation introduced in the case of qudits to fermionic TNs, where it is always understood that one should multiply the elements of the local tensors by the correct signs, as specified by the above equations. We note that a similar discussion can be carried out using the formalism of graded TNs, which has the advantage of offering a more transparent way to translate the algebraic formulation into a graphical one (and viceversa), cf.~Appendix~\ref{sec:graded_TN}.

\section{Graded tensor networks }
\label{sec:graded_TN}

In this appendix we review the formalism of graded tensor networks, as introduced recently in Refs.~\cite{bultinck2017fermionic,kapustin_spin_2018}. Here we only sketch the main definitions, and the interested reader is referred to the latter works for a more systematic treatment.

Given a Hilbert space $V$, we say that $V$ is $\mathbb{Z}_2$-graded if there is a parity operator $\mathcal{P}$ and a decomposition
\be
V=V^e\oplus V^o\,,
\ee
such that $\mathcal{P}\ket{v}=\ket{v}$ for all $\ket{v}\in V^e$,  and $\mathcal{P}\ket{v}=-\ket{v}$ for all $\ket{v}\in V^o$. We say that $V^e$, $V^o$  are the even and odd sectors of the graded space $V$, respectively. Let us now introduce a set of local $\mathbb{Z}_2$-graded Hilbert spaces $\mathcal{H}_j$, with decomposition
\be
\mathcal{H}_{j}=\mathcal{H}_j^e\oplus \mathcal{H}_j^o\,.
\ee
We denote by $\{\ket{i}\}_{i=0}^{d-1}$ a local basis, and by $|i|\in \{0,1\}$ the parity of each basis vector, so that $\mathcal{P}\ket{i}=(-1)^{|i|}\ket{i}$. In the following we will use $d_e$ and $d_o$ for the dimensions of $\mathcal{H}_j^e$ and $ \mathcal{H}_j^o$, respectively. Next, we introduce the notion of graded tensor product, denoted by $\otimes_\mathfrak{g}$. This is a tensor product equipped with a canonical isomorphism $\mathcal{F}$ between different ordering of the local spaces. Specifically, given the graded spaces $V$, $W$, the canonical isomorphism $\mathcal{F}$ is defined by
\begin{align}
	 \mathcal{F}:  V \otimes_{\mathfrak{g}} W &\rightarrow W \otimes_{\mathfrak{g}} V \,,\\
	 |i\rangle \otimes_{\mathfrak{g}}|j\rangle & \rightarrow(-1)^{|i||j|}|j\rangle \otimes_{\mathfrak{g}}|i\rangle \,.
\end{align}
Making use of $\mathcal{F}$, we can identify states in different ordered tensor products of the same local spaces. 

Importantly, the $\mathbb{Z}_2$-grading structure is inherited by the dual space $V^\ast$, generated by the basis $\{\bra{i}\}_{i=0}^{d-1}$, and the isomorphism $\mathcal{F}$ can be naturally extended to tensor products containing also local dual spaces.  Furthermore, a state in $V^\ast\otimes V$ can be naturally mapped onto $\mathbb{C}$  via the linear map
\be
\mathcal{C}: V^{*} \otimes_{\mathfrak{g}} V \rightarrow \mathbb{C}:\left\langle\psi\left|\otimes_{\mathfrak{g}}\right| \phi\right\rangle \rightarrow\langle\psi | \phi\rangle\,.
\ee
Note that $\mathcal{C}$ acts on the ordered graded tensor product $V^\ast\otimes V$. However, one can extend the action of $\mathcal{C}$ to a different ordering, by first applying the canonical isomorphism $\mathcal{F}$. For instance, using this prescription, we can compute
\be
\mathcal{C}\left(|i\rangle \otimes_{\mathfrak{g}}\langle j|\right)=(-1)^{|i||j|} \mathcal{C}\left(\left\langle j\left|\otimes_{\mathfrak{g}}\right| i\right\rangle\right)=(-1)^{|i|} \delta_{i, j}\,.
\ee

The above definitions allow us to generalize the construction of MPSs to graded Hilbert spaces. In particular, we can define the local tensors $\mathcal{A}$ by
\bea
\mathcal{A}[j]=&\sum_{i, \alpha, \beta} A_{\alpha, \beta}^{i} | \alpha )_{j-1} \gtimes|i\rangle_{j}\nonumber \\ 
& \gtimes (\beta |_{j} \in V_{j} \gtimes \mathcal{H}_{j} \gtimes\left(V_{j+1}\right)^{*}\,,
\label{eq:tensor}
\eea
where round kets and bras correspond to the bases of the auxiliary space $V _j  \simeq \mathbb{C}^{D_j}$ and its dual. An fMPS is then constructed by concatenating local tensors, and gluing them together by applying the contraction map $\mathcal{C}$~\cite{bultinck2017fermionic,kapustin_spin_2018}. In the case of periodic boundary conditions, for instance, this leads to
\be
|\psi\rangle=\mathcal{C}( \mathcal{A}[1] \gtimes \mathcal{A}[2] \gtimes \cdots \gtimes \mathcal{A}[N])\,.
\label{eq:gtn_mps}
\ee
A crucial requirement is that the local tensors $\mathcal{A}$ have well-defined parity, which we can assume to be even without loss of generality. This ensures that the fMPS have well-defined parity and that no ambiguity arises in the definition of some useful construction to manipulate them~\cite{bultinck2017fermionic}.

As in the case of qudits, we can introduce a natural graphical representation for graded TNs. For instance, local tensors $\mathcal{A}$ are depicted by
\be
\mathcal{A}=\sum_{i, \alpha, \beta} A_{\alpha, \beta}^{i} | \alpha ) \gtimes|i\rangle\gtimes (\beta | =
\begin{tikzpicture}[baseline={([yshift=-1.5ex]current bounding box.center)}, scale=0.7]
\draw (0,0) node[left]{} -- (0,0.5) node[left]{$i$};
\draw (-0.5,0) node[left]{$\alpha$} -- (0.5,0) node[right]{$\beta$};	
\filldraw[fill=white,draw=black] (0,0) circle [radius=0.15];
\end{tikzpicture}\,.
\ee
When different tensors are joined together, it is always understood that the linear map $\mathcal{C}$ is applied to the corresponding spaces. As usual, before applying $\mathcal{C}$, the local spaces in the graded tensor product must be reordered using the canonical isomorphism $\mathcal{F}$. Clearly, the above formalism can be applied directly to also treat MPOs in graded spaces. To this end, we simply replace the local tensor~\eqref{eq:tensor} with
\begin{align}
\mathcal{M}[j]=\sum_{i, \alpha, \beta} M_{\alpha, \beta}^{m,n} | \alpha )_{j-1} \gtimes|m\rangle_{j}\gtimes  \bra{n}_j& \gtimes (\beta |_{j}\,.
\label{eq:tensor_mpo}
\end{align}
Finally, the contraction in Eq.~\eqref{eq:gtn_mps} can be performed easily using the parity of the local tensors, leading to the more explicit form
\begin{align}
|\psi\rangle=\sum_{i_1,\ldots i_N=0}^{d-1}&{\rm tr}\left(\mathcal{P}_DA^{i_1}\cdots A^{i_N}\right)
|i_1\rangle\gtimes\nonumber\\
& \cdots \gtimes |i_N\rangle\,.
\label{eq:coordinate_form}
\end{align}
Here $\mathcal{P}_D$ is the parity operator acting on the auxiliary graded space with dimension $D$. Note that, due to the parity of $\mathcal{A}$, we have $\mathcal{P}_DA^{n}=(-1)^{|n|}A^{n}\mathcal{P}_D$. 

Based on Eq.~\eqref{eq:coordinate_form}, we see now an explicit correspondence between the fiducial-state and graded TN formalisms for fermionic MPSs.  Indeed, the coefficients in Eq.~\eqref{eq:coordinate_form} are the same appearing in Eq.~\eqref{eq:fMPS_pbc} (up to an even gauge transformation, corresponding to a reordering of the basis vectors in the auxiliary space). Furthermore, using the above prescription for the contraction of graded tensors, it is straightforward to derive, e.g., Eqs.~\eqref{eq:conjugate_tensor} and \eqref{eq:fermionic_stack_tensor} for the adjoint and composition of fMPOs. In fact, the fermionic operators $f^{n,m}$ introduced in Eq.~\eqref{eq:small_f_operators} simply correspond to $\ket{n}\gtimes\bra{m}$.

\section{Examples of fMPUs}

\subsection{Non-locality-preserving fMPUs}
\label{sec:tabulated_local_tensors}

In this appendix we provide further details on the fMPU constructed in Sec.~\ref{sec:nonlocal_fmpu}. First, for completeness, we  tabulate all $16$ matrices $U^{n,m}$ corresponding to the elements $U^{n,m}_{\alpha,\beta}$ in~\eqref{eq:u_00}, \eqref{eq:u_nm}. They read
\begin{widetext}
\begin{subequations}
\begin{align}
U^{0,0} =\left(
\begin{array}{ccc}
1 & 0 & 0 \\
0 & 1 & 0 \\
0 & 0 & 1 \\
\end{array}
\right)\,,\quad 
\ U^{0,1} =\left(
	\begin{array}{ccc}
	0 & 0 & 0 \\
	0 & 0 & 0 \\
	0 & 0 & 0 \\
	\end{array}
	\right)\,,\quad 
\ U^{0,2} =\left(
	\begin{array}{ccc}
	0 & 0 & 0 \\
	0 & 0 & 0 \\
	0 & 0 & 0 \\
	\end{array}
	\right)\,,\quad 
U^{0,3} =\left(
\begin{array}{ccc}
0 & 0 & 0 \\
0 & 0 & 0 \\
0 & 0 & 0 \\
\end{array}
\right)\,,
\end{align}
\begin{align}
 U^{1,0} =\left(
	\begin{array}{ccc}
	0 & 0 & 0 \\
	0 & 0 & 0 \\
	0 & 0 & 0 \\
	\end{array}
	\right)\,,\quad U^{1,1} =\left(
\begin{array}{ccc}
1 & 0 & 0 \\
0 & 0 & 0 \\
0 & 0 & 0 \\
\end{array}
\right)\,,\quad U^{1,2} =\left(
\begin{array}{ccc}
0 & 1 & 0 \\
0 & 0 & 0 \\
0 & 0 & 0 \\
\end{array}
\right)\,,\quad 
\ U^{1,3}=\left(
	\begin{array}{ccc}
	0 & 0 & 1 \\
	0 & 0 & 0 \\
	0 & 0 & 0 \\
	\end{array}
	\right)\,,
\end{align}
\begin{align}
 U^{2,0} =\left(
	\begin{array}{ccc}
	0 & 0 & 0 \\
	0 & 0 & 0 \\
	0 & 0 & 0 \\
	\end{array}
	\right)\,,\quad 
U^{2,1} =\left(
\begin{array}{ccc}
0 & 0 & 0 \\
1 & 0 & 0 \\
0 & 0 & 0 \\
\end{array}
\right)\,,\quad 
U^{2,2} =\left(
\begin{array}{ccc}
0 & 0 & 0 \\
0 & 1 & 0 \\
0 & 0 & 0 \\
\end{array}
\right)\,,\quad 
\ U^{2,3} =\left(
	\begin{array}{ccc}
	0 & 0 & 0 \\
	0 & 0 & 1 \\
	0 & 0 & 0 \\
	\end{array}
	\right)\,,
\end{align}
\begin{align}
U^{3,0} =\left(
\begin{array}{ccc}
0 & 0 & 0 \\
0 & 0 & 0 \\
0 & 0 & 0 \\
\end{array}
\right)\,,\quad 
\ U^{3,1} =\left(
	\begin{array}{ccc}
	0 & 0 & 0 \\
	0 & 0 & 0 \\
	1 & 0 & 0 \\
	\end{array}
	\right)\,,\quad
\  U^{3,2} =\left(
	\begin{array}{ccc}
	0 & 0 & 0 \\
	0 & 0 & 0 \\
	0 & 1 & 0 \\
	\end{array}
	\right)\,,\quad 
U^{3,3} =\left(
\begin{array}{ccc}
0 & 0 & 0 \\
0 & 0 & 0 \\
0 & 0 & 1 \\
\end{array}
\right)\,.
\end{align}
\end{subequations}
Next, we prove Eq.~\eqref{eq:action_1}. We have
\be
U\ket{\Omega}=\left(\sum\operatorname{tr}\left[Z U^{n_{1}, m_{1}} \ldots U^{n_{N}, m_{N}}\right] f_{1}^{n_{1}, m_{1}} \ldots f_{N}^{n_{N}, m_{N}}\right) \ket{\Omega}\,.
\ee
The vacuum is annihilated by $f^{n_j,m_j}$, unless $m_j=0$. Since $U^{n_j,0}$ is non vanishing only if $n_j=0$, there is a single non-vanishing term in the sum, corresponding to $n_j=m_j=0$ for all $j$ and
\be
U\ket{\Omega}=\left({\rm tr} \left[Z\, \left(\openone_3\right)^N\right]\right)\ket{\Omega}= \left(1+1-1\right)\ket{\Omega}=\ket{\Omega}\,.
\ee
Now, let us prove Eq.~\eqref{eq:action_2}. We have
\be
U  {a}^{j_1\dagger}_{i_1} {a}^{ j_2 \dagger}_{i_2}\cdots {a}^{j_P \dagger}_{i_P}\ket{\Omega}=\left(\sum\operatorname{tr}\left[Z U^{n_{1}, m_{1}} \ldots U^{n_{N}, m_{N}}\right] f_{1}^{n_{1}, m_{1}} \ldots f_{N}^{n_{N}, m_{N}}\right)   {a}^{j_1\dagger}_{i_1} {a}^{ j_2 \dagger}_{i_2}\cdots {a}^{j_P \dagger}_{i_P}\ket{\Omega}\,.
\ee
Now, consider $r\neq i_k$, for $k=1,2,\ldots, P$. Arguing as before, we obtain that the only non-vanishing terms in the sum correspond to $m_r=n_r=0$, for which $U^{n_r,m_r}=\openone_3$. Thus, the above expression simplifies to
\be
U  {a}^{j_1\dagger}_{i_1} {a}^{ j_2 \dagger}_{i_2}\cdots {a}^{j_P \dagger}_{i_P}\ket{\Omega}=\left(\sum\operatorname{tr}\left[Z U^{n_{i_1}, m_{i_1}} \ldots U^{n_{i_P}, m_{i_P}}\right] f_{i_1}^{n_{i_1}, m_{i_1}} \ldots f_{i_P}^{n_{i_P}, m_{i_P}}\right)   {a}^{j_1\dagger}_{i_1} {a}^{ j_2 \dagger}_{i_2}\cdots {a}^{j_P \dagger}_{i_P}\ket{\Omega}\,,
\ee
\end{widetext}
where the sum is now restricted over all the sequences $\{n_{i_\ell}\}_{\ell=1}^{P}$, $\{ m_{i_\ell}\}_{\ell=1}^{P}$ with $n_{i_\ell}, m_{i_\ell} =1,2,3$. From the explicit form of the local tensors, we see that this expression is, up to a sign, a translation of fermionic modes on an effective chain of $P$  sites, where to each site we associate three possible fermionic modes $a^{1}=a_2$, $a^{2}=a_1$, $a^{3}=a_1a_2$. This proves Eq.~\eqref{eq:action_2} is proven. We stress that we do not need to specify the sign $(-1)^\gamma$ appearing in Eq.~\eqref{eq:action_2}, because this does not affect unitarity.

\subsection{Majorana-shift operator}
\label{sec:majorana_shifts}

In this appendix we show that the fMPO $M^{(N)}_{AP}$ defined by Eq.~\eqref{eq:maj_anti} provides a valid representation for the translation of Majorana modes.

First of all, it is straightforward to verify that the fMPO corresponding to~\eqref{eq:maj_anti} is unitary. This can be seen, for instance, by constructing the matrices $\mathbb{U}^{k,i}$ in ~\eqref{eq:unitary_stacked}, and observing that they generate the identity operator. Next, we show
\begin{align}
M_{AP}^{(N)}\gamma_{2n}&=\gamma_{2n+1}M_{AP}^{(N)}\,,\quad n\neq N\,,\label{eq:to_prove_1}\\
M_{AP}^{(N)}\gamma_{2n-1}&=\gamma_{2n}M_{AP}^{(N)}\,, \label{eq:to_prove_2}\\
M_{AP}^{(N)}\gamma_{2N}& =-\gamma_{1}M_{AP}^{(N)}\,,\label{eq:to_prove_3}
\end{align}
where
\begin{align}
\gamma_{2n-1}&=a^{\dagger}_n+a_{n}=\tilde{\delta}_{i,j+1}f_n^{i,j}\,,\\
\gamma_{2n}&=i\left(a^{\dagger}_n-a_{n}\right)=i(-1)^j\tilde{\delta}_{i,j+1}f_n^{i,j}
\end{align}
are the Majorana modes introduced in Eq.~\eqref{eq:real_vs_majorana}, and where $f_{n}^{i,j}$ are given in Eq.~\eqref{eq:small_f_operators}. Here, we introduced the function
\be
\tilde{\delta}_{a,b}=
\begin{cases}
	1 & a\equiv  b \quad (\bmod \ 2)\,,\\
	0 & {\rm otherwise}\,.
\end{cases}
\ee
In order to prove Eqs.~\eqref{eq:to_prove_1}-\eqref{eq:to_prove_3} we can use the explicit form
\be
c^{i_{1}, \dots, i_{N}}_{j_{1}, \dots, j_{N}}=\frac{2}{\sqrt{2}}\frac{1}{2^{N/2}}i^{J}\tilde{\delta}_{I,J}\,,
\ee
where $|I|=|i_1|+\cdots |i_N|$, $|J|=|j_1|+\cdots |j_N|$. For instance, using the latter, the l.h.s. of~\eqref{eq:to_prove_1} can be written as (repeated indices are summed over)
\begin{widetext}
\begin{align}
M^{(N)}_{AP}\gamma_{2n}&=\frac{2}{\sqrt{2}}\frac{1}{2^{N/2}}i^{|j_1|+\cdots +|\ell_n|+\cdots +|j_N|}\tilde{\delta}_{I,|j_1|\cdots+|\ell_n|+\cdots +|j_N|} f^{i_1,j_1}\cdots f^{i_n,\ell_n}\cdots  f^{i_N,j_N}\tilde{\delta}_{\ell_n,j_n+1}i(-1)^{j_n}f^{\ell_n,j_n}\nonumber\\
&=\frac{2}{\sqrt{2}}\frac{1}{2^{N/2}}i^{J}(-i)^{|\ell_n|-|j_n|}\tilde{\delta}_{\ell_n,j_n+1}i(-1)^{j_n}(-1)^{\sum_{m>n}(|i_m|+|j_m|)}\tilde{\delta}_{I,J+1} f^{i_1,j_1}\cdots  f^{i_N,j_N}\nonumber\\
&=-\frac{2}{\sqrt{2}}\frac{1}{2^{N/2}}\tilde{\delta}_{I,J+1}i^{J} (-1)^{\sum_{m>n}(|i_m|+|j_m|)} f^{i_1,j_1}\cdots  f^{i_N,j_N}\,,
\label{eq:lhs}
\end{align}
while, in the same way, one can compute
\begin{align}
\gamma_{2n+1}M^{(N)}_{AP}&=-\frac{2}{\sqrt{2}}\frac{1}{2^{N/2}}\tilde{\delta}_{I,J+1}i^J(-1)^{\sum_{m>n}(|i_m|+|i_n|)}f^{i_1,j_1}\cdots  f^{i_N,j_N}\,,
\label{eq:rhs}
\end{align}
\end{widetext}
yielding~\eqref{eq:to_prove_1}. Eqs.~\eqref{eq:to_prove_2} and \eqref{eq:to_prove_3} can be proven in the same way. Thus, using unitarity, we obtain $M^{(N)}_{AP} \gamma_{n} M^{(N)\dagger}_{AP}=\gamma_{n+1}$ for $n\neq 2N$ and $M^{(N)}_{AP} \gamma_{2N} M^{(N)\dagger}_{AP}=-\gamma_{1}$.

Similarly, we can analyze the operator $M^{(N)}_P$ corresponding to the coefficients~\eqref{eq:maj_per}. First, constructing the fMPO representation for $M_P^{(N)}M_P^{(N)\dagger}$, and paying attention to the boundary operator $X$, it is easy to show that $M^{(N)}_P$ is unitary. Next, rewriting the coefficients~\eqref{eq:maj_per} as
\be
c^{i_{1}, \dots, i_{N}}_{j_{1}, \dots, j_{N}}=\frac{2}{\sqrt{2}}\frac{1}{2^{N/2}}(-i)^{J}\tilde{\delta}_{I,J+1}\,,
\ee
and by means of calculations similar to those reported in Eqs.~\eqref{eq:lhs}, \eqref{eq:rhs}, we can also show
\begin{align}
M_{P}^{(N)}\gamma_{n}&=\gamma_{n+1}M_{P}^{(N)}\,,
\end{align}
with the identification $\gamma_{2N+1}=\gamma_1$. Putting all together, we obtain that the fMPO corresponding to the coefficients \eqref{eq:maj_per} provides a valid representation for the Majorana shift with periodic boundary conditions.

\section{Graded canonical form for antiperiodic boundary conditions}
\label{sec:graded_canonical_form}

In this section we provide all the proofs that are needed in order to establish the existence and uniqueness of the graded canonical form for antiperiodic boundary conditions. Throughout this section, we will always work with even tensors, unless specified otherwise, and denote by $\ket{V^{(N)}(\mathcal{A})}$ the fMPS with antiperiodic boundary conditions generated by $\mathcal{A}$. Furthermore, we will denote by $Z$ the parity operator  acting on the space of the matrices~$A^i$.

First, recalling definition~\ref{def:git} for graded irreducible tensors, we introduce the Graded Irreducible Form (GIF) as follows.

\begin{defn}
	We say that an even tensor $\mathcal{A}$ generating an fMPS is in Graded Irreducible Form if: $(i)$: the matrices are of the form $A^n=\oplus_{k=1}^r \mu_k A_k^n $, where $\mu_k\in \mathbb{C}$ and the spectral radius of the transfer matrix $E_k$ associated with $A_k^n$ is equal to one; ($ii$) the parity operator $Z$ has the same block structure as $A^n$ and, for all $k$, $\mathcal{A}_k$ is a graded irreducible (even) tensor. 
\end{defn}

\begin{prop}
	\label{prop:GIF}
	Given any even tensor $\mathcal{A}$, one can always find another even tensor $\mathcal{B}$ in GIF generating the same fMPS with antiperiodic boundary conditions.
\end{prop}
\begin{proof}
Let $Z$ be the parity operator on the auxiliary space. First, if there are no non-trivial invariant graded subspaces, $\mathcal{A}$ is a graded irreducible tensor, and thus already in GIF. Otherwise, take $P$ to be the orthogonal projector onto an invariant graded subspace that does not contain any other non-trivial invariant graded subspace, and define $Q=\openone -P$. Since $[P,Z]=0$, setting $Z_P=PZP$ and $Z_Q=QZQ$ we have $Z_P (PA^{i}P)=(-1)^{|i|}(PA^{i}P)Z_P$, $Z_Q (QA^{i}Q)=(-1)^{|i|}(QA^{i}Q)Z_Q$. Furthermore, $Z_P^2=P$, $Z_Q^2=Q$. Thus, it is immediate to see that the tensors $A^{i}$ and $PA^{i}P+ QA^{i}Q$ generate the same state, and both $PA^{i}P$ and $QA^{i}Q$ generate valid fMPSs, with parity operators on the auxiliary spaces given by $Z_P$ and $Z_Q$, respectively. We can now consider $QA^{i}Q$ and decompose it into smaller blocks using the same steps. This procedure can be iterated until we end up with a tensor in  GIF.	
\end{proof}

The following statement already appeared in Ref.~\cite{bultinck2017fermionic}, but we give here a detailed proof.

\begin{prop}\label{prop:degenerate_GIF}
	Let $\mathcal{A}$ be a graded irreducible tensor, such that there exists a non-graded invariant subspace for all $A^i$. Denoting by $D$ the dimension of the matrices $A^{i}$, and by $D_e$, $D_o$ the dimensions of  the even and odd subspaces, we have $D=2D_e$, and there exists an even gauge transformation such that 
	\begin{subequations}\label{eq:GIF_explicit}
		\begin{align}
		A^i & = \left(\begin{matrix} B^i & 0 \\ 0 & B^i \end{matrix}\right) = \openone \otimes B^i\qquad \text{ if } |i| = 0 \\
		A^i & = \left(\begin{matrix} 0 & B^i \\ B^i & 0 \end{matrix}\right) = \sigma^x \otimes B^i \qquad \text{ if } |i| = 1\,,
		\end{align}
		\label{eq:graded_canonical_form_deg}
	\end{subequations}
	and  $Z=\sigma^z\otimes \openone$. Furthermore, the following are true:
	\begin{enumerate}
		\item there exists an index $i$, such  that $|i|\equiv 1$ and $B^i \neq 0$; \label{condition:1}
		\item there is no projector $P$ such that
		\bea
		B^{i_1}\cdots B^{i_p}P&=& PB^{i_1}\cdots B^{i_p}P \qquad	\forall \{i_k\}_{k=1}^p\,,
		\label{eq:projector_bi}
		\eea
		namely there is no invariant subspace for the algebra generated by the matrices $B^i$; \label{condition:2}
		\item there is no projector $P_e$ such that
		\bea
		B^{i_1}\cdots B^{i_p}P_e&=& P_eB^{i_1}\cdots B^{i_p}P_e\nonumber\\
		\forall \{i_k\}_{k=1}^p&:& \sum_{r=1}^p |i_r|\equiv 0\,,
		\label{eq:projector_even_bi}
		\eea
		namely there is no invariant subspace for the even  subalgebra generated by the matrices $B^i$. \label{condition:3}
	\end{enumerate} 
\end{prop}
\begin{proof}
	Let $P$ be the orthogonal projector onto a proper non-graded invariant subspace of $A^{j}$, $\mathcal{S}_1$, containing no smaller non-trivial invariant subspace. Defining $Q = Z P Z$ we have $[P,Z] \neq 0$, and thus, if $Q$ projects onto $\mathcal{S}_2=Z\mathcal{S}_1$, we have $\mathcal{S}_1\neq \mathcal{S}_2$. If $Y$ is an even invertible matrix, then $A_Y^j=YA^jY^{-1}$ leave invariant the subspaces $YS_1$, $YS_2$ with corresponding orthogonal projector $P_Y$, $Q_Y$. Note that $P_Y\neq Y P Y^{-1}$, since $Y P Y^{-1}$ is not Hermitian and that, since $Y$ is even, $Q_Y=ZP_YZ$. We claim that it is always possible to find an even $Y$ such that $P_Y Q_Y=P_Y ZP_YZ=0$.
	
	 To show this, it is useful to construct a Jordan decomposition of the Hilbert space for $P$, $Q$ into one one- and two-dimensional orthogonal graded subspaces that are invariant under both $P$ and $Q$. To this end, we note that $\Pi=P+Q$ is Hermitian and commutes with $Z$. Hence, there is a common basis of eigenstates. Let $\ket{\phi}$ be an element of this basis, so that $\Pi\ket{\phi}=\lambda \ket{\phi}$ and $Z\ket{\phi}=\pm \ket{\phi}$. If $P\ket{\phi}\in {\rm span}(\ket{\phi})$, then ${\rm span}(\ket{\phi})$ is a one-dimensional invariant graded subspace for both $P$, $Q$. Otherwise, it is easy to see that $\ket{\phi}$, $P\ket{\phi}$ generate a two-dimensional space left invariant by $P$ and $Q$. If $\lambda\neq 0$, this space is generated by $P\ket{\phi}$, $ZP\ket{\phi}$, and taking the even and odd combination of them we see that the subspace is graded. Otherwise $P\ket{\phi}=\pm ZP\ket{\phi}$ so that $P\ket{\phi}$ has also well-defined parity, and so in any case the subspace is graded. This procedure gives us a basis of the Hilbert space $\mathcal{B}=\{\ket{v_1},\ket{w_1}, \ldots , \ket{v_r}, \ket{w_r}, \ket{u_1}, \ldots \ket{u_k}\}$, where $(\ket{v_j},\ket{w_j})$ and $\ket{u_j}$ generate the two- and one-dimensional invariant graded subspaces, respectively. Note that this basis is not orthogonal since $\braket{v_j|w_j}\neq 0$.
	
	Now, let $R_i$ be the orthogonal projector onto a graded invariant subspace for $P$, $Q$ of dimension $2$. Since $[R_i,Z]=0$, we have that $R_iPR_i=\ket{v_i}\bra{v_i}$ and $R_iQR_i=\ket{w_i}\bra{w_i}$,  with $\ket{w_i}=Z\ket{v_i}$, and $\ket{v_i}=\alpha_i\ket{a_i}+\beta_i\ket{b_i}$, where $\ket{a_i}$ ($\ket{b_i}$) is even (odd). It must  be $\alpha_i,\beta_i \neq 0$, otherwise $\ket{v_i}$, $\ket{w_i}$ are proportional. Then, we can define the two-dimensional matrix $Y_i={\rm diag}(\beta_i,\alpha_i)$ which is invertible. Next, for all one-dimensional blocks we define $Y_i=\openone$. We claim that $Y=\bigoplus_i Y_i$ is the desired gauge matrix. Indeed, consider the basis $\mathcal{B}_Y=\{Y \ket{v_1}, Y\ket{ w_1}, \ldots ,Y\ket{ v_r}, Y\ket{ w_r},\ket{ u_1}, \ldots,\ket{u_k}\}$. By construction, this is an orthogonal (although not normalized) basis: indeed, $\braket{v_j|Y^{\dagger} Y|w_j}= |\alpha_i\beta_i|^2-|\alpha_i\beta_i|^2=0$, while other orthogonality relations are immediate. Furthermore, in this basis $P_Y$, and $Q_Y$ are diagonal matrices, with elements that are only $0$ or $1$. Thus, since $P_Y\neq P_Q$, necessarily $P_YQ_Y=0$: if this is not the case, since $P_Y$, and $P_Q$ have the same rank, then $P_YQ_Y$ would project onto an invariant subspace of $\mathcal{S}_1$ with strictly smaller dimension, which is a contradiction. 
	
	Thus, up to an even gauge transformation, we can assume that $PQ=0$. Now, since $P+Q$ is an invariant subspace projector that commute with $Z$, it must be equal to the identity, because $\mathcal{A}$ is graded irreducible. In the basis where $Z$ has the form~\eqref{eq:parity_operator}, this imposes the following structure on $P$ and $Q=\openone-P=Z P Z$
\begin{align}
P&=\frac{1}{2} \begin{bmatrix}
\openone & U\\
U^\dagger & \openone
\end{bmatrix},&Q&=\frac{1}{2} \begin{bmatrix}
\openone & -U\\
-U^\dagger & \openone
\end{bmatrix}\,,
\end{align}
where idempotence requires $U U^\dagger = U^\dagger U = \openone$. Hence, this is only possible if $D_e=D_o$ and $U$ is a unitary matrix. Now, the condition $A^iP=PA^iP$ implies that the matrices $A^i$ are of the following form
\begin{equation}\label{eq:oddalgebrageneric}
\begin{split}
A^i &= \left(\begin{matrix} C^i & 0 \\ 0 & U^\dagger C^i U \end{matrix}\right) \;\;\;\;\text{ if } |i| = 0 \\
A^i &= \left(\begin{matrix} 0 & C^i \\ U^\dagger C^i U^\dagger & 0 \end{matrix}\right) \;\;\;\;\text{ if } |i| = 1\, .
\end{split}
\end{equation}
We can then map this to the form~\eqref{eq:graded_canonical_form_deg} using an even-parity gauge $\openone \oplus U$. 
	
	Next, condition~\ref{condition:2} follows from \ref{condition:3}, so we only need to prove the former. We do this by contradiction. We denote by $S^{(e)}(B^{i})$ the even subalgebras generated by $B^{i}$, while we define $S^{(o)}(B^{i})$ to be the linear space generated by the odd products  $B^{i_1}\ldots B^{i_k}$ with $\sum_j|i_j|\equiv 1$ ($\bmod$ $2$) (note that $S^{(o)}(B^{i})$  is not an algebra, since it is not closed under product). Let $P_e$ be the projector onto an invariant subspace for $S^{(e)}(B^{i})$, denoted by $\mathcal{S}_e$, containing no smaller invariant subspace, and define the linear space 
	\be
	\mathcal{S}_o=\{\ket{w}= B^{(o)}\ket{v_e}:\ket{v_e}\in \mathcal{S}_e, B^{(o)}\in S^{(o)}(B^{i})\}\,.
	\ee
	Namely $\mathcal{S}_o$ is the space generated by applying all possible odd products to $\mathcal{S}_e$. We denote by $P_o$ the orthogonal projector onto $\mathcal{S}_o$. By definition
	\be\label{eq:step_1}
	B^{i}P_e=P_oB^iP_e\,, \quad \forall |i|=1\,,
	\ee
	and also
	\be\label{eq:step_3}
	B^{i}P_o=P_eB^iP_e\,, \quad \forall |i|=1\,.
	\ee
	To see this, we note that $\ket{w}\in\mathcal{S}_o\Rightarrow \ket{w}= B^{(o)}\ket{v_e}$, with $\ket{v_e}\in \mathcal{S}_e$, and $B^{(o)}\in S^{(o)}(B^{i})$, and thus, $B^{i}\ket{w}=B^{i} B^{(o)}\ket{v_e}$. Since $|i|=1$, $ B^{i} B^{(o)}$ is in the even algebra, and thus $ B^{i} B^{(o)}\ket{v_e}\in \mathcal{S}_e$ due to Eq.\eqref{eq:projector_even_bi}. In the same way, we also have
	\be\label{eq:step_2}
	B^{i}P_o=	P_o B^{i}P_o\,, \quad \forall |i|=0\,.
	\ee
	Indeed, if $\ket{w}\in\mathcal{S}_o$, then $\ket{w}= B^{(o)}\ket{v_e}$, with $\ket{v_e}\in \mathcal{S}_e$, and $B^{(o)}\in S^{(o)}(B^{i})$, and thus $B^{i}\ket{w}=B^{i}B^{(o)}\ket{v_e}$.  Since $|i|=0$, $ B^{i} B^{(o)}$ is odd, and thus by definition of $\mathcal{S}_o$, $B^{i}\ket{w}\in \mathcal{S}_o$. Now, consider the projector
	\be
	\tilde{P}=
	\left(\begin{matrix} P_e & 0 \\ 0 & P_o \end{matrix}\right)\,. 
	\ee
	Clearly, $[\tilde{P},Z]=0$. Furthermore,  $A^{i}\tilde{P}=\tilde{P}A^{i}\tilde{P}$, because this condition is equivalent to
	\begin{align}\label{invariantP}
	B^i P_e &= P_e B^i P_e, \forall |i|=0,& B^i P_o &= P_o B^i P_o, \forall |i|=0,\nonumber\\
	B^i P_e &= P_o B^i P_e, \forall |i|=1,& B^i P_o &= P_e B^i P_o, \forall |i|=1.
	\end{align}
	The first equation is verified by definition, while the others correspond to ~\eqref{eq:step_1}, \eqref{eq:step_3}, \eqref{eq:step_2}. The existence of $\tilde{P}$ contradicts the assumptions, and we thus conclude that also property~\ref{condition:3} is true.

\end{proof}

\begin{cor}\label{prop:other_implication}
Let $\mathcal{A}$ be an even tensor in the form~\eqref{eq:GIF_explicit}. Then, $\mathcal{A}$ is graded irreducible if and only if there is no projector $P_e$ such that
\bea
B^{i_1}\cdots B^{i_p}P_e&=& P_eB^{i_1}\cdots B^{i_p}P_e\nonumber\\
\forall \{i_k\}_{k=1}^p&:& \sum_{r=1}^p |i_r|\equiv 0\,.
\eea
\end{cor}
\begin{proof}
Due to Prop.~\ref{prop:degenerate_GIF}, we only need to prove that if there is no such $P_e$ then $\mathcal{A}$ is graded irreducible. Suppose this is not true,  and take a projector $P$ with $A^i P= PA^i P$, and $[P,Z]=0$. From the latter condition, we can write $P$ in the form
\be
P=
\begin{pmatrix}
	P_e & 0 \\
	0 & P_o
\end{pmatrix}\,.
\ee
Using the explicit matrix representation for $A^{i}$, it is immediate to see that $P_e$ is an invariant projector for the even matrix subalgebra generated by $B_i$, which is the desired contradiction. 
\end{proof}

We can now proceed to characterize graded normal tensors, introduced in the definition~\ref{def:graded_normal_tensor}. We begin with a simple observation.

\begin{lem}\label{prop:degenerate_CGF}
	Let $\mathcal{A}$ be a graded normal degenerate tensor. Then, up to an even gauge transformation, $A^i$ are in the form~\eqref{eq:graded_canonical_form_deg}, and there exist two indices $i$, $j$ with $|i|=1$, $|j|=0$, such that $B^{i}\neq 0$, $B^{j}\neq 0$.
\end{lem}
\begin{proof}
Since the transfer matrix associated with $\mathcal{A}$ has two eigenvalues equal to $1$, then necessarily there are invariant subspaces for $A^i$ which, by definition, can only be non-graded. Thus, due to Prop.~\eqref{prop:degenerate_GIF}, there is an even gauge transformation mapping $A^i$ in the form~\eqref{eq:graded_canonical_form_deg}, and there is $i$, with $|i|=1$, such that $B^{i}\neq 0$. . Hence, we only need to prove that there exists one index $i$, with $|i|=0$, such that $B^{i}\neq 0$. Suppose that there is no such $i$, and take $A^i$ in the form~\eqref{eq:graded_canonical_form_deg}. Then, $A^j=\sigma^x\otimes  B^j$ for all $j$. Denoting by $E_A$, $E_B$, the transfer matrices associated with $A^j$ and $B^j$, respectively, this implies $E_A=\sigma^x\otimes \sigma^x\otimes E_B$, in a suitable basis. It follows that the spectral radius of $E_A$ and $E_B$ coincide. Thus, there exists an eigenstate of $E_B$ associated with an eigenvalue, $\lambda$, with $|\lambda|=1$, and hence there are four eigenstates of $E_A$ associated with an eigenvalues $\mu_j$ with $|\mu_j|=1$. This contradicts the fact that $\mathcal{A}$ is a graded normal degenerate tensor.
\end{proof}

Next, we introduce the completely positive map associated with a tensor $\mathcal{A}$, which is of great importance for the study of normal tensors. 
\begin{defn}
	For any tensor $\mathcal{A}$ we define the CPM
	\be
	\mathcal{E}_\mathcal{A}(X)=\sum_{j}A^{j}XA^{j\dagger}\,.
	\label{eq:quantum_channel}
	\ee
\end{defn}
Analogously to the case of qudits, it is possible to characterize the fixed points of the CPM associated with graded normal tensors. 
\begin{lem}\label{lemma:fixed_points}
	Let $\mathcal{A}$ be a graded normal tensor. If $\mathcal{A}$ is non-degenerate, then the unique fixed point of $\mathcal{E}_\mathcal{A}$ is an even operator, which is strictly positive. If $\mathcal{A}$ is degenerate, then one can choose the two fixed points to be a pair consisting of  an even and an odd operator. Furthermore, in the standard graded basis where $Z=\openone\oplus (-\openone)$, they can be chosen of the form
	\be
	\rho^e=
	\begin{pmatrix}
		\rho& 0\\
		0 & \rho
	\end{pmatrix}\,,
	\quad
	\rho^o=
	\begin{pmatrix}
		0& \rho\\
		\rho& 0
	\end{pmatrix}\,,
	\label{eq:fixed_points}
	\ee
	where $\rho$ is  a  strictly positive operator.
\end{lem}
\begin{proof}
	First, we note that  one can always choose the fixed points of the map \eqref{eq:quantum_channel} to  have a well defined parity. Now, if $\mathcal{A}$ is non-degenerate, then its unique fixed point $X$ is a strictly positive operator~\cite{evans_spectral_1978}. This implies ${\rm tr}[X]>0$. On the other hand, the trace of  any odd operator is vanishing. Hence,  $X$ is even.
	
	Suppose now $\mathcal{A}$ is a degenerate graded normal tensor. Then, we can assume wlog that $A^{j}$ is in the form in~\eqref{eq:GItensors}. Define the unitary matrix
	\be
	u=\frac{1}{\sqrt{2}}
	\begin{pmatrix}
		\openone & \openone\\
		-\openone & \openone
	\end{pmatrix}\,.
	\label{eq:similarity_u}
	\ee
	Note that  $u$ does not commute with the parity $Z$. We have
	\be
	\widetilde{A}^j=u A^ju^{\dagger}=
	\begin{pmatrix}
		B^{i} & 0\\
		0 & C^{i}
	\end{pmatrix}\,,
	\label{eq:diagonal_form_GIF}
	\ee
	where
	\be
	C^{i}=
	\begin{cases}
		B^{i} & |i|\equiv 0\,,\\
		-B^{i} & |i|\equiv 1\,.
		\label{eq:definition_c}
	\end{cases}
	\ee
	Defining the CPM
	\be
	\mathcal{E}_{\widetilde{\mathcal{A}}}(X)= \sum_{j}\widetilde{A}^j X \widetilde{A}^{j\dagger}\,,
	\ee
	we have $\mathcal{E}_{\widetilde{\mathcal{A}}}(X)=u \mathcal{E}_{\mathcal{A}}(u^{\dagger}Xu)u^{\dagger}$, so that $\mathcal{E}_{\mathcal{A}}$ and $\mathcal{E}_{\widetilde{\mathcal{A}}}$ have the same spectrum and the eigenstates are related by a similarity transformation. Next, due to Prop.~\ref{prop:degenerate_GIF}, the tensor $\mathcal{B}$ defined by $B^{i}$ is irreducible. Furthermore, the map $\mathcal{E}_{\mathcal{B}}$ can not have eigenvalues $\lambda$ with $\lambda\neq 1$ and $|\lambda|=1$, otherwise they would be also eigenvalues of $\mathcal{E}_\mathcal{A}$. Hence, $\mathcal{B}$ is normal, and the map $\mathcal{E}_{\mathcal{B}}$ has a  unique fixed point $\rho>0$.  Furthermore, the maps $\mathcal{E}_\mathcal{B}$, $\mathcal{E}_\mathcal{C}$ clearly coincide (where the tensor $\mathcal{C}$ is defined by the matrices $C^i$). Accordingly
	\be
	\rho_1=
	\begin{pmatrix}
		\rho & 0\\
		0 & 0
	\end{pmatrix}\,,
	\quad
	\rho_2=
	\begin{pmatrix}
		0& 0\\
		0 & \rho
	\end{pmatrix}\,,
	\ee
	are two eigenvectors of $\mathcal{E}_{\widetilde{\mathcal{A}}}(\cdot)$ with $\lambda=1$. Since $\mathcal{A}$ is normal, these must be the only fixed points of $\mathcal{E}_{\widetilde{\mathcal{A}}}(\cdot)$. Taking now the inverse similarity transformation with respect to $u$, we obtain that the eigenspace associated with $\lambda=1$ of the map $\mathcal{E}_\mathcal{A}(\cdot)$ is spanned by the matrices
	\bea
	\rho^{e}=\begin{pmatrix}
		\rho& 0\\
		0 & \rho
	\end{pmatrix}\,,\quad 
	\rho^{o}=\begin{pmatrix}
		0& \rho\\
		\rho& 0
	\end{pmatrix}\,.
	\eea
\end{proof}

\begin{defn}
	For graded normal non-degenerate tensors, we say that $\mathcal{A}$ is in Graded Canonical Form II (GCFII) if $\Phi=\openone$ is the only fixed point of 
	\be
	\mathcal{E}^{\prime}_\mathcal{A}(X)=\sum_{j}A^{\dagger j}XA^{j}\,,
	\label{eq:transposed_quantum_channel}
	\ee
 	while the fixed point of the CPM~\eqref{eq:quantum_channel} is a  diagonal positive and full-rank matrix $\rho$ . For graded normal degenerate tensors we say that $\mathcal{A}$ is in GCFII if 
	\bea\label{eq:phi_12}
	\Phi_1=\begin{pmatrix}
		\openone& 0\\
		0 & \openone
	\end{pmatrix}\,,\quad 
	\Phi_2=\begin{pmatrix}
		0& \openone\\
		\openone& 0
	\end{pmatrix}
	\eea
	are the only fixed points of  the CPM~\eqref{eq:transposed_quantum_channel}, while the fixed points of ~\eqref{eq:quantum_channel} are given in ~\eqref{eq:fixed_points}, where $\rho$ is a diagonal positive and full-rank matrix.
\end{defn}
From Lemma~\ref{lemma:fixed_points} it is immediate to show that for any normal tensor there is always an even gauge transformation mapping it into GCFII. 

\begin{prop}\label{prop:fixed_point_rank}
	Let $\mathcal{A}$ be an even tensor, and denote by $E_\mathcal{A}$, $\mathcal{E}_\mathcal{A}$ the corresponding transfer matrix and the associated completely positive map, respectively. Then
	\begin{enumerate}
		\item $\mathcal{A}$ is normal non-degenerate iff  $(i)$ $\mathcal{E}_\mathcal{A}$ has a unique eigenvalue $\lambda$ with $|\lambda|=1$; $(ii)$ the corresponding left and right eigenvectors $\Phi$, $\rho$ of the transfer matrix are positive definite operators $\rho$, $\Phi>0$.\label{point:1}
		\item $\mathcal{A}$ is normal degenerate iff  $(i)$ $\mathcal{E}_\mathcal{A}$ has exactly two eigenvectors associated with eigenvalue $\lambda=1$,  and no other eigenvalue $\mu$ with $|\mu|=1$; $(ii)$ the corresponding right (left) eigenvectors $\rho_1$, $\rho_2$ ($\Phi_1$, $\Phi_2$) of the transfer matrix are even and odd, respectively; $(iii)$ the even eigenvectors $\Phi_1$, $\rho_1$ are positive definite operators $\rho>0$. \label{point:2}
	\end{enumerate}
\end{prop}
\begin{proof}
Point~\ref{point:1} follows directly from Prop. $3$ in Ref.~\cite{sanz2010quantum}. Let us prove point~\ref{point:2}, along similar lines. We only need to prove that conditions $(i)$, $(ii)$, $(iii)$ imply that $\mathcal{A}$ is normal degenerate, because the inverse statement follows directly from Lemma~\ref{lemma:fixed_points}. First, we assume wlog that $\mathcal{A}$ is in the form~\eqref{eq:graded_canonical_form_deg} (with the parity operator $Z=\sigma^z\otimes \openone$) and in GCFII. This means that the channel~\eqref{eq:quantum_channel} is trace preserving. We denote by $D_e$, $D_o$ the dimensions of the even and odd subspaces, and $D=D_e+D_o$. We denote by $\{\ket{j}\}_{j=1}^D$ a basis of vectors with well-defined parity.  Following~\cite{sanz2010quantum}, we define $S_n(A)\subseteq \mathcal{M}_{D\times D}$ as the linear space spanned by all possible products of exactly $n$ matrices $A^j$, where $\mathcal{M}_{D\times D}$ is the space of complex $D\times D$ matrices. Similarly, we define $S^{e}_n(A)\subseteq M_{D\times D}$ as the linear space spanned all possible even products of $A^j$,  $A^{j_1}\ldots A^{j_n}$, with $\sum_{k}|j_k|\equiv 0$ ($\bmod$ 2). Finally, given the map $\mathcal{E}$, we introduce the Choi matrix $\omega(\mathcal{E}):=\left( \mathcal{E}\otimes \openone\right)(\Omega)$, where $\Omega=\sum_{i,j=1}^D|ii\rangle\langle jj|$. Since the channel is trace-preserving, the eigenvalues of the transfer matrix $E_{\mathcal{A}}$ with $|\lambda|=1$ have trivial Jordan blocks (see Prop. 6.2 of Ref.~\cite{wolf2012quantum}). Then, from $(i)$ we have the 
\be
\lim_{n\to\infty} \mathcal{E}^n_A= \mathcal{E}^{\infty}_A\,,
\ee
and
\be\label{eq:e_decomposition}
\mathcal{E}^{\infty}_A(X)=\rho_1 {\rm tr}\left(X\right)+ \rho_2 {\rm tr}\left(\Phi_2 X\right)\,,
\ee
where $\Phi_2$ is given in Eq.~\eqref{eq:phi_12}.

We need to show that there exists no proper graded subspace which is left invariant by all $A^j$. First, we note that each element, $A^e\in S^{e}_n(A)$ is a block diagonal matrix
\be
A^e=
\begin{pmatrix}
	A^e_{1} & 0\\
	0 & A^e_{2}
	\end{pmatrix}\,,
\ee
where $A^e_{1}\in \mathcal{M}_{D_{e}\times D_e}$, $A^e_{2}\in \mathcal{M}_{D_{o}\times D_o}$. We claim that there exists some $n$ such that $P_e S^{e}_n(A) P_e= \mathcal{M}_{D_e\times D_e}$, where $P_e$ is the projector onto the even subspace of dimension $D_e$. This amounts to show that the matrices $A^e_1$ span $\mathcal{M}_{D_e\times D_e}$. Suppose that this is not the case. Then, there exists
\be\label{eq:block_diagonal_representation}
F_n=
\begin{pmatrix}
	G_n & 0\\
	0 & 0
\end{pmatrix}\,,
\ee
such that ${\rm tr}(A^{(n)}_k F_n)=0$, for all $A^{(n)}_k \in S^{(e)}_n(A)$ (and so for all $A^{(n)}_k \in S_n(A)$, since $F_n$ is even). Thus
\begin{align}\label{eq:intermediate_step}
\left|{\rm tr}\left(\rho_1 F_n^\dagger F_n\right)\right|=\left|\sum_{k_1,\ldots , k_n}|{\rm tr}\left(A_{k_1}\cdots A_{k_n} F_n\right)|^2\right.\nonumber\\
\left.-{\rm tr}\left(\rho_1 F_n^\dagger  F_n\right)\right|\nonumber\\
=\left|{\rm tr}\left[\Omega \left(\mathcal{E}^n_A\otimes {\openone}\right)\left(\tilde{F}_n\Omega \tilde{F}_n^\dagger\right)\right]-{\rm tr}\left(\rho_1 F_n^\dagger  F_n\right)\right|\,,
\end{align}
where $\tilde{F}_n=F_n\otimes \openone$. We claim
\be
{\rm tr}\left(\rho_1 F_n^\dagger  F_n\right)={\rm tr}\left[\Omega \left(\mathcal{E}^{\infty}_A\otimes {\openone}\right)\left(\tilde{F}_n\Omega \tilde{F}_n^\dagger\right)\right]\,.
\ee
Indeed, using Eq.~\eqref{eq:e_decomposition}, we have
\begin{align}
{\rm tr}\left[\Omega \left(\mathcal{E}^{\infty}_A\otimes {\openone}\right)\left(\tilde{F}_n\Omega \tilde{F}_n^\dagger\right)\right]&= {\rm tr}\left(\rho_1 F_n^\dagger F_n\right)\nonumber\\
&+{\rm tr}\left(\rho_2 F_n^\dagger \Phi_2 F_n\right)\,.
\end{align}
Now, using the fact that $\rho_2$, $\Phi_2$ are odd operators, it is straightforward to show ${\rm tr}\left(\rho_2 F_n^\dagger \Phi_2 F_n\right)=0$, which follows  from~\eqref{eq:block_diagonal_representation}, and the fact that odd matrices are block off-diagonal in this basis. Then, from Eq.~\eqref{eq:intermediate_step}, we have
\begin{align}
\left|{\rm tr}\left(\rho_1 F_n^\dagger  F_n\right)\right| \leq & c_n || \Omega ||_{\infty} {\rm tr}\left(\tilde{F}_n\Omega \tilde{F}^\dagger_n\right)
=Dc_n {\rm tr}\left( F_n F^\dagger_n\right) \,,
\end{align}
where $\lim_n c_n=0$. On the other hand, using that if $\rho$ is full rank, then ${\rm tr} (\rho X)\geq \frac{1}{||\rho^{-1}||}{\rm tr}(X)$  for all $X\geq 0$, we have
\be
\left|{\rm tr}\left(\rho_1 F_n^\dagger  F_n\right)\right|\geq \frac{1}{||\rho_1^{-1}||} \left|{\rm tr}\left(F_n^\dagger  F_n\right)\right|\,,
\ee
and we obtain a contradiction. In the same way, we can prove that there exists some $n$ such that $P_o S^{e}_n(A) P_o= \mathcal{M}_{D_o\times D_o}$, where $P_o$ is the projector onto the odd subspace of dimension $D_o$. 

Now, suppose there exists a graded subspace $V\simeq \mathbb{C}^{D'}\subset \mathbb{C}^{D}$, with $D^\prime<D$, which is left invariant by all matrices $A^i$. Since $V$ is graded, we can take a basis for $V$ of the form
$\{\ket{v_1},\ldots, \ket{v_r}, \ket{w_1}, \ldots , \ket{w_s} \}$, where $\ket{v_j}$ are even, $\ket{w_j}$ are odd. Furthermore, since $D^\prime<D$, either $r<D_e$, or $s<D_o$. Suppose wlog that the former is true, and choose $\ket{u}$ even with $\braket{u|v_j}=0$ for all $j$. Taking now $n$ such that $P_e S^{e}_n(A) P_e= \mathcal{M}_{D_e\times D_e}$, there exists $A^{(n)}_k\in S^{(e)}_n(A)$ such that $\braket{u| A^{(n)}_k|v_1}\neq 0$. This is a contradiction, since we assumed that $V$ is left invariant by all matrices $A^j$. 

\end{proof}

\begin{cor}\label{cor:stability}
	Let $\mathcal{A}$ be a degenerate (non-degenerate) graded normal tensor. Then the blocked tensor $\mathcal{A}_k$ is still a degenerate (non-degenerate) graded normal tensor.
\end{cor}

We are finally in a position to prove the existence and uniqueness of the Graded Canonical Form introduced in the definition~\ref{def:graded_canonical_form}. We begin with the following Theorem, which provides a procedure to cast any even tensor $\mathcal{A}$ into GCF.

\begin{thm}\label{th:GCF}
	After blocking, for any even tensor, $\mathcal{A}$, it is always possible to obtain another even tensor, $\mathcal{A}_{GCF}$, in GCF and generating the same fMPS with antiperiodic boundary conditions. 
\end{thm}

\begin{proof}
	
By Prop.~\ref{prop:GIF} we can assume wlog that $\mathcal{A}$ is in GIF. Furthermore, from Corollary ~\ref{cor:stability}, any graded normal block remains such after blocking, so that we can restrict ourselves to study the case where the tensor $\mathcal{A}$ has a single block in its GIF. We can also assume wlog that the spectral radius of $\mathcal{E}_\mathcal{A}$ is one. There are only two possibilities for the tensor $\mathcal{A}$: $(i)$ there is no invariant subspace for the matrices $A^{j}$ or $(ii)$ there are non-graded invariant subspaces for the matrices $A^{j}$.

Consider case $(i)$. If there is a single eigenvalue with $|\lambda|=1$, then $\lambda=1$ and we are done. Otherwise, there are $p$ eigenvalues $e^{i 2\pi q/p }$, with ${\rm gcd}(q,p)=1$, where $p$ divides the bond dimension $D$~\cite{evans_spectral_1978}. We can then block $p$ times and consider the blocked tensor $\mathcal{A}_p$. Now, the CPM associated with $\mathcal{A}_p$ has exactly $p>1$ eigenvalues equal to $1$. Accordingly, there is necessarily an invariant subspace for $A_p^i$. If there are invariant graded subspaces, we can decompose the block further with a projector $P$, s.t. $[P,Z]=0$ and restart the procedure for each block. Otherwise we fall into case $(ii)$, detailed below.
	
Consider then case $(ii)$, namely suppose that $\mathcal{A}$ has a single block in the GIF, with non-graded invariant subspaces for the matrices $A^{j}$. If $A^{j}$ is a degenerate graded normal tensor we are finished. If this is not the case, we show below that the matrices $A^{j}$ can be decomposed into strictly smaller blocks after blocking. Since the bond dimension $D$ is finite, this is enough to conclude the proof.
	
	From Prop.~\ref{prop:degenerate_GIF}, we can assume without loss of generality that $A^{j}$ has the form~\eqref{eq:graded_canonical_form_deg}. Since there is no invariant subspace for the matrices $B^i$, cf. Prop.~\ref{prop:degenerate_GIF}, the CPM
	\be
	\mathcal{E}_\mathcal{B}(X)=\sum_j B^jX B^{j\dagger}
	\label{eq:b_map}
	\ee 
	can either have a single eigenvalue $\lambda$ with $|\lambda|=1$, [we call this case $(j)$] or exactly $p$ eigenvalues $e^{i 2\pi q/p}$, with ${\rm gcd}(q,p)=1$, where $p$ divides the bond dimension  $D/2$ [we call this case $(jj)$]. Consider case $(jj)$ and take the tensor $\mathcal{A}_p$ obtained by blocking $p$ times. The matrices corresponding to $\mathcal{A}_p$ are
	\begin{subequations}
		\begin{align}
		A^I_p & = \left(\begin{matrix} B_p^I & 0 \\ 0 & B_p^I \end{matrix}\right) = \openone \otimes B_p^I\qquad \text{ if } |I| = 0\,, \\
		A^I_p & = \left(\begin{matrix} 0 & B_p^I \\ B_p^I & 0 \end{matrix}\right) = \sigma^x \otimes B_p^I \qquad \text{ if } |I| = 1\,.
		\end{align}
		\label{eq:graded_canonical_form_blocked}
	\end{subequations}
	where $I=(i_{1},\ldots, i_p )$, $|I|=|i_1|+\cdots + |i_p|$ and
	\be
	B_{p}^{I}=B^{i_1}\cdots B^{i_p}\,.
	\ee
	Now, since $\mathcal{E}_{\mathcal{B}_p}$ has $p>1$ eigenvalues equal to one, there must be an invariant subspace for all $B_{p}^{I}$. Thus, from Prop.~\ref{prop:degenerate_GIF} the matrices $A^I_p$ can be decomposed further into smaller blocks with projectors commuting with  $Z$. Consider now $(j)$. In this case, the tensor $\mathcal{B}$ is normal. Repeating the construction of Lemma~\ref{lemma:fixed_points}, we consider the tensor $\widetilde{\mathcal{A}}$ in Eq.~\eqref{eq:diagonal_form_GIF}, which is related to $\mathcal{A}$ via the similarity transformation $u$ in Eq.~\eqref{eq:similarity_u}. It is  clear from the diagonal structure of  \eqref{eq:diagonal_form_GIF} that $\mathcal{E}_{\widetilde{\mathcal{A}}}$ (and thus $\mathcal{E}_{\mathcal{A}}$) have at least two eigenvalues equal to one. Since we are assuming that $\mathcal{A}$ is not normal, there must be other eigenvalues with absolute value one. It follows from Lemma A.2 in Ref.~\cite{cirac2017matrix} that this is only possible if there is a non-singular matrix $S$ and a phase  $\phi$ such that
	\be
	B^{i}=e^{i\phi} S C^{i} S^{-1}\,,
	\ee
	namely, using Eq.~\eqref{eq:definition_c},
	\begin{subequations}
		\bea
		B^{i}&=&e^{i\phi} S B^{i} S^{-1}\,,\quad |i|=0\\
		B^{i}&=&-e^{i\phi} S B^{i} S^{-1}\,,\quad |i|=1\,.
		\eea
		\label{eq:commutations}
	\end{subequations}
	Since $S$ is invertible, there is $\mu\neq 0$, $\ket{v}\neq 0$ such that
	\be
	S\ket{v}=\mu\ket{v}\,.
	\ee
	Since $\mathcal{B}$ is normal, for  any $n$ there is $B^{i_1}\cdots B^{i_n}$ such that $B^{i_1}\cdots B^{i_n}\ket{v}\neq 0$. By making repeated use of Eq.~\eqref{eq:commutations}, we get
	\be
	S B^{i_1}\cdots B^{i_n}\ket{v}=(-1)^{\sum_{j}|i_j|} e^{i\phi n}\mu\, B^{i_1}\cdots B^{i_n}\ket{v}\,
	\ee
namely, $S$ has an eigenvalue of the form $\pm e^{i\phi n}\mu$ for all $n$. Since $S$ is a finite-dimensional matrix, this is only possible if $\phi=2\pi q/p$ with $p,q \in \mathbb{N}$ and ${\rm g cd}(q,p)=1$. Blocking $p$ times, we now obtain a new tensor in the form~\eqref{eq:graded_canonical_form_blocked}, where now
\begin{subequations}
	\bea
	B^{I}&=& S B^{I} S^{-1}\,,\quad |I|=0\\
	B^{I}&=&-S B^{I} S^{-1}\,,\quad |I|=1\,.
	\eea
	\label{eq:commutations_blocked}
\end{subequations}

Now we prove by contradiction that the even subalgebra generated by the matrices $B^I$ has invariant subspaces, so that the blocked tensor $\mathcal{A}_p$ can be decomposed further by Prop.~\ref{prop:degenerate_GIF}. If this is not true, then by Burnside's theorem~\cite{Lomonosov2004the} the even subalgebra coincides with the full matrix algebra. On the other hand, from Eq.~\eqref{eq:commutations_blocked}, we have that $S$ must commute with any element of the even algebra, and thus $S= \alpha \openone$, where wlog $\alpha\neq 0$.  Thus
\bea
B^{I}&=&- B^{I}\; \forall  |I|=1\Rightarrow B^{I}=0\,, \forall |I|=1\,.
\eea
By Prop.~\ref{prop:degenerate_GIF}, this means that $\mathcal{A}_p$ has invariant graded subspaces, and using Corollary~\ref{prop:other_implication} we arrive at a contradiction.
\end{proof}

Having established the existence of the CGF for any even tensor $\mathcal{A}$, we now prove that the GCF is essentially unique. Our strategy follows closely the one for qudits~\cite{cirac2017matrixRenormalization}.

\begin{defn}
	The even tensors $\mathcal{A}_j$ ($j=1,\ldots,g$) form a basis of graded normal tensors (BGNT) of a tensor $\mathcal{A}$ if:
	(i) $\mathcal{A}_j$ are graded normal (degenerate or non-degenerate); (ii) for each $N$, $\ket{V^{(N)}(\mathcal{A})}$ can be written as a linear combination of $\ket{V^{(N)}(\mathcal{A}_j)}$; (iii) there exists some $N_0$ such that for all $N>N_0$, $\ket{ V^{(N)}(\mathcal{A}_j)}$ are linearly independent. 
\end{defn}

\begin{lem}
\label{equalMPS}
Let $\ket{V_{a,b}}$ be two fMPSs (with antiperiodic boundary conditions) generated by two graded NTs $\mathcal{A}_{a,b}$, with $D_{\alpha}\times D_{\alpha}$  matrices $A^{i}_{\alpha}$ and  parity $Z_{\alpha}$.  Then
	\begin{subequations}
		\bea
		\lim_{N\to\infty} \langle V_\alpha|V_\alpha\rangle&=&c_{\alpha},	\label{norm}\\
		\label{scprod}
		\lim_{N\to\infty} |\langle V_b|V_a\rangle|&=& \text{ 0 or } c_a,
		\eea
\end{subequations}
where $c_{a}=1$ if $\mathcal{A}_a$ is non-degenerate, and $c_{a}=2$ if it is degenerate. In the case where the limit~\eqref{scprod} is non-vanishing, $\mathcal{A}_{a,b}$ are either both non-degenerate or both degenerate. Furthermore, there exist an invertible matrix $X$, a permutation matrix $\Pi$ and a phase $\phi$ such that $A^i_a= e^{i\phi} X \Pi A^i_b \Pi^{-1} X^{-1}$, with $[X,Z_a]=0$, and $\Pi^{-1} Z_a\Pi=\pm Z_b$. 
\end{lem}
\begin{proof}
Eq.~\eqref{norm} is obvious, so we only need to prove~\eqref{scprod}. Suppose that $\mathcal{A}_{a,b}$ are both non-degenerate, and that the limit~\eqref{scprod} is non-vanishing. Then, it follows from Lemma A.2 in Ref.~\cite{cirac2017matrixRenormalization} that $D_a=D_b$, and there exist a phase $\phi$ and a non-singular matrix $Y$, such that $A^i_a=e^{i\phi}YA^{i}_bY^{-1}$. Now, define $\zeta_a=YZ_bY^{-1}$, and $S_a=Z_a\zeta_a$. It is immediate to see that $S_a$ commutes with $A^{i}_a$ for all $i$. Hence, since $A^{i}_a$ is normal, $S_a=\alpha \openone$, namely 
\be\label{eq:z_a}
Z_a=\alpha Y Z_b Y^{-1}\,.
\ee 
Since $Z_a^2=Z_b^2=\openone$, we have $\alpha=\pm 1$.  Furthermore, since $Z_a$ and $\alpha Z_b$ are both diagonal and related by a similarity transformation, they have the same number of $1$'s and $0$'s on the diagonal and there exists a permutation matrix $\Pi$ s.t. $\alpha Z_b= \Pi^{-1} Z_a \Pi$. Plugging this into~\eqref{eq:z_a}, we get $[Z_a,Y\Pi^{-1}]=0$, and the statement follows.

Suppose now that $\mathcal{A}_a$ is a graded normal degenerate tensor, so that wlog $A^i_a$ are in the form~\eqref{eq:graded_canonical_form_deg}. First, we prove by contradiction that if $\mathcal{A}_b$ is non-degenerate, then $\lim_{N\to\infty} |\langle V_b|V_a\rangle|=0$. Suppose that this is not true, and that $\lim_{N\to\infty} |\langle V_b|V_a\rangle|=c\neq 0$. Using a similarity transformations that does not have well-defined parity, $\mathcal{A}_a$ can be brought into the block-diagonal form~\eqref{eq:diagonal_form_GIF}. Then, it follows from Lemma A.2 in Ref.~\cite{cirac2017matrixRenormalization} that there exists an invertible matrix $X$ and a phase $\phi$, such that 
\be
A^i_b=e^{i\phi} X B^i X^{-1}\ \ {\rm or}\ \ A^i_b=e^{i\phi} X (-1)^{|i|}B^i X^{-1}\,.
\label{eq:similarity}
\ee
Since $\mathcal{A}_a$ is a normal degenerate tensor,  the even algebra generated by $B^{i}$ does not have invariant subspaces, and so the same is true for $A^i_b$ due to Eq.~\eqref{eq:similarity}. On the other hand,  we have $[Z_b, A^{j_1}_b\cdots  A^{j_n}_b]=0$ for all products $A^{j_1}_b\cdots  A^{j_n}_b$ with $\sum_k|j_k|\equiv 0$ (mod $2$). Since there are no invariant subspaces for the even algebra, this implies $Z_b=\pm \openone$. By the parity of the tensor $\mathcal{A}_b $, it must $A_b^{|i|}\equiv 0$ for all $i$ with $|i|=0$ (mod $2$), or $A_b^{|i|}\equiv 0$ for all $i$ with $|i|=1$  (mod $2$). Using Eq.~\eqref{eq:similarity} and Prop.~\ref{prop:degenerate_CGF}, we arrive at a contradiction. 

Finally, suppose that both $\mathcal{A}_a$, $\mathcal{A}_b$ are degenerate normal tensors, and that the limit~\eqref{scprod} is non-vanishing. Applying a permutation and an even gauge transformation, we can cast both tensors in the form Eq.~\eqref{eq:graded_canonical_form_deg}, namely $A^i_{a,b}=(\sigma^{x})^{|i|}\otimes B^{i}_{a,b}$. Reasoning as before it follows that either $B^i_b=e^{i\phi} Y B^i_a Y^{-1}$, or $B^i_b=e^{i\phi} (-1)^{|i|} Y B^i_a Y^{-1}$ for some invertible matrix $Y$ and $\phi\in\mathbb{R}$. In both cases, the statement easily follows.
\end{proof}

\begin{rmk}
{\normalfont Note that permutations do not necessarily have a well-defined parity. However, they leave the parity operator $Z$ in diagonal form, which is a necessary condition in order to represent a state defined by~\eqref{eq:fMPS_apbc} as an fMPSs. In fact, it is immediate to show that, if $G$ is a gauge transformation such that $GZG^{-1}$ is diagonal, then $G=\Pi X$, where $[X,Z]=0$, and $\Pi$ is a permutation.}
\end{rmk}

We have now all the necessary ingredients to state the following fundamental theorem for fMPSs.

\begin{thm}
	\label{thm1}
	Let $\mathcal{A}$ and $\mathcal{B}$ be two tensors in GCF, with BGNT $A^i_{k_a}$, $B^i_{k_b}$ ($k_{a,b}=1,\dots,g_{a,b}$), and corresponding parity operators $Z^a_{k_a}$ and $Z^b_{k_b}$, respectively. If for all $N$, $\mathcal{A}$ and $\mathcal{B}$ generate fMPSs that are proportional to each other, then: (i) $g_a=g_b=:g$; (ii) for all $k$ there exist $j_k$, $\phi_{k}$,  an invertible matrix $X_k$, and  a permutation $\Pi_k$ such that $B^i_k= e^{i\phi_k} X_k \Pi_k A^i_{j_k} \Pi_k^{-1} X_k^{-1}$, with $[X_k, Z^b_k]=0$, $\Pi_k^{-1} Z^b_k\Pi_k=\pm Z^a_{j_k}$.
	\end{thm}

Equipped with Lemma~\ref{equalMPS}, the proof of this statement follows without modifications the one presented in Ref.~\cite{cirac2017matrixRenormalization} for Theorem 2.10. Finally, Theorem~\ref{thm2_fundamental} follows as a simple corollary.  Note that, using the same notation as in the statement of   Theorem~\ref{thm2_fundamental}, we have that $Z_a$ and $\Pi Z_b \Pi^{-1}$ have the same block-diagonal structure as $\mathcal{A}$, and in each block they coincide up to a global minus sign.

\section{fMPUs with antiperiodic boundary conditions}
\label{sec:technical_results_antiperiodic}

\subsection{GCF of fMPUs with ABC}

In this appendix we provide the technical proofs of the statements presented in Sec.~\ref{sec:generalized_fMPUs}. We begin with the characterization of the GCF of tensors generating fMPUs with antiperiodic boundary conditions.

\begin{prop}\label{prop:GCF_fMPUs}
	Let $\mathcal{U}$ be in GCF, and suppose $\mathcal{U}$ generates a type I (type II) fMPU $U^{(N)}$ with antiperiodic boundary conditions. Then, $\mathcal{U}/\sqrt{d}$ is graded normal non-degenerate (degenerate).
\end{prop}

\begin{proof}
Let us first consider the case where $U^{(N)}$ is an fMPU of the first kind, namely as in Eq.~\eqref{eq:generalized_fMPU} with $S=e^{i\alpha}\openone$, $\alpha\in \mathbb{R}$. Wlog we can assume that the tensor $\mathcal{U}$ is in GCF. We show that $\mathcal{U}/\sqrt{d}$ is necessarily graded normal non-degenerate. Indeed, since $U^{(N)}$ is unitary for all $N$, $(1/d)^N {\rm tr}(U^{(N)\dagger}U^{(N)})=1$, namely, using Eq.~\eqref{eq:trace_transfer_antiperiodic}, ${\rm tr}E^N_{\mathcal{U}}=1$ for all $N>1$. Using Lemma A.5 in Ref.~\cite{cirac2017matrixRenormalization}, it follows that there is only one nonzero eigenvalue of $E_{\mathcal{U}}$, which is $1$, and thus the GCF of $\mathcal{U}/\sqrt{d}$ contains only one normal non-degenerate block. 

Next, suppose $U^{(N)}$ is an fMPU of the second kind. Then, repeating the above argument, we get ${\rm tr}E^N_{\mathcal{U}}=2$ for all $N>1$, where the factor $2$ comes from the normalization $1/\sqrt{2}$ in Eq.~\eqref{eq:second_kind_apb}. Using again Lemma A.5 in Ref.~\cite{cirac2017matrixRenormalization}, we have only two possibilities: the GCF of $\mathcal{U}/\sqrt{d}$ has two normal non-degenerate blocks or only a single normal degenerate block. Let us assume that the former is true, and arrive at a contradiction. In this case, we can decompose $\mathcal{U}=\mathcal{V}\oplus \mathcal{W}$, where both $\mathcal{V}$, $\mathcal{W}$ are even, and thus $U^{(N)}=(V^{(N)}+W^{(N)})e^{i\alpha}/\sqrt{2}$, where  $V^{(N)}$, $W^{(N)}$ are standard fMPOs with antiperiodic boundary  conditions. Thus
\begin{align}
\openone=U^{(N)\dagger} U^{(N)}&=\frac{1}{2}V^{(N)\dagger} V^{(N)}+\frac{1}{2}W^{(N)\dagger} W^{(N)}\nonumber\\
&+\frac{1}{2}V^{(N)\dagger} W^{(N)}+\frac{1}{2}W^{(N)\dagger} V^{(N)}\,.
\label{eq:partial_1}
\end{align}
Let us fix the value of $N$. We have the formal  expansion
\be
V^{(N)\dagger}V^{(N)}=c_{1}\openone+\sum  c^{\alpha_{1},\ldots \alpha_{j}}_{i_1, \ldots i_j}A^{\alpha_1}_{i_1}\cdots A^{\alpha_j}_{i_j}\,.
\label{eq:partial_2}
\ee
Here, $A^{\alpha}_{j}$ are traceless local operators (with well-defined parity) that, together with $\openone/\sqrt{d}$, form an orthonormal basis of local operators corresponding to site $j$, namely ${\rm tr}(A_j^{\alpha\dagger} A_k^{\beta})=\delta_{\alpha,\beta}\delta_{j,k}$, while $c_1$, $c^{\alpha_{1},\ldots \alpha_{j}}_{i_1, \ldots i_j}$ are some complex coefficients. Note that, since $V^{(N)\dagger}V^{(N)}$ is even, the sum in Eq.~\eqref{eq:partial_2} is over all possible sequences $\{\alpha_j\}$, $\{i_j\}$ such that $\sum_j|\alpha_j|\equiv 0$ ($\bmod$ 2), where $|\alpha_j|$ is the parity of $A^{\alpha_j}$. Similar expansions hold for $W^{(N)\dagger} W^{(N)}$, $V^{(N)\dagger} W^{(N)}$, $W^{(N)\dagger} V^{(N)}$. Now, we can multiply Eq.~\eqref{eq:partial_1} by $1/{d^N}$ and take the trace: using that $\mathcal{V}/{\sqrt{d}}$, $\mathcal{W}/{\sqrt{d}}$ are normal non-degenerate with spectral radius equal to $1$, and Lemma A.5 in Ref.~\cite{cirac2017matrixRenormalization}, we immediately have that $(1/d^N){\rm tr}\left[V^{(N)\dagger}V^{(N)}\right]=1$, $(1/d^N){\rm tr}\left[W^{(N)\dagger}W^{(N)}\right]=1$, $(1/d^N){\rm tr}\left[W^{(N)\dagger}V^{(N)}\right]=0$, $(1/d^N){\rm tr}\left[V^{(N)\dagger}W^{(N)}\right]=0$. In particular $c_1=1$ in Eq.~\eqref{eq:partial_2}. Next, we show that $c^{\alpha_{1},\ldots \alpha_{j}}_{i_1, \ldots i_j}$ in Eq.~\eqref{eq:partial_2} are all vanishing. To this end, we multiply both sides of Eq.~\eqref{eq:partial_1} by $A^{\alpha_j \dagger}_{i_j} \cdots A^{\alpha_1 \dagger}_{i_1}$, with $\sum_j|\alpha_j|\equiv 0$ ($\bmod$ 2), and take the trace, obtaining
\be
0={\rm tr}\left[M_{V,V}\right]+{\rm tr}\left[M_{W,W}\right]+{\rm tr}\left[M_{W,V}\right]+{\rm tr}\left[M_{V,W}\right]\,.
\label{eq:sum_trace}
\ee
Here $M_{V,V}$ is a product of matrices $\tilde{V}[i_j,\alpha_j]$ with elements $\tilde{V}_{\beta,\gamma}[i_j,\alpha_j]=\sum_{n,m} b^{n,m}_{\beta}[i_j,\alpha_j] \mathbb{V}_{\beta,\gamma}^{n,m}$, where $\mathbb{V}_{\beta,\gamma}^{n,m}$ corresponds to the tensor generating $V^{(N)\dagger}V^{(N)}$, as defined in Eq.~\eqref{eq:unitary_stacked}, while $b^{n,m}_{\beta}[i_j,\alpha_j]$ are complex numbers, determined by the choice of $A^{\alpha_j \dagger}_{i_j}$. The matrices $M_{W,W}$, $M_{W,V}$, $M_{V,W}$ are defined similarly. Note that there is no additional parity operator $Z$ in the trace, which is due to the fact that $A^{\alpha_j \dagger}_{i_j} \cdots A^{\alpha_1 \dagger}_{i_1}$ is  even. Note also that ${\rm tr}\left[M_{V,V}\right]= c^{\alpha_{1},\ldots \alpha_{j}}_{i_1, \ldots i_j}$. Now, we consider a chain of $kN$ sites, for $k\geq 1$. We repeat the steps above, but multiply by
\begin{align}
A^{\alpha_{j} \dagger}_{i_j+(k-1)M} \cdots A^{\alpha_1 \dagger}_{i_1+(k-1)M} A^{\alpha_j \dagger}_{i_j+(k-2)M} \nonumber\\
\cdots A^{\alpha_1 \dagger}_{i_1+(k-2)M} \cdots A^{\alpha_j \dagger}_{i_j} \cdots A^{\alpha_1 \dagger}_{i_1}\,.
\end{align}
We get
\be
0={\rm tr}\left[M_{V,V}^{k}\right]+{\rm tr}\left[M_{W,W}^{k}\right]+{\rm tr}\left[M_{W,V}^{k}\right]+{\rm tr}\left[M_{V,W}^{k}\right]\,,
\ee
for all $k\geq 1$. Using once again Lemma A.5 in Ref.~\cite{cirac2017matrixRenormalization}, we conclude that each trace in Eq.~\eqref{eq:sum_trace} is vanishing, and thus $c^{\alpha_{1},\ldots \alpha_{j}}_{i_1, \ldots i_j}$ in Eq.~\eqref{eq:partial_2} is zero as anticipated. In conclusion, the operators $V^{(N)}$, $W^{(N)}$ satisfy
\begin{align}
V^{(N)\dagger}V^{(N)}=W^{(N)\dagger}W^{(N)}&=\openone\,,\label{eq:contradiction_1}\\
V^{(N)\dagger}W^{(N)}=W^{(N)\dagger}V^{(N)}&=0\,,\label{eq:contradiction_2}
\end{align}
where the second equality follows from the argument above, and the fact that ${\rm tr}\left[V^{(N)\dagger}W^{(N)}\right]={\rm tr}\left[W^{(N)\dagger}V^{(N)}\right]=0$. However, Eqs.~\eqref{eq:contradiction_1}, \eqref{eq:contradiction_2} are inconsistent with each other, and we have arrived at the desired contradiction. 
\end{proof}

We are now in a position to prove that any tensor $\mathcal{U}$ generating one of the fMPUs introduced in Sec.~\ref{sec:gfMPUs_two_kinds} is necessarily simple, according to the definition~\ref{def:simpleness_fermionic}.

\begin{prop}
	Suppose that the tensor $\mathcal{U}$ generates a type I (type II) fMPU $U^{(N)}$ with antiperiodic boundary conditions. Then, there exists $k\leq D^4$ such that $\mathcal{U}_k$ is simple (according to the definition~\ref{def:simpleness_fermionic}).
\end{prop}
\begin{proof}
	The case of type I fMPUs can be treated by  following the same steps as in the proof of Prop. III.3 in Ref.~\cite{cirac2017matrix}, so here we will only consider type II fMPUs. Furthermore, in order to simplify the notation, we will make use of the formalism of graded tensor networks, explained in Appendix~\ref{sec:graded_TN}. 
	
	Let us consider a tensor $\mathcal{U}$ generating a type II fMPU with antiperiodic boundary conditions. Due to Prop.~\ref{prop:GCF_fMPUs}, we can assume wlog that $\mathcal{U}$ is a degenerate graded normal tensor. Using the graphical notation explained in Appendix~\ref{sec:graded_TN}, we can write
	\be
	\begin{tikzpicture}[baseline={([yshift=-0.5ex]current bounding box.center)}, scale=0.7]
	
	\draw (0,-0.5) node[left]{} -- (0,0.5) node[left]{};
	\draw (0,0.5) -- (0,1.5) node[left]{};
	\draw (-0.5,0) node[left]{} -- (0.5,0) node[right]{};
	\draw (-0.5,1) node[left]{} -- (0.5,1) node[right]{};
	
	\filldraw[fill=white,draw=black] (0,0) circle [radius=0.25];
	\fill[black] (0,1) circle [radius=0.25];
	
	\end{tikzpicture}
	=E\gtimes \openone+\sum_{\alpha}\mathbb{S}_{\alpha}\gtimes \sigma_{\alpha}\,.
	\label{eq:decomposition_unitary}
	\ee
	We stress that here and throughout this proof, pictures correspond to graded operators, and not to matrix elements. In particular, in the rhs the first part of the graded tensor products acts on the auxiliary indices, whereas the second part on the physical ones.  Furthermore $E$ is the transfer matrix associated with  $\mathcal{U}$. Finally, we choose $\sigma_{\alpha}$ and $\mathbb{S}_{\alpha}$ to be traceless operators with well-defined parity, and with $\sigma_{\alpha}$  mutually orthogonal. By Lemma~\ref{lemma:fixed_points}, the transfer matrix $E$ has two eigenvalues equal to $1$.  Furthermore, there are no Jordan blocks associated to the eigenvalues $1$, whereas there may be one or several Jordan blocks associated to the zero eigenvalues. Denoting by $J< D^2$ the largest size of all Jordan blocks of $E$, we can block $J$ sites and consider $\mathcal{U}_J$.  Since $\mathcal{U}_J$ also generates a type II fMPU, we can use for it the decomposition \eqref{eq:decomposition_unitary}, where we will denote by $E'$ and $\mathbb{S}'_\alpha$ the new operators.  
	
Next, we multiply both sides of  $\openone=U^{(N)\dagger}U^{(N)}$ by $(\sigma^{\alpha_1}_{1}\gtimes \ldots  \gtimes \sigma^{\alpha_m}_{m})^{\dagger}$, with $\sum_{j}|\alpha_j|\equiv 0$ ($\bmod$ 2), where $|\alpha_j|$ is the parity of the operator $\sigma^{\alpha_m}_{m}$. We obtain
	\be
	{\rm tr}(E'S'_{\alpha_1}\ldots S'_{\alpha_m})=0\,,
	\label{eq:ESE=0}
	\ee
	where $S^\prime_{\alpha}$ is  the matrix associated with $\mathbb{S}^\prime_{\alpha}$, namely
	\be
	\mathbb{S}^{\prime}_{\alpha}=\left(S_{\alpha}\right)_{x,y} |x)\gtimes (y|\,.
	\ee
	Note that there is no additional parity operator $Z$ in Eq.~\eqref{eq:ESE=0}, since $\sum_{j}|\alpha_j|\equiv 0$ ($\bmod$ 2).
	\begin{widetext}
Now, we can assume wlog that $\mathcal{U}_{J}$ is in GCFII, namely that transfer matrix associated with $\mathcal{U}_J$ is in the form
	\be
	E^{\prime}=|\rho^{(1)})(\Phi^{(1)}|+|\rho^{(2)})(\Phi^{(2)}|\,,
	\ee
	where
	\begin{align}
	(\Phi^{(1)}|=&\sum_{n=0}^{D-1}(n,n|\,,\quad (\Phi^{(2)}|=\sum_{n=0}^{D-1}(n,D-1-n|\,,\\
	|\rho^{(1)})=&\sum_{n=0}^{D-1}\rho_n|n,n)\,,\quad |\rho^{(2)})=\sum_{n=0}^{D-1}  \rho_n|n,D-1-n)\,,
	\end{align}
	with $\rho_n>0$. We show
	\be
	(\Phi^{(a)}|S^{\prime}_{\alpha_1}\ldots S^{\prime}_{\alpha_m}|\rho^{(b)})=0\,, \quad {\rm for }\ a\neq b\,,
	\label{eq:condition}
	\ee
	for all sequences $\{\alpha_j\}$. Note that we prove this also for $\sum_j|\alpha_j|\equiv 1$ ($\bmod$ 2). In particular, Eq.~\eqref{eq:condition} can be established by using the explicit representation~\eqref{eq:decomposition_unitary}. First, we note
	\be
	\mathbb{S}^\prime_{\alpha}=\sum_{\{i_n,k_n\}}c(\{i_n,k_n\})
	\begin{tikzpicture}[baseline={([yshift=-0.5ex]current bounding box.center)}, scale=0.7]
	\draw (0,-0.5) node[left]{$i_1$} -- (0,0.5) node[left]{$j_1$};
	\draw (0,0.5) -- (0,1.5) node[left]{$k_1$};
	\draw (-0.5,0) node[left]{} -- (0.5,0) node[right]{};
	\draw (-0.5,1) node[left]{} -- (0.5,1) node[right]{};
	\filldraw[fill=white,draw=black] (0,0) circle [radius=0.25];
	\fill[black] (0,1) circle [radius=0.25];	
	\draw(1,0)node{$\ldots$};
	\draw(1,1)node{$\ldots$};
	\draw (2,-0.5) node[left]{$i_J$} -- (2,0.5) node[left]{$j_J$};
	\draw (2,0.5) -- (2,1.5) node[left]{$k_J$};
	\draw (1.5,0) node[left]{} -- (2.5,0) node[right]{};
	\draw (1.5,1) node[left]{} -- (2.5,1) node[right]{};
	\fill[black] (2,1) circle [radius=0.25];	
	\filldraw[fill=white,draw=black] (2,0) circle [radius=0.25];
	\end{tikzpicture}\,.
	\ee
Here $c(\{i_n,k_n\})$ are numerical coefficients arising from expanding $\sigma_\alpha$	in Eq.~\eqref{eq:decomposition_unitary} in terms of the operators $\ket{l}\gtimes \bra{m}$ and reordering. Next, recall that, due to Eq.~\eqref{eq:second_kind_apb}, $U^{j,k}=(\sigma^{x})^{|j|+|k|}\otimes B^{j,k}$, and 
	\bea
	\begin{tikzpicture}[baseline={([yshift=-0.5ex]current bounding box.center)}, scale=0.7]
	
	\draw (0,-0.5) node[left]{$i$} -- (0,0.5) node[left]{$j$};
	\draw (0,0.5) -- (0,1.5) node[left]{$k$};
	\draw (-0.5,0) node[left]{$\gamma$} -- (0.5,0) node[right]{$\delta$};
	\draw (-0.5,1) node[left]{$\alpha$} -- (0.5,1) node[right]{$\beta$};
	
	\filldraw[fill=white,draw=black] (0,0) circle [radius=0.25];
	\fill[black] (0,1) circle [radius=0.25];
	
	\end{tikzpicture}
	&=&
	 \mathbb{U}^{k,i}_{(\alpha,\gamma),(\beta,\delta)}|\gamma)|\alpha)|k\rangle\langle i|(\beta|(\delta|\,,
	\eea
	where $\mathbb{U}^{k,i}_{(\alpha,\gamma),(\beta,\delta)}$ was defined in Eq.~\eqref{eq:unitary_stacked}, and corresponds to the matrix
	\be
	\mathbb{U}^{k,i}=\sum_{j}\left(U^{\ast j,k}\otimes U^{j,i}\right)\left(Z^{(|i|+|k|)}\otimes\openone\right)\,,
	\label{eq:compact_representation}
	\ee 
	where $Z=\sigma^z\otimes \openone$. We can now expand $(\Phi^{(a)}|S^{\prime}_{\alpha_1}\ldots S^{\prime}_{\alpha_m}|\rho^{(b)})$ using the above expressions, arriving at

		\begin{align}
		(\Phi^{(a)}|S^{\prime}_{\alpha_1}\ldots S^{\prime}_{\alpha_m}|\rho^{(b)})=\sum_{\{i_n,j_n,k_n\}} \Lambda(\{i_n,j_n,k_n\}) {\rm tr}\left[\left(\sigma^{x}\right)^{a+1}\left(\sigma^{x}\right)^{\sum_s( |i_s|+|k_s|)}\left(\sigma^{z}\right)^{\sum_s( |i_s|+|k_s|)}\left(\sigma^{x}\right)^{b+1}\right]\,,
		\label{eq:calculation}
		\end{align}
	\end{widetext}
	where $\Lambda(\{i_n,j_n,k_n\})$ is a coefficient that does not depend on $a$ and $b$, and whose exact form is irrelevant for our discussion. Here, we used that the eigenstates  $(\Phi^{(a)}|$ and $|\rho^{(a)})$, when viewed as operators, correspond  to the matrices $\left(\sigma^{x}\right)^{a+1}\otimes \openone$ and $\left(\sigma^{x}\right)^{a+1}\otimes \rho$ (importantly, $\rho$ does not depend on $a$). Now,  if  $a\neq b$ the traces in the rhs are all zero, because they are either proportional to ${\rm tr}(\sigma^x)$ or ${\rm tr}(\sigma^z)$, and Eq.~\eqref{eq:condition} follows immediately. Furthermore, since $\Lambda(\{i_n,j_n,k_n\})$ does not depend on $a$, $b$, we also have
	\be
	(\Phi^{(1)}|S^{\prime}_{\alpha_1}\ldots S^{\prime}_{\alpha_m}|\rho^{(1)})= (\Phi^{(2)}|S^{\prime}_{\alpha_1}\ldots S^{\prime}_{\alpha_m}|\rho^{(2)})\,.
	\label{eq:condition_2}
	\ee
	It follows now from Eqs.~\eqref{eq:ESE=0},\eqref{eq:condition},~\eqref{eq:condition_2} that
	\be
	E^{\prime} S^\prime_{\alpha_1}\ldots  S^\prime_{\alpha_m}E^{\prime} =0\,,
	\label{eq:final_step}
	\ee
	if $\sum_{j}|\alpha_j|\equiv 0$ ($\bmod $ 2). On the other hand, Eq.~\eqref{eq:final_step} is also true if $\sum_{j}|\alpha_j|\equiv 1$ ($\bmod $ 2). This is because Eq.~\eqref{eq:condition} still holds, and $\braket{\Phi^{(j)}|\mathcal{O}|\rho^{^{(j)}}}=0$ if $\mathcal{O}$ is odd, since $\bra{\Phi^{(j)}}$ and $\ket{\rho^{(j)}}$ are either both even or both odd. 
	
	Now, any element, $S$, in the algebra generated by $S^\prime_{\alpha}$ must have zero eigenvalues. Indeed, whether $S$ is even or odd, $S^{ 2}$ is even, and thus ${\rm tr}(S^{2N})=\sum_{j}\lambda_j^{2N}=0$ for all $N$, where $\lambda_j$ are the eigenvalues of $S$. This means that $\lambda_j^2=0$ for all $j$, and thus $\lambda_j=0$ for all $j$. Accordingly, any element, $S$, in the algebra generated by $S^\prime_{\alpha}$ is nilpotent. It follows then from a result by Nagata and Higman~\cite{nagata1952,higman1952ordering}, improved later by Razmyslov~\cite{razmyslov1974trace}, that there exists some $J^\prime < D^2$ such that
	\be
	S^{\prime}_{\alpha_1}\ldots S^{\prime}_{\alpha_{J^\prime}}=0\,,
	\ee
	for any set of $\alpha$'s. At this point the proof can be completed by following without modification the one of Prop.~III.3 in Ref.~\cite{cirac2017matrix}. 
\end{proof}

\begin{cor}\label{cor:local_dimension_degenerate}
Let $\mathcal{U}$ be a type II fMPU, and denote by $d$, $d_e$ and $d_o$ the dimensions of the local physical space, and the corresponding even and odd subspaces, respectively. Then, $d$ is even, and $d_e=d_o$. 
\end{cor}
\begin{proof}
We prove this by contradiction. First, note that if $d_e\neq d_o$, then the same is true by blocking arbitrarily many times. This can be simply proven by induction on the number  of blocked sites, $k$, and using the rearrangement inequality. Then, take $k$ such that $\mathcal{U}_k$ is simple, and assume wlog that $\mathcal{U}_k$ is in the form~\eqref{eq:graded_canonical_form_deg}. Due to the structure of the eigenstates of the transfer matrix associated with eigenvalue $1$, one can construct an fMPU $U^{(N)}_p$ with periodic boundary conditions by adding an operator $X=\sigma^x\otimes \openone$ into the trace (unitarity follows from simpleness and the fact that $(X\otimes X )E_{\mathcal{U}}=E_{\mathcal{U}}$). It is immediate to see that $U^{(N)}_p$ is an odd operator. However, this is a contradiction, because for any $N$ the dimensions of the even and odd subspaces are different, and there can be no odd invertible operator.
\end{proof}

Next, we show that the fMPUs introduced in Sec.~\ref{sec:gfMPUs_two_kinds} feature a $\mathbb{Z}_2$ structure with respect to composition. In fact, it is obvious that the product of two type I fMPUs is still a type I fMPU. Analogously, using Prop.~\ref{prop:normality_fmpus}, we have that the product of a type I and a type II fMPU is still a type II fMPU. In the following, we also show that the product of two type II fMPUs can be represented as a type I fMPU. The proof follows closely the logic of similar derivations presented in Ref.~\cite{bultinck2017fermionic}, and shows that the class of fMPUs introduced in Sec.~\ref{sec:gfMPUs_two_kinds} is closed with respect to composition.

\begin{prop}\label{prop:z2_structure}
	Let $\mathcal{U}$, $\mathcal{V}$ be two degenerate graded normal tensors, generating type II fMPUs $U^{(N)}$, $V^{(N)}$. Then, there exists a representation of $U^{(N)}V^{(N)}$ as a type I fMPU $\forall N$.
\end{prop}
\begin{proof}
	We denote by $Z_U$, $Z_V$ the parity operators acting on the auxiliary space for $\mathcal{U}$, $\mathcal{V}$, respectively, and assume wlog that they are in the form~\eqref{eq:parity_operator}. Let $\mathcal{W}$ be the tensor obtained by composing $\mathcal{U}$ and $\mathcal{V}$. Up to an even similarity transformation, we have
	\begin{align}
	W^{k,i}_{(\alpha,\gamma),(\beta,\delta)}=& \sum_j(-1)^{|\gamma|(|k|+|j|)}U^{k,j}_{\alpha,\beta}V^{j,i}_{\gamma,\delta}\,,\nonumber\\
	\Rightarrow W^{k,i}=&\sum_{j}\left(\openone\otimes Z_V\right)^{|k|+|j|}\left(U^{k,j}\otimes V^{j,i}\right)\,,
	\end{align}
	while the parity operator associated with $\mathcal{W}$ is $Z_{W}=Z_{U}\otimes Z_{V}$.
	Since $\mathcal{U}$ and $\mathcal{V}$ are degenerate, we can write
	\bea
	U^{k,j}&=&(\sigma^x)^{|k|+|j|}\otimes B^{k,j}\,,\\
	V^{j,i}&=&(\sigma^x)^{|j|+|i|}\otimes C^{j,i}\,,
	\eea
	and $Z_{U,V}=\sigma^z\otimes \openone$, so that, permuting the basis elements we have
	\begin{align}
	W^{k,i}=&\sum_j\left[(\sigma^x)^{|k|+|j|}\otimes (\sigma^z)^{|k|+|j|}(\sigma^x)^{|j|+|i|}\right]\nonumber\\
	\otimes &B^{k,j}\otimes C^{j,i}\,.
	\end{align}
In this basis, the parity operator is $Z_{W}=\sigma^z\otimes \sigma^z\otimes \openone\otimes \openone$. Define now the permutation operator acting non-trivially only  on the tensor product of the first two spaces,
$\Pi=\tilde{\Pi}\otimes\openone\otimes \openone$, where
\be
\tilde{\Pi}=
\begin{pmatrix}
1& 0& 0& 0\\
0 & 0&  0& 1\\
0 & 0&  1& 0\\
0 &1& 0& 0
\end{pmatrix}
\ee
We have $\tilde{\Pi}(\sigma^z\otimes \sigma^z  )\tilde{\Pi}^{-1}=\sigma^z\otimes \openone$\,, and also
\begin{align}
\tilde{\Pi} (\sigma^x)^{|k|+|j|}&\otimes (\sigma^z)^{|k|+|j|}(\sigma^x)^{|j|+|i|}\tilde{\Pi}^{-1}\nonumber\\
&=(\sigma^x)^{|k|+|i|}\otimes r_{|k|,|i|,|j|} \,,
\end{align}
where 
\begin{align}
r_{0,0,0}=&r_{1,1,1}=\openone\,,\quad r_{0,0,1}=r_{1,1,0}=y\,, \label{eq:r_form_1}\\
r_{0,1,0}=&r_{1,0,1}=\sigma^x\,,\quad r_{0,1,1}=r_{1,0,0}=\sigma^z\,,\label{eq:r_form_2}
\end{align}
and 
\be
y=\begin{pmatrix}
	0& 1\\
	-1 & 0  
\end{pmatrix}\,.
\ee
Accordingly, using $\Pi$ as a similarity transformation, we end up with a parity operator in the form~\eqref{eq:parity_operator} and
\begin{align}
W^{k,i}=&(\sigma^x)^{|k|+|i|}\otimes D^{k,i}\,,
\label{eq:form_W}
\end{align}
where
\begin{align}
D^{k,i}=\sum_{j} r_{|k|,|i|,|j|} \otimes  B^{k,j}\otimes C^{j,i}\,,
\label{eq:d_form}
\end{align}

	From Eqs.~\eqref{eq:form_W},~\eqref{eq:d_form},\eqref{eq:r_form_1} and \eqref{eq:r_form_2}, we see that the even subalgebra generated by $D^{k,i}$ commutes with the matrix $ y\otimes \openone\otimes \openone$, and is thus reducible. Accordingly, there must be a graded invariant subspace for the matrix $W^{k,i}$. In particular, we can write the corresponding projectors as
	\begin{align}
	P=\begin{pmatrix}
	\tilde{P} & 0 \\
	0 &  \tilde{Q}
	\end{pmatrix}\qquad 	Q=\openone-P=\begin{pmatrix}
	\tilde{Q} & 0 \\
	0 &  \tilde{P}
	\end{pmatrix}
	\end{align}
	where 
	\begin{align}
	\tilde{P}=\frac{(\openone-i y)}{2}\otimes \openone\,,\qquad 
	\tilde{Q}=\frac{(\openone+i y)}{2}\otimes \openone\,.
	\end{align}	
Now, $[P,Z_W]=0$, and $W^{k,i}P=PW^{k,i}P$, so that we can replace the tensors $W^{k,i}$ with $ W^{k,i}=P W^{k,i} P +Q W^{k,i} Q$, and the product $U^{(N)}V^{(N)}$ decomposes as the sum of two fMPOs. It is easy to see that these are exactly the same. This can be seen as follows. First we rewrite
	\begin{subequations}
		\begin{align}
		W^{k,i} & = \left(\begin{matrix} D^{k,i} & 0 \\ 0 & D^{k,i} \end{matrix}\right),\qquad  |i|+|k| = 0 \\
		W^{k,i} & = \left(\begin{matrix} 0 &  D^{k,i} \\  D^{k,i} & 0 \end{matrix}\right)\,, \qquad |i|+|k|= 1\,,
		\end{align}
		\label{eq:explicit_form_D}
	\end{subequations}
	so that
	\begin{subequations}
		\begin{align}
		PW^{k,i}P & = \left(\begin{matrix} \tilde{P}D^{k,i} \tilde{P} & 0 \\ 0 & \tilde{Q}D^{k,i}\tilde{Q} \end{matrix}\right),\qquad  |i|+|k| = 0 \\
		PW^{k,i}P & = \left(\begin{matrix} 0 &  \tilde{P}D^{k,i}\tilde{Q} \\  \tilde{Q}D^{k,i}\tilde{P} & 0 \end{matrix}\right)\,, \qquad |i|+|k|= 1\,,
		\end{align}
	\end{subequations}
	and
	\begin{subequations}
		\begin{align}
		QW^{k,i}Q & = \left(\begin{matrix}  \tilde{Q}D^{k,i}\tilde{Q}& 0 \\ 0 & \tilde{P}D^{k,i} \tilde{P} \end{matrix}\right),\qquad  |i|+|k| = 0 \\
		QW^{k,i}Q & = \left(\begin{matrix} 0 & \tilde{Q}D^{k,i}\tilde{P}  \\ \tilde{P}D^{k,i}\tilde{Q}  & 0 \end{matrix}\right)\,, \qquad |i|+|k|= 1\,.
		\end{align}
	\end{subequations}
It is now straightforward to see that 	$PW^{k,i}P=S 	QW^{k,i}Q S^{-1}$, where $S=\sigma^x\otimes \openone$, and $PZ_WP=-SQZ_WQS^{-1}$. Since the overall sign of the parity does not play any role, and recalling the overall factor $1/2$ coming from the product of the prefactors $1/\sqrt{2}$, we obtain that $U^{(N)}V^{(N)}$ can be represented as a type I fMPU generated by the (even) tensor $PW^{k,i}P$.
\end{proof}

Finally, we can now prove the main result of Sec.~\ref{sec:antiperiodic_fMPUs}, namely the equivalence between fQCA and fMPUs of the first and second kind.

\begin{prop}\label{eq:z_2_structure}
Up to appending inert ancillary fermionic degrees of freedom, any type I or type II fMPU with antiperiodic boundary conditions is a $1D$ fermionic QCA and viceversa .
\end{prop}
\begin{proof}
	
	Any simple tensor obviously  generates a locality-preserving fMPU, so we  only need to show that any fQCA can be represented as a type I or type II fMPU.  In fact, this is almost trivial recalling one of the main results derived in Ref.~\cite{fidkowski2019interacting}, stating that, after blocking and appending a finite number of fermionic ancillas, any $1D$ fQCA can be obtained by composing a finite number of the following  elementary operations: $(i)$ translations of fermionic modes [as defined in Eqs.~\eqref{eq:periodic_tr},~\eqref{eq:antiperiodic_tr}]; $(ii)$ translations of Majorana modes; $(iii)$ depth-two quantum circuits made of unitaries $u$, $v$ acting on neighboring sites (importantly, $u$, $v$ must have well-defined parity, which we can assume to be even wlog). In Sec.~\ref{sec:majorana_shift_operator}, we have already shown that Majorana-shift operators can be represented by type II fMPUs. Furthermore, it is straightforward to verify that translations of fermionic modes are implemented by type I fMPUs, with equal bond and physical dimensions and tensors $T^{i,j}_{\alpha,\beta}=\delta_{j\beta}\delta_{i\alpha}$. It is also simple to see that quantum circuits can be represented by type I fMPUs. This can be done by following the construction for qudits, as explicitly carried out, e.g., in Ref.~\cite{sahinoglu2018matrix}, and showing that one can always decompose the unitaries $u$ and $v$ in terms of even tensors. As a last step, one needs to show that an arbitrary product of type I and type II fMPUs can be represented as an fMPU of type I or type II. This follows from Prop.~\ref{eq:z_2_structure}, so that, putting all together, the statement is proven.
\end{proof}

\subsection{Index theory for fMPUs with ABC}
\label{sec:index_GCF}

Let $U^{(N)}$ be a type II fMPU in the standard form~\eqref{eq:standard_form_type_II}, and denote by $\tilde{\mathcal{U}}$, $\mathcal{M}$ the tensors in GCF associated with $\tilde{U}^{(N)}$ and $M^{(N)}_A$, respectively. Denote by $\mathcal{U}$ the tensor obtained by composing $\tilde{\mathcal{U}}$ and $\mathcal{M}$, and let $k$ be such that $\tilde{\mathcal{U}}_k$ is simple. Then, it is shown in Lemma~\ref{lemma:index_type_II} that the exponentiated index for $\mathcal{U}_{q}$ with $q\geq 2k$ is $\mathcal{I}_f=\sqrt{2}\tilde{\mathcal{I}}_f$, where $\tilde{\mathcal{I}}_f$ is the index of $\tilde{\mathcal{U}}_k$. Now, the tensor $\mathcal{U}_q$, obtained by composing $\tilde{\mathcal{U}}_q$ and $\mathcal{M}_q$, is not necessarily in GCF, so that one needs to make sure that the index of $\mathcal{U}_q$ coincides with that computed in the corresponding GCF. In the following, we show that this is true if we block $\tilde{q}$ times, with $\tilde{q}\geq 4k$.

\begin{lem}
Using the previous notations, the index of the tensor $\mathcal{U}_{\tilde{q}}$ is the same as the one computed in the corresponding GCF, where $\tilde{q}\geq 4k$. 
\end{lem}
\begin{proof}
	First, from the explicit graphical representation it is immediate to show that $\mathcal{W}=\mathcal{U}_{\tilde{q}}$ is simple, and that the transfer matrix reads $E_{\mathcal{W}}=|\rho_1)(\Phi_1|+|\rho_2)(\Phi_2|$, where $|\rho_1)$, $|\phi_1)$ are even, when seen as operators acting in the auxiliary space, and $|\rho_2)$, $|\phi_2)$ are odd. Let $P$, $Q$ be the orthogonal projectors onto the support of $\Phi_1$, $\rho_1^\ast$ respectively, where $\rho_1^\ast$ is the complex conjugate of $\rho_1$. We claim $\mathcal{W}P=\mathcal{W}$ (where matrix multiplication is  intended from right to left, as usual). To this end, we need to show that if $|v)$ is a state in the auxiliary space, such that $P|v)=0$ and $P^{\perp}|v)=|v)$, then 
	\be
	|w)=W^{n,m}_{\alpha,\beta} |\alpha)\ket{n}\bra{m}(\beta|v)=0\,,
	\ee
	where we denoted by $W^{n,m}_{\alpha,\beta}$ the matrix elements associated with $\mathcal{W}$.
	First, we note that since $\Phi_1$ is even, its support is a graded subspace, and $P$ is an even operator. Then, we can assume wlog that $v$ is even, and compute
	\be\label{eq:norm}
	(w|w)=
	\begin{tikzpicture}[baseline={([yshift=0.ex]current bounding box.center)}, scale=0.7]	
	\draw (1.25,1) arc (-270:-90:0.5);
	\draw (0.05,1) arc (90:-90:0.5);
	\draw [fill=white] (0.4,0.25) rectangle (0.9,0.75);
	\draw(0.7,-0.2)node{$E_\mathcal{W}$};
	\draw (-1.2+1.25,1) arc (-270:-90:0.5);
	\draw(1.8,0)node{$|v)$};
	\draw(1.8,1)node{$|v)$};
	\end{tikzpicture}
	=
	\begin{tikzpicture}[baseline={([yshift=0.ex]current bounding box.center)}, scale=0.7]	
	\draw (1.25+0.5,1) arc (-270:-90:0.5);
	\draw (0.05,1) arc (90:-90:0.5);
	\draw [fill=white] (0.2,0.2) rectangle (0.8,0.8);
	\draw (-1.2+1.25,1) arc (-270:-90:0.5);
	\draw [fill=white] (0.2+0.7,0.2) rectangle (0.8+0.7,0.8);
	\draw(1.8+0.5,0)node{$|v)$};
	\draw(1.8+0.5,1)node{$|v)$};
	\draw(+0.5,0.5)node{$\rho_2$};
	\draw(+1.2,0.5)node{$\Phi_2$};
	\end{tikzpicture}\,,
	\ee
	where we used that $\Phi_1|v)=0$ (since $P|v)=0$, and $P$ projects onto the support of $\Phi_1$). We see now that the rhs  of Eq.~\eqref{eq:norm} is zero, because it is proportional to the trace of $\rho_2$. Indeed, the latter is zero, since $\rho_2$ an odd operator. In the same way, one can see that $Q\mathcal{W}=\mathcal{W}$. Then, it is easy to show that it is also $\Phi_2 P=\Phi_2$, and $Q \rho_2=\rho_2$. Now, following the proof of Prop.~IV.5 in Ref.~\cite{cirac2017matrix}, we show  how to obtain the GCF of $\mathcal{W}$ using $P$ and $Q$. To this end, we use Jordan's Lemma, which guarantees a decomposition of the space $\mathbb{C}^D=(\bigoplus_i \mathbb{C}^2)\oplus
	\mathbb{C}^k$ such that in that basis, $P=\bigoplus_i |0\rangle\langle
	0|_i\oplus R$,  $Q=\bigoplus_i |v_i\rangle\langle v_i|\oplus S$, where $R$
	and $S$ are commuting projectors on $\mathbb{C}^k$. We can choose  $|0\rangle\langle
	0|_i$, $|v_i\rangle\langle v_i|$ , $R$, $S$ to be all even operators. Let us now define the (even) projector $\tilde{P}:=\bigoplus_i |0\rangle\langle 0|_i\oplus RS$. 
	We have the following properties:
	\begin{enumerate}[(i)]
		\item $P\tilde{P}=\tilde{P}$.
		\item $PQ=\tilde{P}Q$.
		\item $QP=Q\tilde{P}$.
		\item There exists an invertible $Y$ so that $\tilde{P}Q Y=\tilde{P}$.
	\end{enumerate}
	
We claim that $\tilde{\cal W}:=\tilde{P}{\cal W}\tilde{P}$ is the GCF of ${\cal W}$, when restricted to the range of $\tilde P$. First, using the properties above, it is straightforward to see that $\tilde{\cal W}$ and ${\cal W}$ define the same fMPU for all $N$. Next, we need to show that $\tilde{\cal W}$ is a degenerate graded normal tensor. To this end, we observe that the transfer operator of $\tilde{\cal W}$ is $X\mapsto {\rm tr}(\tilde{P}\Phi_1\tilde{P}X)\tilde{P}\rho_1\tilde{P}+{\rm tr}(\tilde{P}\Phi_2\tilde{P}X)\tilde{P}\rho_2\tilde{P}$. Using the properties of the projector $\tilde{P}$, it is immediate to  see that $\tilde{P}\rho_i\tilde{P}$, $\tilde{P}\Phi_i\tilde{P}$ are the left and right eigenvectors. Furthermore, they are even (odd) for $i=1$ ($i=2$), since $\tilde{P}$ is even, and using the properties of $\tilde{P}$ one can  easily see that  both $\tilde{P}\rho_1\tilde{P}$,  $\tilde{P}\Phi_1\tilde{P}$ are full rank in the range of $\tilde{P}$. Invoking Prop.~\ref{prop:fixed_point_rank} in Appendix~\ref{sec:technical_results_antiperiodic},  we have that $\tilde{\cal W}$ is a degenerate graded normal  tensor. 
	
Next, we show that $\tilde{\mathcal W}$ has the same left and right ranks as the tensor $\hat{\mathcal W}=\sqrt{\Phi_1}\mathcal W \sqrt{\rho_1^\ast}$, where again matrix multiplication is intended from right to left. For this, we take invertible matrices $X$, $Y$, $Z$ such that $X\sqrt{\Phi_1}=P$, $\sqrt{\rho^\ast}_1Z=Q$ (this can always be done: taking $Z^\prime$ invertible such that $Z^\prime \sqrt{\rho_1^\ast}=Q$, we have $Z=Z^{\prime\dagger}$) and $PQY=\tilde{P}$. Then, it is simple to show $Z\hat{\mathcal W}ZY=\tilde{\mathcal W}$, which  proves the claim.
	
Finally, we show that the rank of $\mathcal{U}_{\tilde{q}}$ is the same as $\hat{W}$. Using that for any operator ${\rm rank}(A)={\rm rank}(A^{\dagger}A)$, we have
\begin{equation}\label{eq:rank_sqrt_r}
\begin{tikzpicture}[baseline={([yshift=-.5ex]current bounding box.center)}, scale=0.7]	
\draw(-0.5,0.) node{rank};
\draw [black] (0.5,-0.5) to [round left paren] (0.5,0.5);
\draw [black] (3.2,-0.5) to [round right paren] (3.2,0.5);
\draw (0.8,0) -- (3.15,0);	
\draw (0.8+0.75,-0.6)  -- (0.8+0.75,0.6);
%
\draw [fill=white] (+0.8+0.45,-0.3) rectangle (+0.8+1.05,0.3);
\draw(+0.8+0.8,0.)node{$\mathcal{U}_{\tilde{q}}$};
\filldraw[fill=white,draw=black] (2.5,0) circle [radius=0.45];
\draw(+2.45,0.)node{\footnotesize $\sqrt{\rho_1}$};
\draw[dashed, thick, color=gray](0.5,0.15) -- (3.2,-0.15);
\end{tikzpicture}
=
\begin{tikzpicture}[baseline={([yshift=-1.5ex]current bounding box.center)}, scale=0.7]	
\draw(-0.5,0.25) node{rank};
\draw [black] (0.5,-0.5) to [round left paren] (0.5,1);
\draw [black] (2.7,-0.5) to [round right paren] (2.7,1);
\draw (0.8,0) -- (2.2,0);	
\draw (0.8,0.8) -- (2.2,0.8);	
\draw (0.8+0.75,-0.6)  -- (0.8+0.75,1.4);
\draw [fill=white] (+0.8+0.45,-0.3) rectangle (+0.8+1.05,0.3);
\draw [fill=black] (+0.8+0.45,0.8+-0.3) rectangle (+0.8+1.05,0.8+0.3);
\draw(+0.8+0.8,0.)node{$\mathcal{U}_{\tilde{q}}$};
\draw (2.15,0) arc (-90:90:0.4);
\filldraw[fill=white,draw=black] (2.35,0.4) circle [radius=0.25];
\draw(+2.35,0.4)node{\footnotesize $\rho_1$};
\draw[dashed, thick, color=gray](0.15,0.4) -- (3.2,0.4);	
\end{tikzpicture}
\end{equation}
where we separated input and output with a gray dotted line, and where we denoted by a black box the tensor $\bar{\mathcal{U}}_{\tilde{q}}$. Note that the cut determines the order of multiplication of the matrices involved. Next, define
\be
\mathcal{V}_{\tilde{q}}=
\begin{tikzpicture}[baseline={([yshift=-0.5ex]current bounding box.center)}, scale=0.7]	

\draw (0,-0.6) node[left]{} -- (0,1.25) node[left]{};
\draw (-0.6,0) node[left]{} -- (1,0) node[right]{};
\draw (-0.6,0.75) node[left]{} -- (1.5,0.75) node[right]{};
\filldraw[fill=white,draw=black] (0,-0.1) circle [radius=0.37];
\draw [fill=white] (-0.35,0.4) rectangle (0.35,1.1);
\draw(+0.,0.7)node{\footnotesize $\tilde{\mathcal{U}}_{\tilde{q}}$};
\draw(+0.,-0.1)node{\footnotesize $\mathcal{M}_{\tilde{q}}$};
\draw(+0.,-0.1)node{\footnotesize $\mathcal{M}_{\tilde{q}}$};
\filldraw[fill=white,draw=black] (1,0.75) circle [radius=0.35];
\draw(+0.95,0.75)node{\footnotesize $\sqrt{\tilde{\rho}}$};
\end{tikzpicture}\,,
\ee
where $\tilde{\rho}$ is the right eigenstate associated with eigenvalue $1$ of the transfer matrix $E_\mathcal{\tilde{U}}$. Since $\tilde{\mathcal{U}}$ is in GCF, the rank of $\tilde{\rho}$ is maximum, and clearly the left and right ranks for $\mathcal{V}_{\tilde{q}}$, $\mathcal{U}_{\tilde{q}}$ coincide. Finally, using once again that ${\rm rank}(A)={\rm rank}(A^{\dagger}A)$ we also have
\begin{equation}\label{eq:rank_plain}
\begin{tikzpicture}[baseline={([yshift=-.5ex]current bounding box.center)}, scale=0.7]	
\draw(-0.5,0.) node{rank};
\draw [black] (0.5,-0.5) to [round left paren] (0.5,0.5);
\draw [black] (3.,-0.5) to [round right paren] (3.,0.5);
\draw (0.8,0) -- (2.7,0);	
\draw (0.8+0.75,-0.6)  -- (0.8+0.75,0.6);
%
\draw [fill=white] (+0.8+0.45,-0.3) rectangle (+0.8+1.05,0.3);
\draw(+0.8+0.8,0.)node{$\mathcal{V}_{\tilde{q}}$};
\draw[dashed, thick, color=gray](0.5,0.15) -- (3.,-0.15);
\end{tikzpicture}
=
\begin{tikzpicture}[baseline={([yshift=-1.5ex]current bounding box.center)}, scale=0.7]	
\draw(-0.5,0.25) node{rank};
\draw [black] (0.5,-0.5) to [round left paren] (0.5,1);
\draw [black] (2.7,-0.5) to [round right paren] (2.7,1);
\draw (0.8,0) -- (2.2,0);	
\draw (0.8,0.8) -- (2.2,0.8);	
\draw (0.8+0.75,-0.6)  -- (0.8+0.75,1.4);
\draw [fill=white] (+0.8+0.45,-0.3) rectangle (+0.8+1.05,0.3);
\draw [fill=black] (+0.8+0.45,0.8+-0.3) rectangle (+0.8+1.05,0.8+0.3);
\draw(+0.8+0.8,0.)node{$\mathcal{V}_{\tilde{q}}$};
\draw (2.15,0) arc (-90:90:0.4);
\draw[dashed, thick, color=gray](0.15,0.4) -- (3.2,0.4);	
\end{tikzpicture}
\end{equation}

Now, it it straightforward to verify that  
\be
\rho_1=
\begin{tikzpicture}[baseline={([yshift=-0.5ex]current bounding box.center)}, scale=0.7]	
\draw (-0.55,0.4) node[left]{} -- (0.75,0.4) node[left]{};
\draw (-0.55,0.8) node[left]{} -- (0.75,0.8) node[left]{};
\draw (-0.55,1.2) node[left]{} -- (0.75,1.2) node[left]{};
\draw (-0.55,1.6) node[left]{} -- (0.75,1.6) node[left]{};
\draw [fill=white] (-0.3,0.) rectangle (0.5,2);
\draw (0.75,0.4) arc (-90:90:0.6);
\draw (0.75,0.8) arc (-90:90:0.2);
\filldraw[fill=white,draw=black] (1,1) circle [radius=0.25];
\draw(0.1,1)node{$E_{\mathcal{U}}$};
\draw(1,1)node{\footnotesize $\tilde{\rho}$};
\end{tikzpicture}
\ee
Here, we have used that the right even eigenvector associated with the transfer matrix $E_{\tilde{\mathcal{M}}}$ is simply the identity. We can now plug this expression into the rhs of Eq.~\eqref{eq:rank_sqrt_r}, and express $\mathcal{U}_{\tilde{q}}$ in terms of $\mathcal{U}$. Finally, recalling that $\mathcal{U}_k$ is simple, and making use of Eq.~\eqref{eq:zip} for the Majorana shift operator, after a straightforward calculation we obtain the rhs of Eq.~\eqref{eq:rank_plain}, so that the lhs of  Eq.~\eqref{eq:rank_sqrt_r} and \eqref{eq:rank_plain} also coincide. Hence, since the ranks for $\mathcal{V}_{\tilde{q}}$ and $\mathcal{U}_{\tilde{q}}$ are the same, we have just proven that multiplying on the right by $\sqrt{\rho_1}$ does not change the rank. In a  similar way, using that ${\rm rank}(A)={\rm rank}(AA^{\dagger})$, we can show that the rank does not change by multiplying the input auxiliary space by $\sqrt{\Phi_1^\ast}$. Finally, the same argument can be used for the NE-SW cut, thus completing the  proof. 

\end{proof}

\bibliography{./bibliography}

\end{document}